\documentclass{amsart}
\usepackage{amsaddr}
\usepackage[a4paper, total={5.8in, 9in}]{geometry}
\usepackage{mathtools}
\usepackage{amssymb}
\usepackage{tikz}
\usepackage{tikzsymbols}
\usepackage{array}
\usepackage{hhline}
\usepackage{bbm}
\usepackage{mathrsfs}
\usepackage{tabularx}
\usepackage{enumitem}
\usepackage[export]{adjustbox}
\usepackage[
unicode,
psdextra=false,     
hypertexnames=false 
]{hyperref}

\DeclareSymbolFont{mathscrLC}{OT1}{pzc}{m}{n} 
	
	\usetikzlibrary{calc, automata, chains, arrows.meta, quotes, backgrounds, positioning, shapes.misc, shapes.geometric}
	
	\tikzset{
		LabelStyle/.style = { rectangle, rounded corners, draw,
			minimum width = 2em, fill = yellow!50,
			text = red, font = \bfseries, fontsize = 6 }}

	\newtheorem{thm}{Theorem}
	
	\newtheorem{rem}{Remark}
	\newtheorem{prop}{Proposition}
	\newtheorem{lem}{Lemma}
	\newtheorem{cor}{Corollary}
	
	\newtheorem{alg}{Algorithm}
	
	\theoremstyle{definition}
	\newtheorem{exmp}{Example}
	\newtheorem{ass}{Assumption}

	\newcommand{\Z}{\mathbb{Z}}
	\newcommand{\R}{\mathbb{R}}
	\newcommand{\C}{\mathbb{C}}
	\newcommand{\ee}{\mathbbm{e}}
	
	\newcommand{\one}{\mathsf{1}}
	
	\newcommand{\One}{\mathbbm{1}}
	\newcommand{\diag}{\mathrm{diag}}
	
	\newcommand{\transpose}{\mathsf{T}}
	\renewcommand{\Re}{\mathrm{Re}}
	
	\newcommand{\dd}{\mathrm{d}}
	\newcommand{\set}[2]{\{\ #1 \ | \ #2 \ \}}

\title{Coexistence coalitions in propagule disperser quasi-communities}

\author[Aguirre, Capit\'an and Alonso]{%
	\NoCaseChange{Leonardo Aguirre$^1$,
	Jos\'e A. Capit\'an$^2$
	and David Alonso$^1$}}

\address{$^1$ Centre d'Estudis Avan\c{c}ats de Blanes,
	Consejo Superior de Investigaciones Cient\'ificas
	(CEAB-CSIC), Blanes, Catalonia}

\address{$^2$ Complex Systems Group and Department
	of Applied Mathematics, Universidad Polit\'ecnica
	de Madrid, Madrid, Spain}

\thanks{Leonardo Aguirre:
	\href{mailto:leonardo.aguirre.lab@gmail.com}
	{\nolinkurl{leonardo.aguirre.lab@gmail.com}}.}

\thanks{David Alonso:
	\href{mailto:dalonso@ceab.csic.es}
	{\nolinkurl{dalonso@ceab.csic.es}}.}

\begin{document}

\begin{abstract}
Many natural ecosystems harbor large numbers of coexisting species competing
for far fewer distinct resources, in apparent defiance of the competitive
exclusion principle. Various mechanisms have been proposed to explain this
apparent paradox, often pertaining to organisms with a two-stage
sessile--propagule life cycle. Here we develop a stochastic model class for
such propagule disperser communities that combines competition--colonization
trade-offs, spatial heterogeneity, demographic stochasticity, as well as
inherited trait variation, and recover several classical models as special or
limiting cases. Using bifurcation analysis, we classify equilibrium coalitions
near the extinction threshold and give sufficient conditions for their
realization by macroscopic equilibria away from the threshold, bypassing the
costly numerical computation of the actual equilibrium states. Illustrative
examples examine the resulting trait distributions and coalition patterns,
demonstrating the interactive effects of different coexistence mechanisms.
\end{abstract}

\maketitle

\section{Introduction}\label{sec:introduction}

The biodiversity paradox is one of the fundamental enigmas in theoretical
community ecology, referring to the tension between highly diverse natural
communities and mechanistic explanations provided by classical niche theory,
which seem to severely limit coexistence opportunities through the competitive
exclusion principle. A considerable portion of the debate has revolved around
propagule disperser species, whose life cycle consists of a sessile adult stage
and a mobile propagule stage, which occur, for instance, among intertidal
organisms such as barnacles \cite{Iwasa1986} and ascidians \cite{Edwards2010},
as well as in terrestrial plant communities such as forests \cite{Leck2008} and
grasslands \cite{Tilman1994}.

Numerous attempts have been made to explain high diversity in propagule
disperser communities using minimalist demonstration models \cite{Evans2013}
that incorporate one of several coexistence-enabling mechanisms. Prominent
approaches fall roughly into three categories: (i) nonlinear trade-offs between
competition and colonization traits, (ii) spatiotemporal heterogeneity of the
habitat, and (iii) ecological neutrality. Each of these approaches provides a
plausible explanation for coexistence in some cases but has been ruled out in
others. This ambiguous state of affairs has led to a large body of work debating
the relative merits of these explanations and questioning their generality or
predictive ability \cite{Vellend2010}.

Rather than taking sides, we start from the premise that real ecosystems might
feature several of these mechanisms at the same time and to varying degrees. It
would then be desirable to have an overarching framework that combines the
traditional single-mechanism models in order to quantify their explanatory merits
and to investigate non-trivial interactions. Here we develop such a framework as a
stochastic model class for propagule disperser communities that covers
competition--colonization trade-offs, environmental heterogeneity, as well as
ecologically neutral drift. We begin exploring coexistence conditions by
studying macroscopic equilibria, and prove sufficiency of certain spectral
conditions resulting in a relatively inexpensive algorithm for generating
coexistence phase diagrams.

The article is organized as follows: In Section \ref{sec:coexistence_mechanisms}
we briefly revisit several prominent coexistence models. Section
\ref{sec:model_class} introduces the general notion of propagule disperser
quasi-communities and formalizes it as a class of stochastic process models. We
derive the macroscopic continuum approximation and demonstrate how it yields
several classical coexistence models as special cases. Section
\ref{sec:coexistence_results} contains the main result of this article, namely
spectral coexistence predictions together with sufficient conditions for their
realization, and illustrates it with a few examples. The final Section
\ref{sec:discussion} discusses how this unified propagule disperser framework
could be used in subsequent research. Mathematical proofs are sketched in the
main text with details deferred to the respective appendices. Appendix
\ref{app:notations} contains an overview of mathematical notations and
terminology used throughout the article.

\section{Classical coexistence mechanisms}\label{sec:coexistence_mechanisms}

Many approaches to stationary coexistence in propagule
disperser models build on Levins' classical metapopulation model
\cite{Levins1969}, which considers a focal species in a (meta-)habitat
consisting of $S$ sites, each of which can either be occupied by the focal
species or be empty. The Levins model further posits a constant
rate\footnote{Symbols of the form $c_{\text{\_}}$ will generally denote rate
parameters with units of inverse time in the model under consideration. The
semantic correspondence between subscripts and ecological interpretation is not
meant to be consistent across different models.} $c_1$ at which an occupied site
becomes empty due to local extinction, and another rate $c_2$ at which the local
population of some occupied site colonizes an empty site. In the deterministic
mean-field description, the occupied-site count $x$ is governed by the ordinary
differential equation
\begin{equation}\label{eq:Levins}
	\dot{x} = - c_1 x + c_2 x \left(1-\frac{x}{S}\right) 
\end{equation}
Note that (\ref{eq:Levins}) can also be viewed as a model of a logistically
growing well-mixed population with growth rate $c_2$, carrying capacity $S$, and
additional mortality $c_1$. This ambiguity arises from the crucial assumption
that colonization is equally likely between any two sites, irrespective of their
spatial distance. An alternative local interpretation of (\ref{eq:Levins}) is
obtained by viewing the equation not as describing occupation across spatially
well-separated sites, but as describing abundance within a single macrosite
subdivided into many establishment microsites. Various later coexistence models
build on one of these interpretations by extending the model to multiple species
and introducing a particular coexistence-enabling mechanism. We now briefly
review three widely used approaches.

\subsection{Displacement competition}\label{subsec:displacement_competition}

A popular coexistence narrative for propagule dispersers invokes a
competition–colonization trade-off, i.e. that greater material investment in
each propagule increases its chance of establishing as a sessile individual,
while at the same time reducing the rate at which propagules can be produced and
hence the effective colonization rate. Levins and Culver \cite{Levins1971} used
this narrative to extend (\ref{eq:Levins}) to a two-species model in which the
second species is a better colonizer but a worse competitor in the sense of
displacement competition, meaning that propagules of the first species may
successfully colonize microsites already occupied by sessiles of the second
species, thereby replacing the resident individual. This model was extended to
any finite set $\mathfrak{I}$ of species by Tilman \cite{Tilman1994}, where
again better colonizers are always outcompeted by worse colonizers -- a
requirement that was slightly relaxed in a further extension by Calcagno
\cite{Calcagno2006}.\footnote{Related seed-size trade-offs include stress
	tolerance--fecundity models under heterogeneous recruitment conditions
	\cite{MullerLandau2010}, as well as evolutionary models of seed polymorphism
	\cite{Geritz1995}.} The latter approach leads to a system of ordinary
differential equations for state variables $x^i$ (written with upstairs index
$i\in\mathfrak{I}$ for consistency with later sections) representing the
fraction of establishment microsites occupied by species $i$:
\begin{equation}\label{eq:Calcagno}
	\dot{x}^i = - c_1 x^i + c_{2,i} x^i\Big(1-\sum_{i'} x^{i'}\Big) + x^i \sum_{i'} \left(c_{2,i} p^i_{i'} x^{i'} - p^{i'}_i c_{2,i'} x^{i'}\right)\ ,\quad i \in\mathfrak{I}
\end{equation}
Here again $c_1$ is the mortality rate, and $c_{2,i}$ is the rate at which
species $i$ attempts to colonize a uniformly chosen microsite, which always
succeeds if the microsite is empty and succeeds with probability $p^i_{i'}$ if
it is occupied by an individual of species $i'$. When colonization rates are
distinct, we may impose a total order on $\mathfrak{I}$ by $i<i' \Leftrightarrow
c_{2,i} < c_{2,i'}$. Calcagno then makes the additional assumption that
$p^i_{i'}$ is proportional to $(1 + \exp(\beta (c_{2,i} - c_{2,i'})))^{-1}$, a
sigmoid function of the difference in colonization attempt rates that smooths
the strict displacement ordering of Tilman's model. For unit proportionality
factor and $\beta\to\infty$, Tilman's strict ordering is recovered, namely
$p^i_{i'}=1$ for $i<i'$ and $p^i_{i'}=0$ for $i>i'$. In this case, the
triangular coefficient structure can be exploited to give an algebraic criterion
for colonization rates allowing coexistence at an equilibrium state, which then
turns out to be unique and stable. Such general criteria are harder to obtain in
Calcagno's model, but the examined cases are consistent with uniqueness and
stability of coexistence equilibria whenever they exist \cite{Calcagno2006}.

It should be noted that, in the multispecies models above, strict preemption
(i.e. no displacement competition) generically permits only single-species
equilibria whereas coexistence requires exactly balanced demographic rates
\cite{Yu2001}. Whether displacement competition is a plausible mechanism depends
on the system of interest. It appears to be at work in some intertidal organisms
such as ascidians \cite{Edwards2010} and in Tilman's original studies
\cite{Gleeson1990} on grassland species, where a competitive advantage through
displacement seems to originate from lower soil nutrient requirements, and a
trade-off with colonization ability could plausibly arise from differences in
metabolic allocation between roots and above-ground reproductive organs. On the
other hand, there are many highly diverse communities, such as tropical forests,
where displacement competition appears to be negligible. Nevertheless, there has
been a long tradition of seeking some form of competition--colonization
trade-off in forest communities, which supposedly manifests itself in an
observable trade-off between seed mass and seed production rate. A critical
evaluation of this rather contentious narrative has been conducted by Moles and
Leishman in \cite[Ch.10]{Leck2008} with the overall conclusion that seed mass
correlates with several effects at particular stages of the plant life cycle
whose combined impact on species fitness is far from obvious. This may suggest
that explicitly resolving additional life-cycle stages, such as propagules,
could help relate models more closely to field observations.

\subsection{Environmental heterogeneity}\label{subsec:environmental_heterogeneity}

Another basic coexistence narrative works via spatiotemporal heterogeneity of
the environment, where survival and colonization are moderated through local
habitat properties. A precursor of this approach is the seminal work of Hanski
and Ovaskainen \cite{Ovaskainen2001}, who extend (\ref{eq:Levins}) to a
``spatially realistic Levins model'' of a single-species metapopulation on a
finite set of sites $i\in\mathfrak{I}$:
\begin{equation}\label{eq:Hanski}
	\dot{x}^i = - c_{1,i} x^i + (1-x^{i})\sum_{i'} K^i_{i'}x^{i'}
\end{equation}
where $x^i$ is the expected value of site $i$ being occupied, the $c_{1,i}$ are
mortality rates and the $K^i_{i'}$ are effective colonization rates from $i'$ to
$i$. While (\ref{eq:Hanski}) only considers a single species and no explicit
propagule dynamics, there are extensions to propagule disperser communities. An
early example of this sort is the theory article \cite{Iwasa1986} by Iwasa and
Roughgarden, where the authors find that species are able to coexist through
environmental heterogeneity, with maximal biodiversity bounded by the number of
distinct macrosite types. A more recent contribution by Padmanabha et al.
\cite{Padmanabha2024} additionally takes into account an explicit propagule
migration network linking macrosites. Paraphrasing their model using the present
notation, we consider species $n\in\mathfrak{N}$ inhabiting a metahabitat
consisting of macrosites $b\in\mathfrak{B}$, each comprising several
establishment microsites. We then index sessile concentrations $x^i$ as well as
propagule concentrations $y^i$ with species-macrosite identifiers $i=(n,b)\in
\mathfrak{N}\times\mathfrak{B} = \mathfrak{I}$. The dynamics of these state
variables are formulated as a continuous-time Markov chain, and the authors
derive its macroscopic continuum approximation
\begin{align}
	\dot{x}^i &= -c_{1,i} x^i + \left(1-\sum_{B(i')=B(i)} x^{i'}\right) c_{5,i} y^i \label{eq:Padovians_x}\\
	\dot{y}^i &= -c_{2,i} y^i + \sum_{N(i')=N(i)} d^i_{i'} c_{3,i'} x^{i'} - \left(c_{4,i}y^i - \sum_{N(i')=N(i)} d^i_{i'}c_{4,i'}y^{i'}\right) - c_{5,i} y^i \label{eq:Padovians_y}
\end{align}
where $c_{1,i}$ and $c_{2,i}$ are mortality rates of sessiles and propagules
respectively\footnote{Padmanabha et al.~\cite{Padmanabha2024} actually set
	$c_{2,i}=0$.}, $c_{3,i}$ are propagule production rates, $c_{4,i}$ are rates at
which propagules attempt migration between macrosites and $d^i_{i'}$ is the
fraction of propagules leaving source $i'$ that reach destination $i$. Lastly,
$c_{5,i}$ are rates at which propagules attempt to establish on a microsite,
which always succeeds if the microsite is empty and kills the propagule
otherwise (strict preemption). The authors then perform a quasi-static
approximation by assuming the propagule dynamics to unfold on a faster time
scale and thus recover a multi-species version of (\ref{eq:Hanski}) with some
effective colonization kernel $(K^i_{i'})$. It is shown that stable coexistence
generically occurs over a broad parameter regime.

Since environmental heterogeneity is ubiquitous in nature, this modeling
approach appears rather attractive. However, it does not seem to provide a
universal explanation for biodiversity as there are examples of biodiverse
communities in reasonably homogeneous environments, as discussed in the
above-mentioned grassland studies.

\subsection{Ecological neutrality}\label{subsec:ecological_neutrality}

Yet another famous, albeit controversial, coexistence narrative is the neutral
theory of biodiversity, which posits that different species' trait profiles are
equivalent insofar as they contribute to fitness, resulting in a stochastic
Markov process of joint species abundances that unfolds purely through
``ecological drift''. The simplest formalization of this narrative is a Markov
process governed by the same birth, death, and speciation propensities across
all species \cite{Hubbell2011}. In an ecologically neutral model, the observed
biodiversity is not a property of some macroscopic equilibrium state but rather
of the process's (quasi-)stationary state distribution, representing typical
abundances with approximately stationary statistics over humanly observable time
scales. Traditionally, neutral theory has been most useful as a null model for
discerning statistically significant species differences. A fresh perspective
has been given by Clark et al. \cite{Clark2007}, who argue that stochastic
variation can conceal ecologically consequential species differences along
unobserved dimensions. Their argument motivates resolving additional variation
without assuming ecological equivalence. Here, in particular, the effect of
intraspecific trait variability seems somewhat underappreciated when it comes to
the classical theory of propagule disperser models \cite{Snell2019}. Indeed, it
has been empirically demonstrated that propagule mass, as a correlate of
establishment probability, can vary considerably within a species
\cite{Eriksson1999} and seems to possess a heritable component
\cite{Pelabon2021}. Accounting for imperfect trait-inheritance inside species
could lead to more nuanced models of ecological neutrality.

\section{Propagule disperser quasi-communities}\label{sec:model_class}

Having revisited three classical approaches for modeling stationary coexistence,
we now devise an overarching model class that combines their characteristic
ingredients. We consider a community of several propagule disperser species,
each occurring in different trait-types, which inhabit a metahabitat subdivided
into macrosites, each comprised of several establishment microsites. Sessiles
reproduce by non-faithful inheritance of their trait-type into propagules which
travel across macrosites and may at some point attempt to establish as new
sessiles occupying a microsite. Note that non-faithful inheritance of
trait-types effectively makes each species into a quasi-species in the sense of
\cite{Page2002}, and this term will henceforth be used along with
``quasi-community'' for a community of quasi-species. We now develop that
informal description into a stochastic process model.

\subsection{State space}

We classify individuals according to identifier triples  $(n,a,b)$ whose entries
are labels of quasi-species $n\in\mathfrak{N}$, trait-type $a\in\mathfrak{A}_n$,
and macrosite $b\in\mathfrak{B}$.\footnote{Fraktur letters such as
	$\mathfrak{N}, \mathfrak{A}, \mathfrak{B}$ etc. will generally refer to abstract
	index sets.} We then denote the set of all identifiers by $\mathfrak{I} =
\bigcup_{n\in\mathfrak{N}}\{n\}\times \mathfrak{A}_n \times \mathfrak{B}$ and
define the canonical projections $N((n,a,b)):= n$, $A((n,a,b)):= a$,
$B((n,a,b)):= b$. We introduce for each identifier $i\in\mathfrak{I}$ a state
variable $X^i$, written with upstairs index $i$, representing the number of
sessile individuals of quasi-species $N(i)$ and trait-type $A(i)$ that are
established at (some microsite inside) macrosite $B(i)$, as well as a state
variable $Y^i$ giving the analogous number of propagules present at macrosite
$B(i)$. We further denote by $S^b$ the number of microsites contained within
macrosite $b$.

In order to deal with state vectors and covectors, we set $\mathcal{V}:=
\R^\mathfrak{I}$ and write $\{\ee_{i}\}_{i\in\mathfrak{I}}$ for the standard
basis of $\mathcal{V}$ and $\{\ee^{i}\}^{i\in\mathfrak{I}}$ for its linear dual
basis in $\mathcal{V}^\ast = \R_\mathfrak{I}$.\footnote{See Appendix
	\ref{app:linear_algebra_notations} for these and other linear algebra
	notations.} We collect the state variables into sessile and propagule states
$X,Y\in\mathcal{V}$, respectively, and write the complete state as a block
column vector
$$\begin{pmatrix} X \\ Y \end{pmatrix} \in \mathcal{X}\times\mathcal{Y} \subset \mathcal{V}\times \mathcal{V}$$
with $\mathcal{X} = \set{X\in \mathcal{V}}{X \in \Z_{\geq 0}^\mathfrak{I},\ \forall b\in\mathfrak{B}: \sum_{B(i)=b} X^i \leq S^b}$ and $\mathcal{Y} = \set{Y\in \mathcal{V}}{Y \in \Z_{\geq 0}^\mathfrak{I}}$.

\subsection{Stochastic kinetics}

We set up the propagule disperser dynamics as a stochastic model, more
specifically a continuous-time Markov chain, on the state space
$\mathcal{X}\times\mathcal{Y}$. This is to say that the dynamics unfolds as a
stochastic process whose trajectories are produced as sequences of random state
jumps separated by exponentially distributed holding times. Each such jump event
represents a physical event involving sessiles and propagules such as death,
propagule dispersal, etc. We will refer to the distinct kinds of jump events as
\emph{kinetic events} $k\in\mathfrak{K}$ since they govern abundances $X^i,Y^i$
in a way reminiscent of stochastic chemical kinetics. In order to specify the
stochastic process, we need to prescribe for each kinetic event two items:
\begin{enumerate}
	\item \emph{Propensity $r_k:\mathcal{X}\times\mathcal{Y}\to \R_{\geq 0}$.} This
specifies the probability $r_k(X,Y)\dd t+o(\dd t)$ that event $k$ happens in
the next infinitesimal time interval of length $\dd t$, conditional on the
current state $(X,Y)^\transpose$. In this article, all propensity functions are
first-order mass-action kinetic, meaning that they correspond to physical
events performed spontaneously by single $i$-sessile or $i$-propagule
individuals. As a consequence, all propensities are of the form $r_k(X,Y) = c_k
X^i$ or $r_k(X,Y) = c_k Y^i$ where the \emph{propensity coefficient} $c_k$ is the
constant propensity that the event happens to a specific $i$-sessile or
$i$-propagule. 
	\item \emph{Random jump vector $V_k$.} This is a
$(\mathcal{V}\times\mathcal{V})$-valued random variable that prescribes how the
current state changes when event $k$ happens. More precisely, whenever $k$
happens, a realization of $V_k$ is drawn and added to the state vector. The
distribution of $V_k$ may also depend on the current state.
\end{enumerate}

The set $\mathfrak{K}$ of kinetic events is subdivided into five groups:
\begin{itemize}[leftmargin=*]
	\item \emph{Sessile mortality events} $k_{1,i}$: We posit propensity coefficients $c_{1,i}\equiv c_{k_{1,i}}> 0$ for the death of $i$-sessiles, leading to the propensity and the constant jump vector
	$$r_{k_{1,i}}(X,Y)=c_{1,i}X^i\quad \text{and}\quad V_{k_{1,i}} = \begin{pmatrix} -\ee_i \\ 0
	\end{pmatrix}$$
	\item \emph{Propagule mortality events} $k_{2,i}$: In the same way as for $i$-sessiles, we also posit propensity coefficients $c_{2,i}\equiv c_{k_{2,i}}\geq 0$ for the death of  $i$-propagules, leading to the propensity and the jump vector
	$$r_{k_{2,i}}(X,Y)=c_{2,i}Y^i\quad \text{and}\quad V_{k_{2,i}} = \begin{pmatrix} 0 \\ -\ee_i
	\end{pmatrix}$$
	
	\item \emph{Propagule dispersal events} $k_{3,i}$: We posit propensity coefficients $c_{3,i}\equiv c_{k_{3,i}} > 0$ for $i$-sessiles to release a propagule, leading to the propensity
	$$r_{k_{3,i}}(X,Y) = c_{3,i} X^i$$
	Each such propagule may be of a different trait-type $a'$ and be dispersed to a different macrosite $b'$. The probability for such a change to happen, i.e. for the propagule to be of type  $i'=(N(i),a',b')$, is posited as $0\leq d^{i'}_{i}\leq 1$. Defined in this way, the matrix $(d^{i'}_{i})$ has to be column-substochastic  (allowing column sums smaller than 1 allows for propagules to be lost from the modeled macrosites during dispersal). Following this prescription, we obtain a random jump vector with realizations
	$$V_{k_{3,i}} = \begin{pmatrix} 0 \\ \ee_{i'}\end{pmatrix} \text{ with probability } d^{i'}_{i} \text{ and } V_{k_{3,i}} = \begin{pmatrix} 0 \\ 0	\end{pmatrix} \text{ with probability } 1-\sum_{i'}d^{i'}_{i}$$
	
	\item \emph{Propagule migration events} $k_{4,i}$: Beyond the initial dispersal, propagules may also migrate to a different macrosite later on, either actively or by being carried there. We posit propensity coefficients $c_{4,i}\equiv c_{k_{4,i}}\geq 0$ for $i$-propagules to embark on such a migration, leading to
	$$r_{k_{4,i}}(X,Y) = c_{4,i} Y^i$$
	and probabilities $0\leq \tilde{d}^{i'}_{i} \leq 1$ with which an $i$-propagule reaches macrosite $b'$, i.e. it becomes a propagule of type $i'=(N(i),A(i),b')$. Again, we note that $(\tilde{d}^{i'}_{i})$ is column-substochastic and obtain a random jump vector with realizations
	$$V_{k_{4,i}} = \begin{pmatrix} 0 \\ \ee_{i'} - \ee_i	\end{pmatrix} \text{ with probability } \tilde{d}^{i'}_{i} \text{ and } V_{k_{4,i}} = \begin{pmatrix} 0 \\ -\ee_i	\end{pmatrix} \text{ with probability } 1-\sum_{i'}\tilde{d}^{i'}_{i}$$
	
	\item \emph{Establishment events} $k_{5,i}$: Finally, we
	 posit propensity coefficients $c_{5,i}\equiv c_{k_{5,i}} > 0$ with which $i$-propagules attempt to establish at some randomly chosen microsite inside macrosite $B(i)$, leading to 
	 $$r_{k_{5,i}}(X,Y) = c_{5,i} Y^i$$
	 If the microsite is free, which it is with probability $1 - \sum_{B(i')=B(i)}\tfrac{X^{i'}}{S^{B(i)}}$, the propagule succeeds in establishing with some posited probability $p^i_0 > 0$. If the microsite is occupied by a sessile of type $i'=(n',a',B(i))$, which it is with probability $\tfrac{X^{i'}}{S^{B(i)}}$, the propagule succeeds with probability $p^i_{i'}\leq p^i_0$. The resulting probability distribution of the random jump vector $V_{k_{5,i}}$ is given in Table \ref{tbl:mechanisms}.
\end{itemize}

We make another assumption on dispersal/trait-inheritance probabilities $d_i^{i'}$ and
migration probabilities $\tilde{d}_i^{i'}$, namely that within each
quasi-species any trait-type/macrosite combination can be reached from any other
by a suitable sequence of dispersal and migration events, possibly stretching
over several generations:

\begin{ass}\label{ass:block_irreducibility}
	Consider the directed graph on vertex set $\mathfrak{I}$ which has an edge $i\to i'$ precisely if $N(i)=N(i')$ and there exists a sequence $i_1,\ldots,i_L = i'\in\mathfrak{I}$ with $d_i^{i_1} > 0$ and $c_{4,i_l}\tilde{d}_{i_{l}}^{i_{l+1}} > 0$ for any $l=1,\ldots,L-1$. We assume that, for every species $n\in\mathfrak{N}$, the subgraph induced by  $\set{i\in\mathfrak{I}}{N(i)=n}$ is strongly connected.
\end{ass}

\begin{table}[tbp]
	\begin{tabularx}{\linewidth}{|l|l|l|X|}
	\hline
	$k$ & parameters & $r_k$ & $V_k$ \\
	\hline\hline
	$k_{1,i}$ & $c_{1,i}\in (0,\infty)$ & $c_{1,i}X^i$ & $V_k = \begin{pmatrix} -\ee_i \\ 0 \end{pmatrix}$\\
	\hline
	$k_{2,i}$ & $c_{2,i}\in [0,\infty)$ & $c_{2,i}Y^i$ & $V_k = \begin{pmatrix} 0 \\ -\ee_i \end{pmatrix}$\\
	\hline
	$k_{3,i}$ & $\begin{aligned}
		&c_{3,i}\in (0,\infty)\\
		&(d^{i'}_i)\in [0,1]^\mathfrak{I},\\
		&d^{i'}_i = 0 \text{ for } N(i')\neq N(i)\\
		&\sum_{i'} d^{i'}_i \leq 1 
	\end{aligned}$ & $c_{3,i}X^i$ & $\begin{aligned}
		&\Pr\left[ V_k = \begin{pmatrix} 0 \\ \ee_{i'} \end{pmatrix} \right] = d^{i'}_i\\
		& \Pr\left[ V_k = \begin{pmatrix} 0 \\ 0 \end{pmatrix} \right] = 1 - \sum_{i'} d^{i'}_i
	\end{aligned}$\\
	\hline
	$k_{4,i}$ & $\begin{aligned}
		&c_{4,i}\in [0,\infty)\\
		&(\tilde{d}^{i'}_i)\in [0,1]^\mathfrak{I},\\
		&\tilde{d}^{i'}_i = 0 \text{ if } N(i')\neq N(i)\\
		&\hspace{1.25cm}\text{or }A(i')\neq A(i)\\
		&\sum_{i'} \tilde{d}^{i'}_i \leq 1 
	\end{aligned}$ & $c_{4,i}Y^i$ & $\begin{aligned}
		&\Pr\left[ V_k = \begin{pmatrix} 0 \\ \ee_{i'} - \ee_i \end{pmatrix} \right] = \tilde{d}^{i'}_i\\
		& \Pr\left[ V_k = \begin{pmatrix} 0 \\ -\ee_i \end{pmatrix} \right] = 1 - \sum_{i'} \tilde{d}^{i'}_i
	\end{aligned}$\\
	\hline
	$k_{5,i}$ & $\begin{aligned}
		&c_{5,i}\in (0,\infty)\\
		&p^i_0 \in (0,1]\\
		&(p^i_{i'}) \in [0,1]_\mathfrak{I}\\
		&p^i_{i'}\leq p^i_0\\
		&p^i_{i'} = 0 \text{ for } B(i')\neq B(i)
	\end{aligned}$ & $c_{5,i}Y^i$ & $\begin{aligned}
		&\Pr\left[ V_k = \begin{pmatrix} \ee_i-\ee_{i'} \\ -\ee_{i} \end{pmatrix} \right] = p^i_{i'}\tfrac{X^{i'}}{S^{B(i)}}\ ,\quad i'\neq i\\
		&\Pr\left[ V_k = \begin{pmatrix} \ee_i \\ -\ee_i \end{pmatrix} \right] = p^i_{0}\left(1-\sum_{B(i')=B(i)}\tfrac{X^{i'}}{S^{B(i)}}\right)\\
		&\Pr\left[ V_k = \begin{pmatrix} 0 \\ -\ee_{i} \end{pmatrix} \right] = 1 - \sum_{i'\neq i} p^i_{i'}\tfrac{X^{i'}}{S^{B(i)}} \\ 
		&\hspace{3.5cm} -p^i_0\left(1-\sum_{B(i')=B(i)}\tfrac{X^{i'}}{S^{B(i)}} \right)
	\end{aligned}$\\
	\hline
\end{tabularx}

	\caption{Overview of kinetic events in the stochastic model along with parameter domains (Assumption \ref{ass:block_irreducibility} is omitted from this table).} \label{tbl:mechanisms}
\end{table}

\subsection{Macroscopic approximation}

To investigate the large-abundance behavior, consider a sequence of stochastic
models for which the reference size satisfies $S_0\to\infty$ and
$\frac{S^b}{S_0}\to\bar S^b\in(0,\infty)$. Writing $x^i=X^i/S_0$ and
$y^i=Y^i/S_0$, the density-dependent Markov-chain law of large numbers
\cite{Kurtz1970} gives a deterministic limit with \emph{macroscopic vector
field} $\bar v$ on $\mathcal V\times\mathcal V$, meaning that the limiting
concentrations solve
\begin{equation}\label{eq:macroscopic_equation}
	{\textstyle\begin{pmatrix}	x \\ y \end{pmatrix}}^{\boldsymbol{\cdot}} = \bar{v}(x,y),
\end{equation}

\begin{prop}[macroscopic vector field]\label{prop:macroscopic_vectorfield}
	The macroscopic vector field obtained from the kinetic events $(r_k, V_k), k\in\mathfrak{K}$, takes the form
	\begin{equation}\label{eq:macroscopic_vectorfield}
		\bar{v}(x,y) = \begin{pmatrix}
			-\left(C_1 + \diag(P^\transpose \bar{S}^{-1} C_5 y)\right) & \left(P_0\bar{S} - \diag((P_0 E - P)x)\right)\bar{S}^{-1}C_5\\
			DC_3 & - (C_2 + (\One-\tilde{D})C_4 + C_5)
		\end{pmatrix}\begin{pmatrix} x \\ y \end{pmatrix}
	\end{equation}
		where $C_1,\ldots, C_5$ are diagonal matrices containing the propensity coefficients, $D$ and $\tilde{D}$ are matrices containing the dispersal/trait-inheritance and migration probabilities,  $P_0$ is a diagonal matrix containing the basal establishment probabilities, $P$ is a matrix containing the displacement probabilities, $\bar{S}$ is a diagonal matrix containing the relative macrosite sizes and $E$ is a matrix whose entries $E^i_{i'}$ are equal to $1$ whenever $B(i)=B(i')$ and $0$ otherwise.\footnote{Formal definitions of these parameter matrices can be found in Appendix \ref{app:macroscopic_vf}.}
	If the scaled initial states converge in probability to $(x(0),y(0))^\transpose$ and the solution of (\ref{eq:macroscopic_equation}) exists on $[0,T]$, then the scaled stochastic process converges to it uniformly in probability on $[0,T]$.
\end{prop}
\begin{proof}[Proof (sketch, details in Appendix \ref{app:macroscopic_vf})]
The rescaled transition rates and mean jump vectors lead to the asymptotic mean
dynamics (\ref{eq:macroscopic_vectorfield}). The finite-time convergence is the
standard density-dependent Markov-chain law of large numbers; the required
regularity and boundedness are verified in the appendix.
\end{proof}

In the stochastic model, event propensities and jump vectors are such that the
dynamics naturally preserves the state space $\mathcal{X}\times\mathcal{Y}$.
This property carries over to the macroscopic dynamics and the \emph{macroscopic
state space} $\bar{\mathcal{X}}\times\bar{\mathcal{Y}}$ with $\bar{\mathcal{X}}
:= \set{x\in\mathcal{V}}{x\geq 0,\ \ee^i\cdot E x \leq \bar{S}^{B(i)} }$ and 
$\bar{\mathcal{Y}} = \set{y \in \mathcal{V}}{y \geq 0}$, namely
\begin{cor}
	The flow of the macroscopic vector field $\bar{v}$ preserves $\bar{\mathcal{X}}\times\bar{\mathcal{Y}}$.
\end{cor}
\begin{proof}
On the boundary of $\bar{\mathcal X}\times\bar{\mathcal Y}$, a zero sessile
coordinate has only non-negative establishment inflow, and a zero propagule
coordinate has only non-negative dispersal and migration inflows. On a capacity
face of $\bar{\mathcal X}\times\bar{\mathcal Y}$, displacement gains and losses
cancel and vacant-site establishment vanishes, giving $(\ee^i E
x)^{\boldsymbol{\cdot}}=-\ee^i E C_1x\leq0$. Thus the vector field points into
the physical state space on every boundary face.
\end{proof}

As we will be mostly interested in equilibrium coexistence, we employ 
Proposition \ref{prop:macroscopic_vectorfield} in order to simplify the
equilibrium condition $\bar{v}(x,y) = 0$ as follows.
\begin{prop}[equilibrium condition]\label{prop:stationary}
	The solutions of $\bar{v}(x,y) = 0$ are determined by
	\begin{gather}
		\left(- C_1 - \diag(P^\transpose \bar{S}^{-1}Kx) + (P_0 - \bar{S}^{-1}\diag((P_0 E - P)x))K\right) x = 0 \label{eq:stationary_x}\\
		y = C_5^{-1} K x \label{eq:stationary_y}
	\end{gather}
	with the non-negative matrix $K = C_5 (C_2 + (\One-\tilde{D})C_4 + C_5)^{-1} D C_3$.
\end{prop}
\begin{proof}[Proof (sketch, details in Appendix \ref{app:equilibrium_condition})]
We exploit the fact that the propagule row in
(\ref{eq:macroscopic_vectorfield}) is linear in $x$ and $y$, leading to
(\ref{eq:stationary_y}) and subsequently also (\ref{eq:stationary_x}).
Furthermore, as detailed in the appendix, $K$ is non-negative as the product of
non-negative matrices and the inverse of a non-singular M-matrix.
\end{proof}

Notably, only the sessile equation (\ref{eq:stationary_x}) is a non-trivial
condition on the existence of a non-zero equilibrium state. A similar approach
can be taken to derive an approximation of the  non-stationary macroscopic
dynamics (\ref{eq:macroscopic_vectorfield}) in the quasi-static regime where the
propagule dynamics is much faster than the sessile dynamics --- a plausible
assumption in certain ecological scenarios. Specifically, the following limit is
an accurate approximation of the sessile dynamics if the propensity coefficients
$c_{2,i}, c_{3,i}, c_{4,i}$ are much larger than the $c_{1,i}$, $c_{5,i}$.

\begin{prop}[quasi-static limit]\label{prop:quasi-static}
	Let all $c_{2,i}>0$ and, for $\kappa\in(0,1]$, consider (\ref{eq:macroscopic_vectorfield}) with propagule propensity coefficients $c_{m,i}$, $m=2,3,4$, replaced by rescaled coefficients $\tfrac1\kappa c_{m,i}$. Furthermore consider a family of initial conditions $(x_\kappa(0),y_\kappa(0))\in\bar{\mathcal X}\times\bar{\mathcal Y}$ such that, for $\kappa\to 0$, $x_\kappa(0)$ converges to some point $x(0)\in\bar{\mathcal{X}}$ and $y_\kappa(0)$ remains bounded in $\bar{\mathcal{Y}}$. Then, the sessile part of a solution to (\ref{eq:macroscopic_vectorfield}) converges for $\kappa\to 0$ uniformly on any bounded time interval to the corresponding solution of
	\begin{equation}
		\dot{x} = \left(- C_1 - \diag(P^\transpose \bar{S}^{-1}\tilde{K}x) + (P_0 - \bar{S}^{-1}\diag((P_0E - P)x))\tilde{K}\right) x \label{eq:quasi-static}
	\end{equation}
	with the non-negative matrix $\tilde{K} = C_5(C_2 + (\One-\tilde{D})C_4)^{-1} D C_3$.
\end{prop}
\begin{proof}[Proof (sketch, details in Appendix \ref{app:quasi-static})]
After substituting the rescaled propensity coefficients into the macroscopic
vector field, it can be written as two coupled equations for sessile dynamics
$\dot{x}$ and propagule dynamics $\dot{y}$. The latter can be brought into the
form of a singularly perturbed differential equation, namely $\kappa\dot{y} =
F(y; x,\kappa)$, where $F$ is affine-linear in $y$. Tikhonov's singular
perturbation theorem then asserts that, in the limit $\kappa\to 0$, the sessile
part of a solution to the coupled $(\dot{x},\dot{y})$-system converges
uniformly on bounded time intervals to the solution of the sessile equation
coupled to the algebraic constraint $F(y;x,0) = 0$, given that the latter
determines a globally exponentially asymptotically stable solution to $\dot{y}
= F(y;x,0)$ for any $x$. This premise is verified in the appendix. Solving the
algebraic constraint for $y$ and substituting this into the sessile equation
yields (\ref{eq:quasi-static}).
\end{proof}
\subsection{Classical cases}\label{subsec:special_case_equations}
We demonstrate that the classical models revisited in Section \ref{sec:coexistence_mechanisms} can be obtained as special cases of the propagule disperser model class.

\begin{itemize}[leftmargin=*, topsep=4pt, itemsep=4pt]
	\item \textbf{Displacement competition:} Since the displacement competition models from Section \ref{subsec:displacement_competition} do not model propagules explicitly but instead regard their dynamics as instantaneous, we apply the quasi-static approximation (\ref{eq:quasi-static}). Furthermore, those models only consider a single trait-type per species and only a single macrosite of unit size, i.e. $|\mathfrak A_n|=1$, $|\mathfrak B|=1$, $D = \One = \tilde{D}$, $\bar{S}=\One$. Assuming that all propagules have non-zero mortality and that attempted colonization on empty microsites always succeeds, we get (\ref{eq:quasi-static}) with $P_0=\One$, $E$ being a matrix of ones, and the diagonal matrix $\tilde{K} = C_5 C_2^{-1} C_3$. Taking the $i$-th component and rearranging terms yields
	$$\dot{x}^i = -c_{1,i}x^i + \tilde{K}^i_i x^i (1-\sum_{i'}x^{i'}) + x^i\sum_{i'} \left(\tilde{K}^i_i p^i_{i'} x^{i'} - p^{i'}_i \tilde{K}^{i'}_{i'} x^{i'}\right)$$
	which is evidently of the same form as (\ref{eq:Calcagno}) with specific mortality rates $c_{1,i}$, colonization rates $\tilde{K}^i_i$ and displacement probabilities $p^i_{i'}$.
	
	\item \textbf{Environmental heterogeneity:} Since the models from Section \ref{subsec:environmental_heterogeneity} consider just a single trait-type per species, the matrices $D,\tilde{D}$ merely encode propagule dispersal and migration probabilities among macrosites. While these two sets of probabilities might generally come from different ecological mechanisms, the model (\ref{eq:Padovians_x},\ref{eq:Padovians_y}) restricts attention to $D=\tilde{D}$ as the source article \cite{Padmanabha2024} is framed around ``explorer'' propagules without a distinct initial dispersal mechanism. Furthermore, all macrosites have unit size, i.e. $\bar{S}=\One$, there is no displacement competition, i.e. $P=0$, and any propagule that attempts to establish in an empty microsite also succeeds, i.e $P_0=\One$. This simplifies the macroscopic dynamics (\ref{eq:macroscopic_vectorfield}) to
		\begin{equation}\label{eq:spatial_heterogeneity_MFA}
		\begin{pmatrix}	x \\ y \end{pmatrix}^{\boldsymbol{\cdot}} = \begin{pmatrix}
			-C_1 & (\One - \diag(Ex)) C_5\\
			\tilde{D}C_3 & - (C_2 + (\One-\tilde{D})C_4 + C_5)
		\end{pmatrix}\begin{pmatrix} x \\ y \end{pmatrix}
	\end{equation}
		Taking the $i$-th sessile and propagule entries leads to
	\begin{align*}
		\dot{x}^i &= -c_{1,i} x^i + \left(1-\sum_{B(i')=B(i)} x^{i'}\right) c_{5,i} y^i \\
		\dot{y}^i &= -c_{2,i} y^i + \sum_{N(i')=N(i)} \tilde{d}^i_{i'} c_{3,i'} x^{i'} - \left(c_{4,i}y^i - \sum_{N(i')=N(i)} \tilde{d}^i_{i'}c_{4,i'}y^{i'}\right) - c_{5,i} y^i
	\end{align*}
	which is precisely (\ref{eq:Padovians_x},\ref{eq:Padovians_y}).
	
	\item \textbf{Trait-drift:} In order to focus on intraspecific trait variability, we specialize the propagule disperser model to a single unit-sized macrosite, $|\mathfrak{B}|=1$, $\bar S=\One$, and choose $\tilde D=0$, so that $c_{4,i}$ represents loss through emigration from the modeled macrosite, while $D$ models trait inheritance. We also assume there is no displacement competition, i.e. $P=0$, and may thus simplify (\ref{eq:macroscopic_vectorfield}) to
	\begin{equation}\label{eq:trait-drift_equation}
		\begin{pmatrix}	x \\ y \end{pmatrix}^{\boldsymbol{\cdot}} = \begin{pmatrix}
			-C_1 & (1 - \one^\transpose x)P_0 C_5\\
			DC_3 & - (C_2 + C_4 + C_5)
		\end{pmatrix}\begin{pmatrix} x \\ y \end{pmatrix}
	\end{equation}
	where $\one\in\mathcal{V}$ denotes a vector of ones. We will call (\ref{eq:trait-drift_equation}) the \emph{trait-drift model}. Note that the effects of imperfect trait-type inheritance remain in the macroscopic
	approximation unlike the classical demographic drift fluctuations discussed in the neutral theory of biodiversity.
\end{itemize}

\section{Equilibrium coexistence}\label{sec:coexistence_results}

In this section we turn to equilibria of the macroscopic propagule disperser
dynamics (\ref{eq:macroscopic_equation},\ref{eq:macroscopic_vectorfield}) and
investigate whether there are general conditions for coexistence, i.e. for
positive concentrations of more than one quasi-species at equilibrium.
Conditions of this sort have been given in some classical special cases such as
Tilman's displacement competition model \cite{Tilman1994}. As the main result of
this article, we use spectral Perron data from the linearization at the zero equilibrium
together with establishment parameters to classify admissible
coalitions and give sufficient conditions under which each such coalition is
represented by an equilibrium --- possibly unstable. Importantly, the
determination of coalitions does not require numerical computation of the
respective equilibrium states.

\subsection{The trait-drift model} We first analyze the trait-drift
specialization (\ref{eq:trait-drift_equation}) and derive its spectral
equilibrium conditions. The macroscopic equilibrium equation is obtained by
specializing (\ref{eq:stationary_x}), which can be brought into the form
\begin{equation}\label{eq:trait-drift_EV_equation}
	C_1^{-1}P_0 K x = (1-\one^\transpose x)^{-1} x
\end{equation}
Since all matrices entering into $C_1^{-1}P_0K$ are $\mathfrak N$-block-diagonal,
so is $C_1^{-1}P_0K$ itself, with every principal $n$-block corresponding to a
quasi-species.\footnote{See Appendix \ref{app:block_notation} for this
terminology and the following block vector/matrix notations.} In Appendix
\ref{app:trait-drift_spectrum} we invoke the Perron--Frobenius theorem to show
that, for any $n\in\mathfrak{N}$, the spectral radius $r_n$ of $[C_1^{-1}P_0
K]^n_n$ is non-negative, real, and simple, and its one-dimensional Eigenspace is
spanned by an Eigenvector $\rho_n$ whose $n$-block entries can be chosen
positive with unit sum, while all its other entries are zero.
Now let $\mathcal{R}:=\set{r\in\R}{r>1,\ \exists n:\ r = r_n}$ and, for any $r\in\mathcal{R}$, define the simplex $\bar{\mathcal{X}}_{0,r} \subset \bar{\mathcal{X}}$ as the convex hull of $\set{(1-\tfrac{1}{r})\rho_n}{n\in\mathfrak{N},\ r_n = r}$, and the corresponding full-state simplex as $$\bar{\mathcal{M}}_{0,r} := \set{\left(x, C_5^{-1}K x\right)\in \bar{\mathcal{X}}\times\bar{\mathcal{Y}}}{x \in \bar{\mathcal{X}}_{0,r}}$$

\begin{prop}\label{prop:trait-drift}
	If $\mathcal{R} = \varnothing$, the only non-negative equilibrium of (\ref{eq:trait-drift_equation}) is the zero equilibrium and it is stable. Otherwise the zero equilibrium is unstable and, for every $r\in\mathcal{R}$, the simplex $\bar{\mathcal{M}}_{0,r}$ consists of non-zero non-negative equilibria. All equilibria in  $\bar{\mathcal{M}}_{0,r}$, $r<\max(\mathcal{R})$, are unstable.
\end{prop}
\begin{proof}[Proof (sketch, details in Appendix
	\ref{app:trait-drift_proof})] For every occupied $n$-block, a non-zero
non-negative solution $x$ of (\ref{eq:trait-drift_EV_equation}) is a positive
multiple of the block Perron Eigenvector $\rho_n$. All occupied blocks must
therefore have the same block Perron Eigenvalue $r_n=r$. Moreover, the
magnitude of $x$ is determined by $\one^\transpose x=1-\frac1r$. A non-zero
non-negative solution therefore exists precisely when $r>1$, and the set of all
such solutions associated with $r$ is exactly $\bar{\mathcal{X}}_{0,r}$. 
Together with $y=C_5^{-1}Kx$, this gives the equilibrium simplex
$\bar{\mathcal{M}}_{0,r}$. For stability of the zero equilibrium, one observes
that $\dd\bar v(0,0)$ is $\mathfrak N$-block-diagonal with Metzler principal
blocks. Each principal block has a negative, zero, or positive dominant
Eigenvalue according as $r_n<1$, $r_n=1$, or $r_n>1$. Stability of the zero
equilibrium relative to the macroscopic state space $\bar{\mathcal
	X}\times\bar{\mathcal Y}$ is established using the positive linear semigroup
$\exp(t\dd\bar v(0,0))$ and comparison with the nonlinear dynamics. Finally,
for a non-zero equilibrium in $\bar{\mathcal{M}}_{0,r}$ with
$r<\max(\mathcal{R})$, one can explicitly construct a positive Eigenvalue of
its Jacobian by combining components in the occupied blocks (associated to the
block-dominant Eigenvalue $r$) with a Perron--Frobenius Eigenvector in a block
whose dominant Eigenvalue is larger than $r$.  The equilibrium is therefore
unstable.
\end{proof}

The above proposition conforms to a common intuition under single-resource
limitation, namely that a resident species $n$ can be invaded if the amount of
resource left unused at resident equilibrium, here the amount of unoccupied
space $1/r_n$, allows positive growth of the invader. One might then suspect
that such invasion leads to a stable equilibrium, in
generalization of some classical special cases like  \cite{Tilman1994}. This is
false for general trait-drift models as demonstrated by the following example.
Here, the invasion criterion by itself does not determine the subsequent
attractor due to non-trivial dynamics among the dominant quasi-species'
trait-types.
\begin{exmp}\label{exmp:trait-drift_counterexample}
	Consider an environment consisting of a single macrosite $\mathfrak{B} = \{1\}$ of unit size inhabited by a single quasi-species, $\mathfrak{N} = \{1\}$, with two trait-types $\mathfrak{A}_1 = \{1,2\}$. The key to achieving instability is to choose very uneven mortalities, contrary to Proposition \ref{prop:trait-drift_special}, along with some other more technical parameter choices detailed in Appendix \ref{app:trait-drift_counterexample}. This yields $r=r_1 = 4$ and a very uneven non-zero non-negative equilibrium at
	\begin{equation*}
		x = \begin{pmatrix} 0.7425 \\ 0.0075 \end{pmatrix}\quad \text{and}\quad
		y = \begin{pmatrix} 5.94 \\ 0.2 \end{pmatrix}
	\end{equation*}
	The Jacobian $\dd\bar v(x,y)$ has a pair of complex conjugate Eigenvalues with positive real parts, approximately $0.212\pm1.620\,\mathrm i$, so the equilibrium is unstable. Numerical integration displays sustained oscillations around it, as shown in Fig.~\ref{fig:trait-drift_counterexample}.
\end{exmp}
\begin{figure}[tbp]
	\centering
	\includegraphics[width=\textwidth]{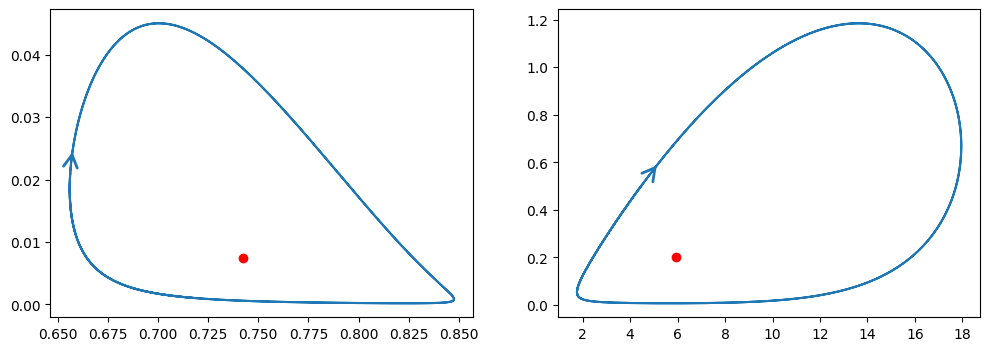}
	\caption{Projection of a numerical trajectory from Example \ref{exmp:trait-drift_counterexample}. The left panel shows $(x^1,x^2)$ and the right panel $(y^1,y^2)$. Arrow tips indicate the direction of motion and the red dot marks the unstable equilibrium.}\label{fig:trait-drift_counterexample}
\end{figure}

This example establishes instability of the positive equilibrium and suggests
oscillatory long-term behavior may already occur in very simple trait-drift
scenarios. A detailed analysis of the resulting attractors remains for future
work.

On the other hand, there is also a special kind of  trait-drift models where
identical effective mortalities of the dominant quasi-species do lead to
neutrally stable equilibria:
\begin{prop}\label{prop:trait-drift_special}
	Suppose that parameters of the trait-drift model (\ref{eq:trait-drift_equation}) are such that $\mathcal{R}\neq \varnothing$ and, for any $i,i'\in\mathfrak{I}$ with $r_{N(i)} = r_{N(i')} = \max(\mathcal{R})$,
	$$c_{1,i} = c_{1,i'}\quad \text{ and }\quad c_{2,i} + c_{4,i} + c_{5,i} = c_{2,i'} + c_{4,i'} + c_{5,i'}$$
	Then the dominant equilibrium simplex $\bar{\mathcal{M}}_{0,\max(\mathcal{R})}$ is a neutrally stable attractor.
\end{prop}
\begin{proof}
	See Appendix \ref{app:trait-drift_special_proof}.
\end{proof}

Propositions \ref{prop:trait-drift} and \ref{prop:trait-drift_special} show that
equilibrium coexistence of several quasi-species under pure trait-drift dynamics
requires fine-tuned parameter values -- a non-generic constraint which is
unlikely to hold in real ecosystems.\footnote{See Appendix \ref{app:genericity}
for the usage of ``generic'' in this article.}
\begin{cor}\label{cor:trait-drift}
	For generic parameter values and each $n\in\mathfrak N$ with spectral radius $r_n>1$ of $[C_1^{-1}P_0K]^n_n$, the trait-drift dynamics (\ref{eq:trait-drift_equation}) has a unique non-zero equilibrium supported on the $n$-block. Its sessile component is the associated non-negative Eigenvector scaled to have sum $1-1/r_n$. This equilibrium is unstable if $r_n<\max_{n'}r_{n'}$.
	\qed
\end{cor}

Note that Corollary \ref{cor:trait-drift} can be read as a simple algorithm for
deriving the equilibrium of a generic trait-drift model by just computing
dominant Eigenvalues and associated Eigenvectors --- an easy computational task.
While generic parameter values imply monodominance at macroscopic equilibrium,
the dominant quasi-species' trait-type distribution may still possess several
modes, as demonstrated in the following example.

\begin{exmp}\label{exmp:trait-drift}
We consider an environment consisting of a single macrosite $\mathfrak{B} =
\{1\}$ subdivided into 1000 establishment microsites and a single quasi-species,
$\mathfrak{N}=\{1\}$, subdivided into eleven trait-types $\mathfrak{A}_1 =
\{0,1,2,\ldots,10\}$ which compete for microsites in a strictly preemptive way
through $P=0$, $P_0=\One$. We posit a nonlinear trade-off between sessile
mortality and dispersal traits. For this, we consider the unit interval
$\mathcal{U}=[0,1]$ as a one-dimensional trait space and associate to
$i$-sessiles the trait value $u_i := \tfrac{A(i)}{10}$. We then parametrize
sessile mortality and dispersal coefficients as $c_{1,i} = c_1(u_i):=1+u_i$ and
$c_{3,i}=c_3(u_i):=5+30u_i-60u_i^2+40u_i^3$. For every existing nearest-neighbor
identifier we set $d^{i+1}_i=0.2=d^{i-1}_i$; at the boundary the nonexistent
transition is omitted, and $d^i_i=1-\sum_{i'\neq i}d^{i'}_i$. The remaining
propensity coefficients are $c_{2,i}=1$, $c_{4,i}=0$, and $c_{5,i}=10$.
Fig.\ref{fig:trait-drift} shows the equilibrium state along with the implemented
trade-off. Despite its artificial nature, this example demonstrates that
equilibrium trait-type abundances can be multi-modal even in the absence of true
quasi-species coexistence. Intuitively, this comes from the nonlinear 
production--mortality ratio of trait-types, $c_3(u)/c_1(u)$, which has an
interior maximum at $u=1/4$ and an equally high boundary maximum at $u=1$. The
way in which dispersal nonlinearity couples to trait spread is illustrated in
the phase diagram of Fig.\ref{fig:trait-drift}, where the horizontal $s_1$-axis
deforms the dispersal nonlinearity as $c_3(u;s_1)=s_1c_1(u)+(1-s_1)c_3(u)$, and
the vertical $s_2$-axis increases the spread of trait values around the middle
according to $u_i(s_2)=(1-s_2)\tfrac12+s_2\tfrac{A(i)}{10}$.
\end{exmp}

\begin{figure}[tbp]
	\centering
	\includegraphics[width=.79\textwidth]{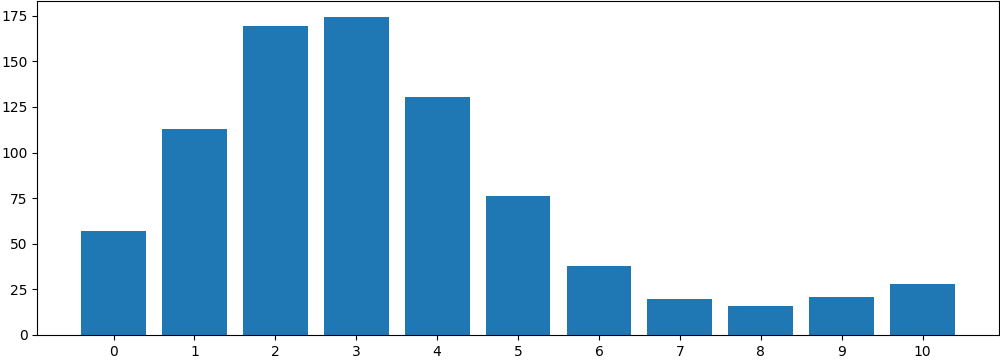}\\
	\begin{minipage}[t]{.4\textwidth}
		\centering
		\includegraphics[height=4.1cm, valign=t]{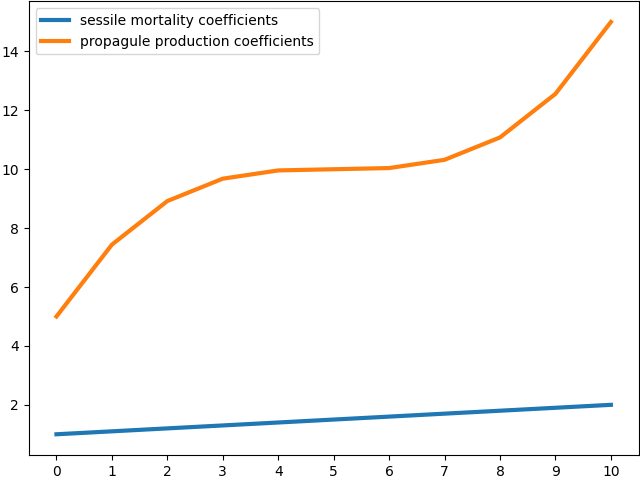}\hspace{.5cm}
	\end{minipage}
	\begin{minipage}[t]{.4\textwidth}
		\centering
		\includegraphics[height=4.3cm, valign=t]{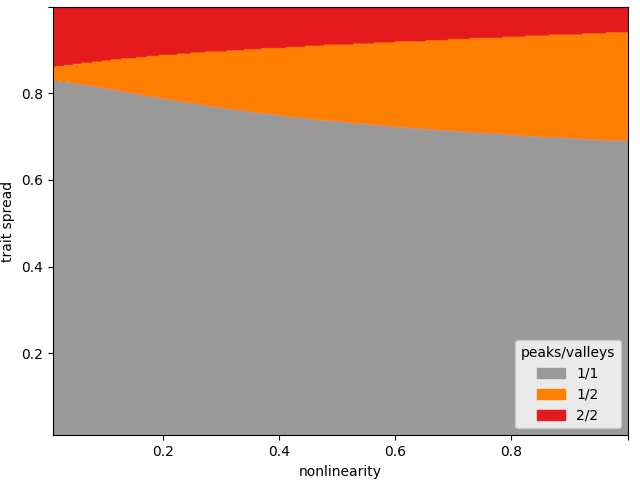}
	\end{minipage}
	\caption{Illustration of Example \ref{exmp:trait-drift}. Top: Equilibrium abundances of sessile trait-types. Bottom left: Sessile mortality coefficients $c_{1,i}$ and dispersal coefficients $c_{3,i}$ with $i=(1,a,1)$ as a function of trait-type $a=0,\ldots,10$. Bottom right: Phase diagram for deformed nonlinearity and spread of trait values (peaks/valleys including those at the boundary). The plotted parameter region $0.01\leq s_1\leq 0.98$, $0.01\leq s_2\leq 0.99$ ensures existence of a positive equilibrium with at least one peak/valley.}\label{fig:trait-drift}
\end{figure}

\subsection{The general case}

We now turn to the general propagule disperser model class, which can exhibit
equilibrium coexistence as demonstrated by the special cases in Sections
\ref{subsec:displacement_competition},\ref{subsec:environmental_heterogeneity}.
Here, computing equilibrium states becomes considerably more intricate than for
the trait-drift model. However, it turns out that \emph{coexistence coalitions}
near the extinction-threshold can still be entirely classified by the
zero-equilibrium Perron data and establishment parameters, without computing the
corresponding non-zero equilibria themselves. After having demonstrated this, we
will introduce a parameter subset stretching away from the extinction threshold on which the classification extends to a more general coexistence theorem.

As a matter of notation, we collect the model parameters into $q\in\mathcal{Q}$
(see Table \ref{tbl:mechanisms}). As before, we use (\ref{eq:stationary_x}) to
write the equilibrium condition in terms of sessile concentrations, resulting in
\begin{equation}\label{eq:sessiles_stationarity}
	F(x) x = 0 \quad \text{ with }\quad F(x):= -\left(C_1 + \diag(P^\transpose \bar{S}^{-1}Kx)\right) + (P_0 - \bar{S}^{-1}\diag((P_0 E - P)x))K
\end{equation}
We observe that the matrix $F(x)$ is $\mathfrak{N}$-block-diagonal, 
denote by $z_0^n$ the dominant Eigenvalue of the principal $n$-block of $F(0)$
and invoke the Perron--Frobenius theorem (see Appendix
\ref{app:general_spectrum}) to conclude that $z_0^n$ is real and simple. We then 
choose right and left Eigenvectors $\varpi_n\in[\mathcal{V}]_n$ and
$\varpi^n\in[\mathcal{V}^\ast]^n$ with positive $n$-block entries, normalized by
$\one^\transpose\varpi_n=1$ and $\varpi^n\cdot\varpi_{n'}=\delta^n_{n'}$, and denote the span of block-dominant Eigenspaces by
$\mathcal{Z}:=\mathrm{span}\set{\varpi_n}{n\in\mathfrak{N}}$. We collect the
dominant Eigenvalues into a vector $\zeta_0:=\sum_{n\in\mathfrak N}
z_0^n\varpi_n\in \mathcal{Z}$, and define the vectors
\begin{equation}\label{eq:cone_rays_unprojected}
	\psi_n := (F(0) - F(\varpi_n)) \cdot \varpi\quad \text{ with }\quad \varpi = \sum_n \varpi_n
\end{equation}
For any quasi-species coalition $\mathfrak{n}\subset \mathfrak{N}$, we define a
dominant-Eigenspace projector
$$\pi_\mathfrak{n}: \R^\mathfrak{I} \to \mathcal{Z},\quad x\mapsto \sum_{n\in\mathfrak{n}} (\varpi^n\cdot x)\varpi_n$$ 
and use these projectors to delineate a subset of coalitions
\begin{equation}\label{eq:locally_admissible_coalitions}
	\mathfrak{C}(q):=\left\{\varnothing\neq\mathfrak n\subset\mathfrak N
	\ \middle|\ [\zeta_0]_\mathfrak n\in \sum_{n\in\mathfrak{n}} \R_{>0}\cdot \pi_\mathfrak{n}(\psi_n)\right\}
\end{equation}
Below we will use $\mathfrak{C}(q)$ to classify coexistence coalitions close to an ``extinction threshold'' defined as follows: We subdivide $\mathcal{Q} = \mathcal{Q}_- \cup \mathcal{Q}_0 \cup \mathcal{Q}_+$ with
\begin{align*}
	&\mathcal{Q}_- := \set{q\in\mathcal{Q}}{\forall n \in \mathfrak{N}:\ z_0^n < 0}\\
	&\mathcal{Q}_0 := \set{q\in\mathcal{Q}}{\forall n \in \mathfrak{N}:\ z_0^n \leq 0 \text{ and } \exists n:\ z_0^n = 0}\\
	&\mathcal{Q}_+ := \set{q\in\mathcal{Q}}{\exists n \in \mathfrak{N}:\ z_0^n > 0}
\end{align*}
The \emph{extinction threshold} $\mathcal{Q}_0$ separates $\mathcal{Q}_-$ from
$\mathcal{Q}_+$. Its name will be justified as part of Theorem~\ref{thm:near-threshold}
showing that non-zero non-negative equilibria exist if and only if $q\in
\mathcal{Q}_+$. To arrive at a coexistence classification near the threshold,
Appendix \ref{app:general_local_proof} further constructs a generic subset of
$\mathcal{Q}_+$ which decomposes as a disjoint union of connected open sets
$\mathcal{H}\subset\mathcal{Q}_+$, called \emph{generic coexistence chambers},
such that $q\mapsto \mathfrak{C}(q)$ is constant on each such chamber. We will
now attempt to classify the non-negative equilibria on those chambers, beginning
with points close to the extinction threshold. To make this precise, we call
some open subset $\mathcal U\subset\mathcal Q_+$ a \emph{generic near-threshold
	cover} if $\overline{\mathcal U\cap\mathcal H}\cap\mathcal Q_0 \neq
\varnothing$, for every generic coexistence chamber $\mathcal H$, and
$$
\bigcup_{\mathcal H}
\overline{\mathcal{U}\cap\mathcal{H}}\cap\mathcal Q_0 = \mathcal Q_0
$$
where the union ranges over all chambers (and the closure is
taken in $\mathcal{Q}$). Note that the cover, thus defined, need not contain a
full neighborhood of each threshold point.

We are now ready to state our first main result by giving a classification of
coexistence coalitions near or below the extinction threshold:

\begin{thm}[Near-threshold coexistence]\label{thm:near-threshold}
	For any $q\in\mathcal{Q}_-\cup\mathcal{Q}_0$, the zero equilibrium is the only non-negative equilibrium, it is stable, and $\mathfrak{C}(q)=\varnothing$. On the other hand, for any $q\in\mathcal{Q}_+$, the zero equilibrium is unstable, there exists a non-zero non-negative equilibrium, and $\mathfrak{C}(q)\neq\varnothing$. Furthermore, there exists a generic near-threshold cover
	$\mathcal U\subset\mathcal Q_+$ such that, for every
	$q\in\mathcal U$, the zero equilibrium is unstable and the set of
	non-zero non-negative equilibria is in one-to-one correspondence
	with $\mathfrak C(q)\neq\varnothing$. More precisely, for each $\mathfrak n\in\mathfrak C(q)$ there is exactly one equilibrium $(x_\mathfrak n,y_\mathfrak n)^\transpose\in\bar{\mathcal X}\times\bar{\mathcal Y}$ whose positive quasi-species blocks are precisely those in $\mathfrak n$.
\end{thm}
\begin{proof}[Proof (sketch, details in Appendix~\ref{app:general_local_proof})] The assertions about the zero equilibrium can be framed in terms of the Perron spectrum of $F(0)$ after linearizing (\ref{eq:sessiles_stationarity}). On the other hand, any given above-threshold $q\in\mathcal{Q}_+$ can be connected to
the extinction threshold $\mathcal{Q}_0$ along a deformation curve $s\mapsto
q(s)$, with $q(1)=q$ and $q(0)\in\mathcal{Q}_0$, and such that
$\mathfrak{C}(q(s))$ is kept constant for $0<s\leq 1$, while deforming all
equilibria towards zero. Blowing up the zero equilibrium at the threshold
reduces the leading-order asymptotics of the equilibrium equation to the cone
condition defining $\mathfrak C(q)$.  Generic non-degeneracy and the implicit
function theorem then give exactly one non-negative branch of deformed
equilibria for every coalition in $\mathfrak C(q)$. The branches extend to
$s=1$ for $q$ in suitably chosen regions of each generic coexistence chamber
accumulating at the extinction threshold. A compactness argument excludes
additional non-negative equilibria in these regions. This exact classification
persists under small changes of the model parameters. Taking the union of the
resulting neighborhoods over all generic coexistence chambers produces an open
set on which the asserted one-to-one correspondence holds.  Every chamber
approaches the extinction threshold along the deformation, and density of the
generic parameter points shows that the resulting chamber traces cover all of
$\mathcal Q_0$. This gives the required generic near-threshold cover.
\end{proof}

Note that the above theorem does not claim that some coalition
$\mathfrak{n}\in\mathfrak{C}(q)$ is impervious to invasion by another
quasi-species $n'\notin \mathfrak{n}$.
Intuitively one might attempt to turn the near-threshold classification of
Theorem~\ref{thm:near-threshold} into a global result by connecting an arbitrary
$q\in\mathcal{Q}_+$ to the extinction threshold via a parameter deformation
which keeps coexistence coalitions constant. It turns out that this approach can
fail because a curve of deformed equilibria, lying in some face of
$\bar{\mathcal{X}}\times\bar{\mathcal{Y}}$, could collide with a subface. However, there is a
semi-algebraic subset $\mathcal{Q}_1\subset \mathcal{Q}_+$, defined in 
Appendix \ref{app:general_proof}, on which the
deformation approach proceeds without subface collisions, yielding our second
main result:
\begin{thm}[General coexistence]\label{thm:general_case}
	For any $q\in\mathcal{Q}_1$, the zero equilibrium is unstable,
	$\mathfrak{C}(q)\neq\varnothing$, and, for every $\mathfrak{n}\in\mathfrak{C}(q)$, there is at least one equilibrium $(x_\mathfrak n,y_\mathfrak n)^\transpose\in\bar{\mathcal X}\times\bar{\mathcal Y}$ whose positive quasi-species blocks are precisely those in $\mathfrak n$. In particular, there are at least $|\mathfrak{C}(q)|$ non-zero non-negative equilibria.
\end{thm}
\begin{proof}[Proof (sketch, details in Appendix \ref{app:general_proof})]
	The statements about the zero equilibrium follow in the same way as in Theorem
\ref{thm:near-threshold}. Next, assuming for a moment that $q\in \mathcal Q_1$
lies in some generic coexistence chamber, one again constructs a parameter
deformation $s\mapsto q(s)$ connecting $q$ to $q(0)\in\mathcal{Q}_0$, which
preserves the selected coalitions and gives, near the threshold, a unique
non-degenerate equilibrium for each of them. The conditions defining $\mathcal
Q_1$ exclude collisions with proper block-coordinate faces. On every closed
deformation interval bounded away from $s=0$, capacity bounds and this
exclusion confine the equilibria to a common compact subset of the relative
interior of the corresponding species face. Homotopy invariance of Brouwer
degree and the continuation principle then give a connected set of equilibria
reaching $s=1$. Finally, for each coalition at a general parameter point
$q\in\mathcal Q_1$, take generic approximations preserving that coalition and
its strict face condition. A subsequence of the corresponding target equilibria
converges, and the same exclusion prevents loss of species in the limit.
\end{proof}

Note that the conclusion of Theorem \ref{thm:general_case} does not exclude
additional equilibria or additional coexistence coalitions. The proofs of Theorem \ref{thm:near-threshold} and  \ref{thm:general_case} yield
convenient algorithms for finding coexistence coalitions which are detailed in
Appendix \ref{app:algorithms}. The near-threshold classification replaces
nonlinear equilibrium solves by considerably cheaper blockwise Perron
computations and linear systems over candidate coalitions. Evaluating the
sufficient conditions away from the threshold can, however, incur additional
cost. Still the algorithms are versatile enough to produce coexistence phase
diagrams for modest numbers of quasi-species, as shown in the following

\begin{exmp}\label{exmp:general}
	\begin{figure}[tbp]
	\includegraphics[width=.9\textwidth]{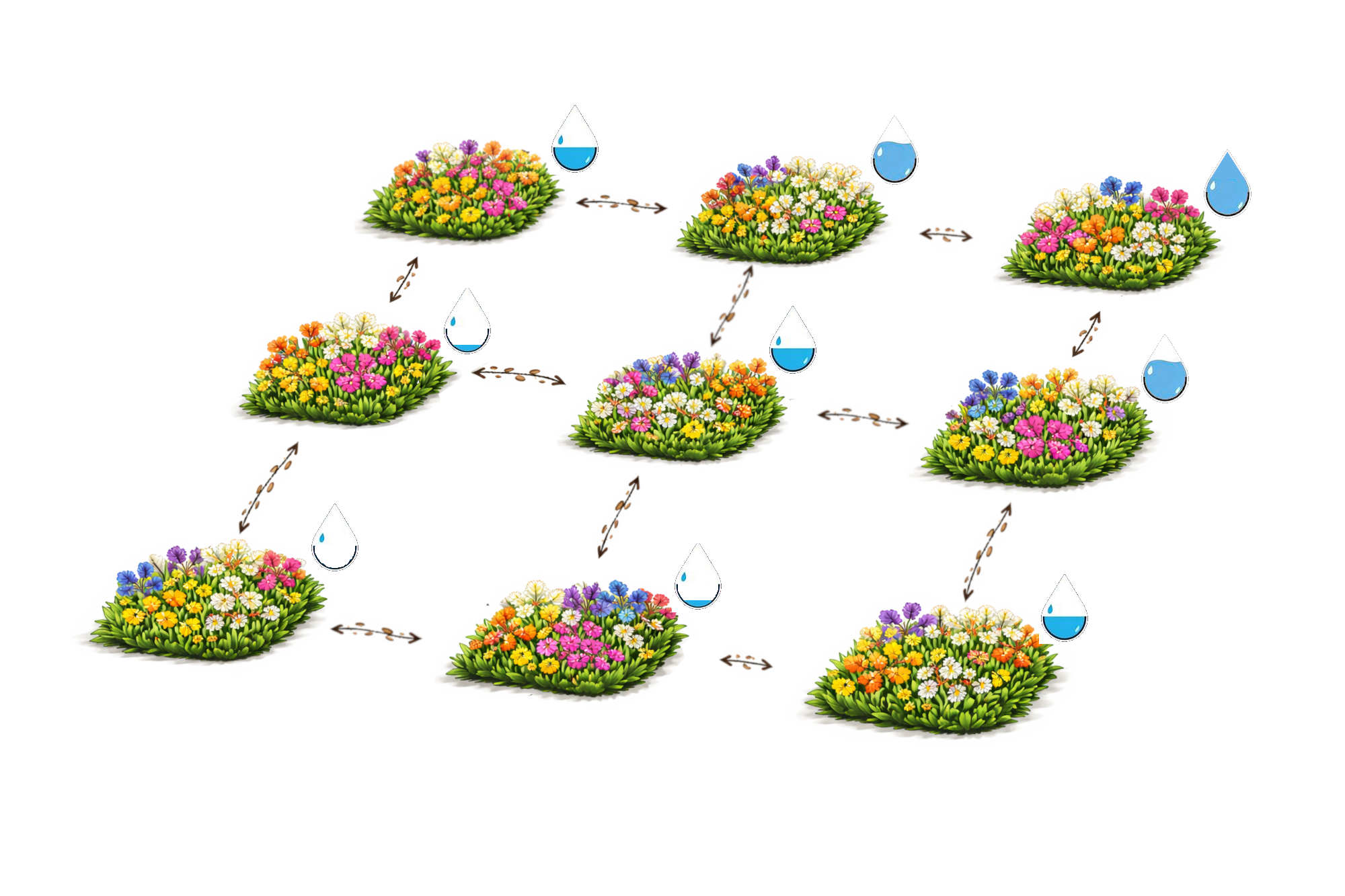}
	\caption{Illustration of the setup in Example \ref{exmp:general} showing the three-by-three grid of macrosites with diagonal resource availability gradient and dispersal connections indicated. }\label{fig:general_example}
\end{figure}
Consider a quasi-community with three quasi-species $\mathfrak{N} = \{0,1,2\}$ and three trait-types per species, $\mathfrak{A}_n = \{-1,0,1\}$. The environment consists of a three-by-three grid of macrosites, $\mathfrak{B} = \{-1,0,1\}\times\{-1,0,1\}$, each of unit size, which possess different levels of resource availability, e.g. soil water. Complementarily, trait-types differ through a one-dimensional resource requirement trait. The sessile mortality coefficient is fixed at $c_{1,i}=5$, while resource requirement and availability jointly determine the propagule production rate, $c_{3,i}=100 f^-_{N(i),A(i)}(s_2)f^+_{B(i)}(s_1)$. The difference between resource requirements of two types enters into the displacement probabilities through Calcagno's sigmoid function (see Sec.\ref{subsec:displacement_competition}). These choices represent an allocation trade-off as in \cite{Gleeson1990}:
greater investment in reproductive organs increases
propagule production, while reduced investment in roots
raises resource requirements. Higher requirements then
increase the probability of
displacement by individuals with lower requirements. A full account of model parameters is given in Appendix \ref{app:general_exmp}. In order to get snapshots of the coexistence phases in this model, we introduce three superparameters $s_1,s_2,s_3$:
\begin{enumerate}
	\item \emph{Environmental heterogeneity superparameter} $s_1\in [0,1]$ parametrizes resource availability across the metahabitat. Specifically, resource availability at macrosite $(b_1, b_2)\in\mathfrak{B}$ is posited as
	$$f^+_b(s_1) := 1 + s_1 \tfrac{1}{4}(b_1+b_2)$$
	Thus there is a resource availability gradient running diagonally across the habitat (see Fig.\ref{fig:general_example}). At zero heterogeneity, all environments have the same resource availability and, at maximum heterogeneity $s_1=1$, minimum and maximum resource availability are $\tfrac{1}{2}$ and $\tfrac{3}{2}$.
	
	\item \emph{Trait heterogeneity superparameter} $s_2\in [0,1]$ parametrizes resource requirement across species and trait-types. Specifically, sessiles of quasi-species $n$ and trait-type $a$ have resource requirement	
	$$f^-_{n,a}(s_2) := \tfrac{1}{2} + \tfrac14 s_2(n-1) + \tfrac{1}{16}a$$
	Thus the three trait-types of quasi-species $n$ always have equal spacing of $\tfrac{1}{16}$ in resource requirement independently of $s_2$. The superparameter controls how much the trait intervals of different quasi-species overlap: They coincide at $s_2=0$ and are pulled apart for increasing $s_2$, becoming completely disjoint for $s_2>\tfrac{1}{2}$.
	
	\item \emph{Displacement strength superparameter} $s_3\in[0,1]$ controls the strength of displacement competition such that the probability $p^i_{i'}$ for an establishing propagule with identifier $i = (n,a,b)$ and a resident sessile with identifier $i'=(n',a',b)$ is		
	$$p^i_{i'} = s_3 \left(1 + \exp\left(\beta(f^-_{n,a}(s_2) - f^-_{n',a'}(s_2))\right)\right)^{-1}$$
	according to Calcagno's sigmoid with fixed steepness $\beta$.
\end{enumerate}
\begin{figure}[tbp]
	\centering
	\includegraphics[width=\linewidth]{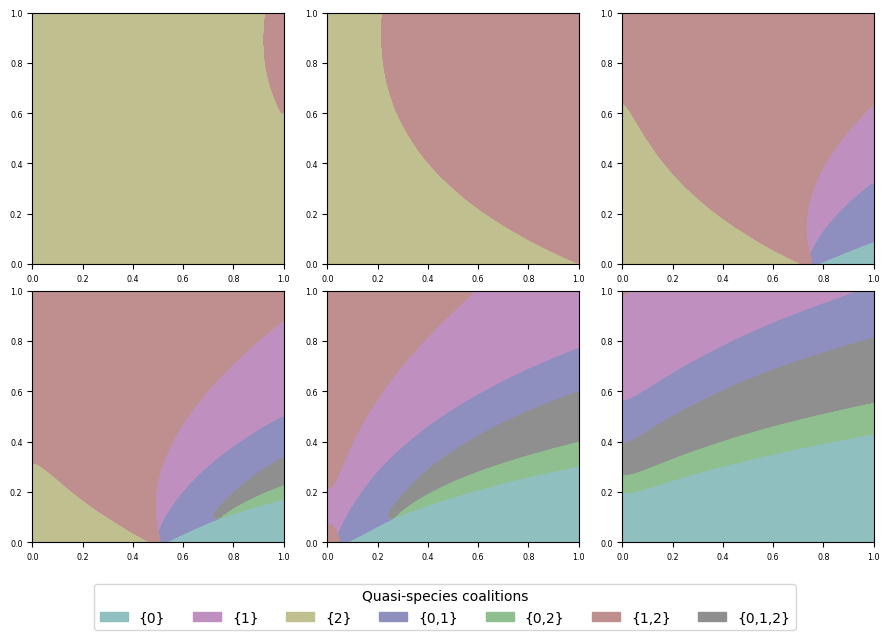}
	\caption{Near-threshold non-invadable coalitions for
		Example \ref{exmp:general} in heterogeneity superparameter space for six selected values of displacement strength, namely $s_3 = 0.0605,0.0825,0.0930,0.1040,0.1315,0.1730$, ordered increasingly from the upper-left to the lower-right panel. Environmental heterogeneity is indicated on the horizontal $s_1$-axis and trait heterogeneity on the vertical $s_2$-axis.}\label{fig:phase_diagram_general}
\end{figure}
We sampled the superparameter cube $[0,1]^3$ on a
$100\times100\times1000$ grid. Throughout this cube, all three
quasi-species have positive block-dominant Eigenvalues, so
the singleton coalitions $\{0\}$, $\{1\}$, $\{2\}$ always belong
to $\mathfrak C(q)$. We classified mixed coalitions through the near-threshold classification steps of Algorithm \ref{alg:general}, and display the resulting non-invadable coalitions at constant displacement strength slices, each of which represents a two-dimensional heterogeneity space comprising environmental and trait heterogeneity (see Fig.\ref{fig:phase_diagram_general}): At low $s_3$, the entire heterogeneity space predicts dominance of quasi-species 2, which is consistent with the trait-drift Example \ref{exmp:trait-drift}. However, already a small amount of displacement competition, namely $s_3 = 0.0605$, predicts coexistence of quasi-species 1 and 2 provided that there is high environmental heterogeneity and intermediate-to-high trait heterogeneity. A further small increase in $s_3$ extends the predicted range of the (1,2)-coalition across an increasing portion of heterogeneity space, followed by all other coalitions. At around $s_3=0.1315$, the initially prevalent predicted $2$-monodominance range disappears, and starting at $s_3=0.1730$ the other ranges begin to form bands covering all environmental heterogeneity conditions but only certain ranges in the trait-heterogeneity direction. This pattern is sustained for the remaining range of $s_3$-values, pushing the respective bands further up in the $s_2$-direction, i.e. requiring larger trait heterogeneity. In this fashion, the bands with highest trait heterogeneity progressively disappear and the band with lowest trait heterogeneity, namely the predicted $0$-monodominance range, takes over an increasingly large portion of heterogeneity space until, for displacement strength close to 1, almost all heterogeneity conditions predict $0$-monodominance with a small remaining (0,2)-range for low environmental heterogeneity and very large trait heterogeneity. This example demonstrates the intricate predicted relationships between displacement competition, environmental heterogeneity and trait heterogeneity in quasi-communities.
For the superparameter choice $s = (1.0, 1.0, 0.3)$ the numerical continuation heuristic in Algorithm \ref{alg:generalB} reaches a target equilibrium displayed in Fig.\ref{fig:bifurcation_curve}. Notably the respective parameters are not in $\mathcal{Q}_1$, demonstrating that membership is merely sufficient but not necessary for successful continuation.
\begin{figure}
	\centering
	\includegraphics[width=.75\linewidth]{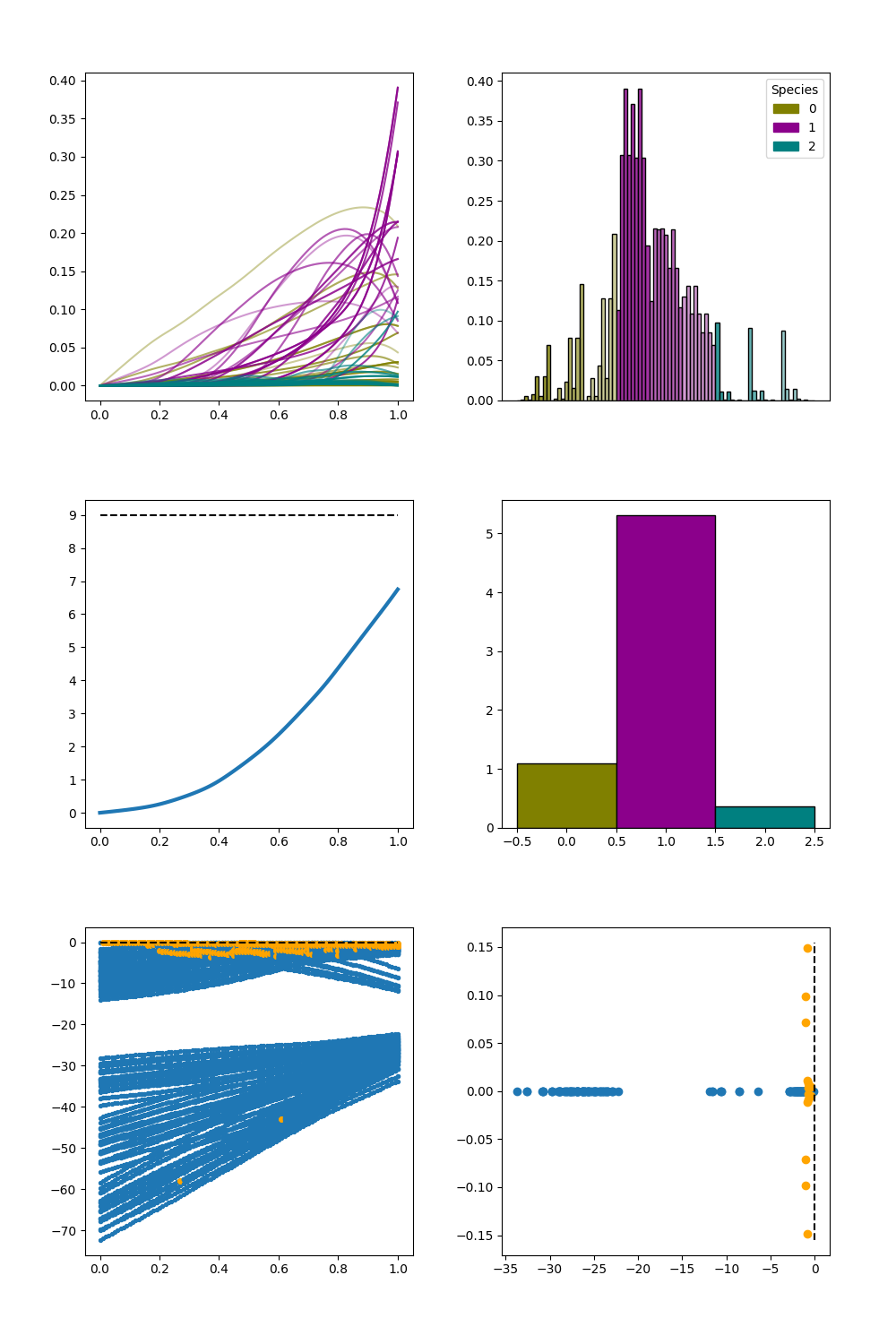}
\caption{Illustrative numerical continuation from the transcritical bifurcation
	locus to a target equilibrium. Top row: The 81 sessile concentrations $x^i$,
	with trait-types distinguished by shade. Middle row: Total sessile abundance
	over all macrosites; the left-hand panel compares it with the capacity bound 9,
	and the the right-hand panel gives the total sessile concentration of each
	quasi-species at the target equilibrium. Bottom row: Real parts of the Jacobian
	Eigenvalues along the sampled continuation points and their locations in the
	complex plane at the target equilibrium (real Eigenvalues in blue and non-real
	Eigenvalues in orange). At the sampled points, the plotted spectral abscissa
	remains negative.}\label{fig:bifurcation_curve}
\end{figure}

\end{exmp}

\section{Discussion}\label{sec:discussion}

In this article, we devised a class of stochastic process models for propagule
disperser quasi-communities that represents ingredients from several prominent
approaches to species coexistence, and verified that the class contains various
instances studied in the literature. We furthermore provided a computationally
inexpensive classification of macroscopic coexistence coalitions near the extinction threshold, along with a semi-algebraic parameter subset on which every near-threshold coalition is realized by at least one macroscopic equilibrium state.
Additional equilibria and coalitions may occur, and the associated equilibria
need not always be stable.

The class of propagule disperser models enables a more nuanced view of
diversity-promoting mechanisms by interpolating between existing
single-mechanism models and should therefore have broad relevance for the
discussion of coexistence, biodiversity, and conservation. More concretely, we
see several promising research directions emanating from the present foundation.

\emph{First}, the present work raises some questions in \emph{applied
mathematics}. In particular, the highly consequential topic of macroscopic
stability was largely postponed for another occasion. However, the contrast
between the mortality-homogeneous scenario in Proposition
\ref{prop:trait-drift_special} and the oscillatory
long-term behavior in Example \ref{exmp:trait-drift_counterexample} indicates
the range of dynamics covered by the propagule disperser class. Characterizing
stability regimes in the general model therefore remains an important problem.
Another related avenue is to investigate the full stochastic dynamics of
propagule disperser models, particularly around macroscopic equilibria. Here,
some preliminary exploration revealed an interesting scenario in which
equilibria are macroscopically stable but can have an arbitrarily small margin
of stability --- a situation where stochastic effects could play a substantial
role.

A \emph{second} avenue for future research pertains to \emph{ecological theory}.
Example \ref{exmp:general} gave a preliminary demonstration of non-trivial
interactions between different coexistence mechanisms---a phenomenon that
warrants further investigation in less artificial setups. On a different note,
it has been hypothesized that colonization events in metapopulations are subject
to strong Allee effects \cite{Ovaskainen2001}. These could be incorporated into
propagule disperser models by letting propensity coefficients depend on the
current state. Propensity coefficients could also be made time-dependent to
represent environmental change, such as recurring meteorological seasons. For
such non-autonomous scenarios, there might be a connection with recent research
on certain generalized Lotka--Volterra systems, where the invasion graph and
associated heteroclinic connections have been shown to persist under suitable
parameter perturbations \cite{Almaraz2024}. It would be interesting to ask
whether analogous robust transient organization occurs in propagule disperser
models. Yet another theory angle would be to relate propagule disperser models
to evolutionary dynamics. While the incorporated trait-drift mechanism could
already be considered a sort of evolutionary device, it seems like a crude way
to model actual evolutionary phenomena such as speciation events. Nevertheless,
trait-drift may still play a role in the way species evolve --- a speculation
that could potentially be formalized through the propagule disperser model
class.

A \emph{third} direction for further investigation points towards
\emph{empirical studies}. The fact that our model class explicitly resolves the
propagule stage as well as dispersal connections inside a metahabitat presents
various opportunities for injecting additional domain knowledge into the model
calibration procedure. The need for this in modeling terrestrial plant
communities has been raised in \cite[Ch.10]{Leck2008}, which makes clear that
typically measured propagule traits, such as seed mass, exert a very specific
influence on the plant's life cycle whose effects cannot be faithfully
attributed to overall species fitness. The present work was partly motivated by
the need to represent specific life-history details like these in an appropriate
way. More generally, the fact that the model class accommodates a wide spectrum
of resolutions, from relatively detailed accounts down to minimalist
demonstration models, could also help to elucidate the role that model
complexity plays in representing propagule disperser communities and natural
ecosystems at large. We hope that some of these avenues will prove fruitful in
future research.

\section*{Funding}

LA was funded through grant UNIQUE (PID2021-127202NB-C21) awarded to DA. This
work was also funded through grant PRIORITY (PID2021-127202NB-C22) awarded to
JAC. Both grants were funded by MCIN/AEI/10.13039/501100011033 and ``ERDF. A way
of making Europe''.

\section*{Supplementary information}

The code for reproducing Examples \ref{exmp:trait-drift_counterexample}, \ref{exmp:trait-drift}, \ref{exmp:general} can be found at \url{https://gitlab.com/aguirrle/coexistenceexamples.git}.

\bibliographystyle{plain}
\bibliography{references}

\appendix

\section{Mathematical terminology and notations}\label{app:notations}

\subsection{Cartesian spaces and linear algebra}\label{app:linear_algebra_notations}

For any finite index set $\mathfrak{I}$, the Cartesian space $\R^\mathfrak{I}$ is the vector space of upstairs-indexed tuples $(x^{i})^{i\in\mathfrak{I}}$. Any $x\in\R^\mathfrak{I}$ can be uniquely expressed as
\begin{equation}\label{eq:notations_vector}
	 x = \sum_{i\in\mathfrak{I}} x^i \ee_i
\end{equation}
with the standard basis vectors
$\ee_i = (\delta^j_i)^{j\in\mathfrak{I}}$, where $\delta^j_i$ is the Kronecker-delta symbol defined by
$$\delta^j_i = \begin{cases}
	1 & \text{if } j = i\\
	0 & \text{otherwise.}
\end{cases}$$
Alternatively one might totally order $\mathfrak{I} = \{i_1,\ldots, i_{|\mathfrak{I}|}\}$ and write $x=(x^{i})^{i\in\mathfrak{I}}$ as a column vector
\begin{equation}\label{eq:notations_vector_alt}
	x = \begin{pmatrix} x^{i_1} \\ \vdots \\ x^{i_{|\mathfrak{I}|}} \end{pmatrix}
\end{equation}
While this notation is more familiar in the context of $\R^N \equiv \R^{\{1,\ldots,N\}}$, the present article will almost exclusively use the presentation (\ref{eq:notations_vector}) as it avoids an explicit arbitrary ordering of the index set and leads to more compact notation.

The dual vector space $(\R^\mathfrak{I})^\ast = \R_\mathfrak{I}$ of one-forms or covectors is identified with the Cartesian space of downstairs-indexed tuples $(u_{i})_{i\in\mathfrak{I}}$ by virtue of the evaluation map 
$$(u_{i})_{i\in\mathfrak{I}}\cdot (x^{i})^{i\in\mathfrak{I}} := \sum_{i\in\mathfrak{I}} u_i x^i$$ 
Again any covector $u\in(\R^\mathfrak{I})^\ast$ can be presented through the dual basis covectors
$\ee^i$ as
\begin{equation}\label{eq:notations_covector}
	 u = \sum_{i\in\mathfrak{I}} u_i \ee^i
\end{equation}
or, upon ordering $\mathfrak{I}$, as a row vector
\begin{equation}\label{eq:notations_covector_alt}
	u = \left(u_{i_1},\ldots, u_{i_{|\mathfrak{I}|}}\right)
\end{equation}
The standard basis and its dual then satisfy
\begin{equation}\label{eq:notation_covector-vector}
	\ee^i\cdot \ee_j = \delta^i_j
\end{equation}
The transpose operator is the linear duality $\R^\mathfrak{I} \leftrightarrow \R_\mathfrak{I}$ defined with respect to the standard basis and its dual by
$$(\ee_i)^\transpose = \ee^i\quad \text{and}\quad (\ee^i)^\transpose = \ee_i$$
which can be intuitively understood as flipping column vectors into row vectors and vice-versa.

We index families of vectors in $\R^\mathfrak I$ with subscripts, as in $\ee_i$ or $v_1,v_2,\ldots$, and families of covectors in $\R_\mathfrak I$ with superscripts, as in $\ee^i$ or $\nu^1,\nu^2,\ldots$. Their entries carry the opposite index position, as in (\ref{eq:notations_vector},\ref{eq:notations_covector}).

Linear maps $F: \R^\mathfrak{I} \to \R^\mathfrak{I}$ are identified with elements of the tensor product space $\R^\mathfrak{I}\otimes (\R^\mathfrak{I})^\ast$ and can be presented as
\begin{equation}\label{eq:notation_matrix}
	 F = \sum_{i,j\in\mathfrak{I}} F^j_i \ee_j\otimes\ee^i
\end{equation}
where $\{\ee_j\otimes\ee^i\}$ is the standard basis of  $\R^\mathfrak{I}\otimes (\R^\mathfrak{I})^\ast$ acting linearly on $\R^\mathfrak{I}$ as
\begin{equation}\label{eq:notation_matrix-vector}
	(\ee_j\otimes\ee^i) \cdot \ee_k = \delta^i_k \ee_j
\end{equation}
This action extends to a composition rule for linear maps $\R^\mathfrak{I}\to\R^\mathfrak{I}\to\R^\mathfrak{I}$ through the ``tensor contraction''
\begin{equation}\label{eq:notation_matrix-matrix}
	(\ee_l\otimes\ee^k)\cdot(\ee_j\otimes\ee^i) = \delta^k_j \ee_l\otimes \ee^i
\end{equation}
Generally, the dot symbol $\cdot$ will always denote some linear action or multiplication which should be clear from the context and is sometimes omitted for less verbose notation. Again, upon ordering $\mathfrak{I}$, a linear map $F$ can be written as a matrix
\begin{equation}\label{eq:notation_matrix_alt}
	F = \begin{pmatrix}
		F^{i_1}_{i_1} & \cdots & F^{i_1}_{i_{|\mathfrak{I}|}} \\ 
		\vdots & \ddots & \vdots \\
		F^{i_{|\mathfrak{I}|}}_{i_1} & \cdots & F^{i_{|\mathfrak{I}|}}_{i_{|\mathfrak{I}|}}
	\end{pmatrix}
\end{equation}
While the presentations (\ref{eq:notations_vector_alt},\ref{eq:notations_covector_alt},\ref{eq:notation_matrix_alt}) lead to the conventional mnemonic for covector-vector, matrix-vector and matrix-matrix multiplication, they lead to cumbersome notation for structured index sets which have no canonical ordering and no fixed size. For this reason, the present article almost always prefers the presentations (\ref{eq:notations_vector},\ref{eq:notations_covector},\ref{eq:notation_matrix}) together with the multiplication rules (\ref{eq:notation_covector-vector},\ref{eq:notation_matrix-vector},\ref{eq:notation_matrix-matrix}). However, as a figure of speech, the term ``matrix'' will still be used occasionally for linear maps $\R^\mathfrak{I}\to\R^\mathfrak{I}$.

We write $\One$ for the identity map and $\one$ for the vector whose entries are all one. For a vector $x$, $\diag(x)$ is the diagonal matrix with entries $x^i$. For a matrix $F$, $\diag(F)$ retains its diagonal entries and sets its off-diagonal entries to zero.

\subsection{Block vectors and matrices}\label{app:block_notation}

Suppose $\mathfrak{I}$ is a finite index set and $N:\mathfrak{I}\to\mathfrak{N}$ is a projection to another index set $\mathfrak{N}$. This can be viewed as partitioning $\mathfrak{I}$ into subsets $N^{-1}(n)$ for $n\in\mathfrak{N}$ and a terminology for dealing with accordingly ``partitioned'' block vectors from $\R^\mathfrak{I}$ is defined as follows. For any subset $\mathfrak{n}\subset\mathfrak{N}$ the $\mathfrak{n}$-block space of $\R^\mathfrak{I}$ is $\R^{N^{-1}(\mathfrak{n})}$ and, for any $x\in\R^\mathfrak{I}$, the $\mathfrak{n}$-block of $x$ is the vector
$$(x^i)^{i\in N^{-1}(\mathfrak{n})} = \sum_{i\in N^{-1}(\mathfrak{n})} x^i\ee_i\ \in\ \R^{N^{-1}(\mathfrak{n})}$$
Note that the same symbols $\ee_i$ are implicitly used for the standard bases of $\R^\mathfrak{I}$ and $\R^{N^{-1}(\mathfrak{n})}$. We extend this terminology to single indices $n\in\mathfrak{N}$ by defining the $n$-block of $x$ simply as its $\{n\}$-block. The $\mathfrak{n}$-blocks and $n$-blocks of covectors from $(\R^\mathfrak{I})^\ast = \R_\mathfrak{I}$ are defined analogously by restricting the index set from $\mathfrak{I}$ to $N^{-1}(\mathfrak{n})$ etc. In the same fashion, given two subsets $\mathfrak{n},\mathfrak{n}'\subset\mathfrak{N}$ and a linear map $F:\R^\mathfrak{I}\to\R^\mathfrak{I}$, the $(\mathfrak{n},\mathfrak{n}')$-block of $F$ is defined as
$$\sum_{\substack{i'\in N^{-1}(\mathfrak{n})\\ i\in N^{-1}(\mathfrak{n}')}} F^{i'}_i \ee_{i'}\otimes\ee^i\ \in\  \R^{N^{-1}(\mathfrak{n})}\otimes\left(\R^{N^{-1}(\mathfrak{n}')}\right)^\ast$$
and, for $n,n'\in\mathfrak{N}$, the $(n,n')$-block of $F$ is simply its $(\{n\},\{n'\})$-block. Furthermore, the principal $\mathfrak{n}$-block of $F$ is its $(\mathfrak{n},\mathfrak{n})$-block and the principal $n$-block is the principal $\{n\}$-block.
$F$ is called $\mathfrak{N}$-block-diagonal if, for any $n,n'\in\mathfrak{N}$, the $(n,n')$-block of $F$ is zero unless $n=n'$ (i.e. unless the block is principal). Additionally, the following notation allows us to project to blocks inside the same ambient space:
\begin{align*}
	[x]_\mathfrak{n} &:= \sum_{i\in N^{-1}(\mathfrak{n})} x^i\ee_i\ \in\ \R^\mathfrak{I}\\
	[\nu]^\mathfrak{n} &:= \sum_{i\in N^{-1}(\mathfrak{n})} \nu_i\ee^i\ \in\ (\R^\mathfrak{I})^\ast\\
	[F]_{\mathfrak{n}'}^\mathfrak{n} &:= \sum_{\substack{i'\in N^{-1}(\mathfrak{n})\\ i\in N^{-1}(\mathfrak{n}')}} F^{i'}_i \ee_{i'}\otimes\ee^i\ \in\  \R^{\mathfrak{I}}\otimes\left(\R^\mathfrak{I}\right)^\ast
\end{align*}
A single-index version of this notation is defined as $[x]_n := [x]_{\{n\}}$, $[\nu]^{n} := [\nu]^{\{n\}}$, $[F]_{n'}^n := [F]_{\{n'\}}^{\{n\}}$.

The brackets retain the ambient coordinates and fill the excluded entries
with zeros. Whenever a spectrum, determinant, or inverse of a principal
block is discussed, it is understood on the corresponding block space
$\R^{N^{-1}(\mathfrak n)}$. In particular, the artificially added zero
rows and columns are not included in a block determinant or in its
dominant Eigenvalue. An inverse on a block may subsequently be extended
by zero to the ambient space. This distinction is useful when a block
is invertible although its zero extension is not.

\subsection{Non-negative matrices}\label{app:non-negative}

In this subsection, linear maps can be safely referred to as ``matrices'' with respect to the standard bases and some arbitrary ordering because the ordering is immaterial for present purposes. A matrix $F: \R^\mathfrak{I}\to\R^\mathfrak{I}$ is called non-negative, $F\geq 0$, or positive, $F>0$, if all entries $\ee^i F \ee_j$, $i,j\in\mathfrak{I}$ are non-negative or positive, respectively. A non-negative matrix $F$ is called irreducible if 
$$\forall i,i'\in\mathfrak{I}:\ \exists i_0,\ldots,i_L\in\mathfrak{I}:\ i_0=i,\ i_L=i',\ \forall 1\leq l \leq L:\ \ee^{i_{l}}F\ee_{i_{l-1}} > 0$$

The spectral radius of a matrix $F$ is $\max\set{|\lambda|}{\lambda \text{ is an Eigenvalue of } F}$.
For convenience, part of the classical Perron--Frobenius theorem for irreducible matrices is reproduced here from \cite[Ch.2, Thm.1.3, Thm.1.4]{Berman1994}:

\begin{lem}[Perron--Frobenius]\label{lem:Perron-Frobenius}
	Let $F$ be an irreducible non-negative matrix. Then the spectral radius of $F$ is a simple Eigenvalue whose one-dimensional Eigenspace is spanned by a vector with strictly positive entries. Furthermore, no other Eigenvector of $F$ has strictly positive entries.
\end{lem}

Here ``simple'' means algebraically simple, so no non-trivial Jordan
chain belongs to the Perron Eigenvalue. The same assertions hold for
left Eigenvectors by applying the theorem to the transpose. Also, a
non-zero non-negative Eigenvector of an irreducible non-negative matrix
is necessarily strictly positive: a zero coordinate in the Eigenvector
equation forces all its incoming neighbors to be zero, and strong
connectivity propagates this conclusion to every coordinate.

The \emph{spectral abscissa} of a matrix $F$ is the maximum of the real parts of its Eigenvalues. An Eigenvalue attaining the spectral abscissa will be called a \emph{dominant Eigenvalue}; it need not be unique for an arbitrary matrix. A simple consequence of Lem.\ref{lem:Perron-Frobenius} is

\begin{lem}\label{lem:essentially-non-negative}
	Let $F$ be an irreducible matrix whose off-diagonal entries are non-negative, i.e. $\forall i\neq j: \ee^iF\ee_j \geq 0$. Then $F$ has a unique dominant Eigenvalue. It is real and simple, and its one-dimensional Eigenspace is spanned by a vector with strictly positive entries. Furthermore, no other Eigenvector of $F$ has strictly positive entries.
\end{lem}
\begin{proof}
	For large enough $s>0$, the matrix $s\One + F$ is non-negative and irreducible. According to Lem.\ref{lem:Perron-Frobenius} its spectral radius $r$ is a simple, real Eigenvalue. Every other Eigenvalue $\lambda$ satisfies $|\lambda|\leq r$ and $\lambda\neq r$, hence $\Re(\lambda)<r$, even when $|\lambda|=r$. Thus $r$ is the unique dominant Eigenvalue, and the Eigenvectors of $s\One + F$ satisfy the positivity claims of the lemma. But these are exactly the Eigenvectors of $F$ with associated Eigenvalues translated by $-s$. In particular, $r-s$ is the unique dominant Eigenvalue of $F$.
\end{proof}

Matrices with non-negative off-diagonal entries are also called
\emph{Metzler matrices}. Their exponentials are non-negative, since
$e^{tF}=e^{-ts}e^{t(F+s\One)}$ and the power series of the second
exponential has non-negative terms. A matrix is \emph{Hurwitz} if all
its Eigenvalues have strictly negative real parts. An \emph{M-matrix}
has the form $s\One-B$ with $B\geq0$ and $s$ at least the spectral
radius of $B$; it is nonsingular when that inequality is strict.
In that case its inverse is non-negative, as follows from the geometric
series for $(\One-s^{-1}B)^{-1}$. These observations explain the
connection between M-matrices, stable linear loss equations, and the
order estimates used below.

One further consequence will be useful. If $F$ is irreducible and
Metzler, and $u>0$, then its dominant Eigenvalue lies between the
smallest and largest of the ratios $(Fu)^i/u^i$. Indeed, multiplication
by a positive left dominant Eigenvector expresses that Eigenvalue as
a weighted average of these ratios. In particular, $Fu>0$ implies a
positive dominant Eigenvalue, and $Fu<0$ implies a negative one.

\subsection{Equilibria}\label{app:equilibria}

Let $\mathfrak{I}$ be a finite index set, $\mathcal{U}\subset\R^\mathfrak{I}$ an open subset and $v:\mathcal{U}\to\R^\mathfrak{I}$ a sufficiently smooth vector field (all vector fields used in this article are polynomial). An equilibrium of $v$ is a point $u\in\mathcal{U}$ such that $v(u)=0$. The tangent space of such an equilibrium, i.e. $\R^\mathfrak{I}$, can be decomposed into the direct sum of three subspaces, each invariant under the Jacobian $\dd v(u)$. They are called the stable, unstable, and center subspaces and are spanned by the generalized Eigenspaces of $\dd v(u)$ whose Eigenvalues have negative, positive, and zero real part, respectively. The center manifold theorem \cite[Thm.3.2.1]{Wiggins2003} asserts that in the vicinity of the equilibrium $u$ there are corresponding stable, unstable, and center manifolds containing $u$, tangent to the respective subspaces and invariant under the flow of $v$. Moreover, trajectories in the stable and unstable manifolds have the same local asymptotic behavior as their linearizations in the stable and unstable subspaces, whereas the asymptotics on the center manifold are not determined by the respective linearization. Therefore, stability of the equilibrium can be shown in two steps:
\begin{enumerate}
	\item Prove that the unstable subspace is zero.
	\item If the center subspace is non-zero, show Lyapunov stability of the restricted flow on the center manifold; for asymptotic stability, also show local attraction to $u$.
\end{enumerate}
Furthermore, a \emph{neutrally stable attractor} is a set $\mathcal U_0\subset\mathcal U$ consisting of equilibria that is Lyapunov stable as a set and attracts every trajectory starting in some neighborhood of $\mathcal U_0$. A sufficient criterion is the following normal-attraction condition. Suppose that $\mathcal U_0$ is a compact $C^1$ embedded submanifold consisting of equilibria, the generalized center Eigenspace of $\dd v(u)$ equals $\mathrm T_u\mathcal U_0$ at every $u\in\mathcal U_0$, and all remaining Eigenvalues have real part at most $-\alpha$ for one $\alpha>0$ independent of $u$. The stable-foliation theorem for normally hyperbolic invariant manifolds then makes $\mathcal U_0$ a neutrally stable attractor \cite{HirschPughShub1977}. This criterion also records the required semisimplicity of the zero Eigenvalue in directions tangent to $\mathcal U_0$.

The above notions will also be needed at the boundary of closed subsets  $\mathcal{S}\subset\R^\mathfrak{I}$ which are forward-invariant under the flow of $v$. In this case, stability, attraction, and asymptotic
stability of equilibria in $\mathcal{S}$ will always be understood
relative to the subspace topology of $\mathcal{S}\subset\R^\mathfrak{I}$.  Whenever center-manifold theory is invoked at an equilibrium on the boundary of 
$\mathcal{S}$, we additionally assume that $v$ admits a sufficiently smooth extension to an open
neighborhood in $\R^\mathfrak{I}$. The center-manifold theorem is
applied to this ambient extension, whereas stability is assessed only
with respect to inward-pointing directions.

More explicitly, an equilibrium $u\in\mathcal S$ is Lyapunov stable
relative to $\mathcal S$ if, for every $\varepsilon>0$, there is a
$\delta>0$ such that every solution starting in $\mathcal S$ at
distance less than $\delta$ from $u$ exists for all $t\geq0$ and stays
at distance less than $\varepsilon$. It is locally attractive relative
to $\mathcal S$ if all solutions starting in some relative
neighborhood of $u$ converge to $u$. Relative asymptotic stability
means both properties. For a compact set of equilibria, replace
distance to $u$ by distance to that set; attraction to the set does not
assert convergence to the particular equilibrium nearest the initial
state. No smoothness of $\mathcal S$ is needed for these definitions.
For example, $\dot u=-u^2$ is asymptotically stable at zero relative to
$[0,\infty)$, although zero is not stable in the ambient real line.
Thus at a boundary point an ambient center-manifold calculation must
be accompanied by control of the trajectories allowed by $\mathcal S$.

\subsection{Generic properties}\label{app:genericity}

We partition parameter space into \emph{strata} by fixing which entries of $D$, $\tilde D$, and $C_4$ are zero and which are strictly positive. All parameter values otherwise remain free to vary within the constraints of Table~\ref{tbl:mechanisms}, including the allowed establishment probabilities. Each stratum is semialgebraic, and the directed graphs in Assumption~\ref{ass:block_irreducibility} are constant on it. That assumption therefore selects the admissible strata.

A property holds \emph{generically} on a stratum if its exceptions are contained in a proper algebraic subset of the stratum's Zariski closure. In particular, it holds on a dense relatively open subset of the stratum.

For example, Corollary~\ref{cor:trait-drift} implicitly uses the fact that principal species blocks of $C_1^{-1}P_0K$ with positive Perron Eigenvalues generically have distinct Perron Eigenvalues. Here the entries of both $C_1$ and $P_0$ are among the parameters allowed to vary.

\section{Macroscopic vector field}\label{app:macroscopic_vf}

In the model at hand, all propensity functions $r_k$ are of kinetic order one, i.e. each event is triggered by a single sessile or propagule. Hence, along a sequence $S_0\to\infty$ for which $S^b/S_0\to\bar S^b>0$, the transition rates have the density-dependent form
\[
\frac{r_k(X,Y)}{S_0}=\bar r_k(x,y),\qquad
x^i=\frac{X^i}{S_0},\qquad y^i=\frac{Y^i}{S_0}.
\]
The respective random vectors $V_k$, encoding the state-dependent choices made when an event is triggered, are defined so as to depend on
the occupancy fractions
\[
\frac{X^i}{S^{B(i)}}
=\frac{x^i}{S^{B(i)}/S_0}
\longrightarrow\frac{x^i}{\bar S^{B(i)}}.
\]
Thus the conditional mean jumps converge locally uniformly to functions
$\bar V_k(x,y)$. If the capacity ratios are already equal to their
limits, this is an exact identity. Equivalently, resolving each possible
outcome as a separate fixed jump gives finitely many jump channels with
density-dependent rates. Their limiting rate functions are locally
Lipschitz on the physical state space. On the concentration scale each
jump is of order $S_0^{-1}$ and the total rate is of order $S_0$ on
bounded sets. Hence, after localization to a bounded region, the expected
predictable quadratic variation on a fixed time interval is
of order $S_0^{-1}$. Doob's maximal inequality
\cite[Theorem~4.4.4]{Durrett2019}
therefore implies that the martingale fluctuations vanish
uniformly in probability on that interval. The drift is the sum of each event rate times its
conditional mean jump. With convergent scaled initial states, the
density-dependent law of large numbers \cite{Kurtz1970} therefore gives
the limiting vector field
$$\begin{pmatrix} x \\ y\end{pmatrix}^{\boldsymbol{\cdot}} = \sum_k \bar{r}_k(x,y) \bar{V}_k(x,y)$$
We are going to express these equations for $\dot{x}, \dot{y}$ in terms of the model parameters according to Table \ref{tbl:mechanisms}. To do this in a parsimonious way, we collect the parameters into matrices as follows:
\begin{align*}
	C_m =& \sum_i c_{m,i}\ee_i\otimes\ee^i,\ m=1,\ldots,5\\
	D =& \sum_{i,i'} d^{i'}_i \ee_{i'}\otimes \ee^i\\
	\tilde{D} =& \sum_{i,i'} \tilde{d}^{i'}_i \ee_{i'}\otimes \ee^i\\
	P_0 =& \sum_i p^i_0\ee_i\otimes\ee^i\\
	P =& \sum_{i,i'} p^i_{i'} \ee_i\otimes \ee^{i'}\\
	\bar{S} =& \sum_i \bar{S}^{B(i)} \ee_i\otimes\ee^i\\
	E =& \sum_{B(i)=B(i')} \ee_i\otimes\ee^{i'} 
\end{align*}
For a successful displacement the sessile jump is
$\ee_i-\ee_{i'}$: the colonist is gained and the resident is lost.
Establishment on a vacant microsite contributes only $\ee_i$.
In either case, and also after a failed establishment attempt, the
propagule is removed. In particular, replacement by the same identifier
$i'=i$ leaves sessile abundance unchanged and still consumes one
propagule. This accounts for both the gain and loss terms below.
We now express $\dot{x}$ in terms of those matrices:
\begin{align*}
	\dot{x}={}&-\sum_i c_{1,i}x^i\ee_i\\
	&+\sum_i c_{5,i}y^i
	\sum_{B(i')=B(i)}p^i_{i'}\frac{x^{i'}}{\bar S^{B(i')}}
	(\ee_i-\ee_{i'})\\
	&+\sum_i c_{5,i}y^ip_0^i
	\left(1-\sum_{B(i')=B(i)}
	\frac{x^{i'}}{\bar S^{B(i')}}\right)\ee_i\\
	={}&-\sum_i\left(c_{1,i}
	+\sum_{B(i')=B(i)}
	p_i^{i'}\frac{c_{5,i'}y^{i'}}{\bar S^{B(i')}}\right)x^i\ee_i\\
	&+\sum_i\left(p_0^i\bar S^{B(i)}
	-\sum_{B(i')=B(i)}(p_0^i-p_{i'}^i)x^{i'}\right)
	\frac{c_{5,i}y^i}{\bar S^{B(i)}}\ee_i\\
	=& - \left(C_1 + \diag(P^\transpose \bar{S}^{-1} C_5 y)\right) x\\
	&\quad + \left(P_0\bar{S} - \diag((P_0 E - P)x)\right)
	\bar{S}^{-1} C_5 y
\end{align*}

Similarly we may also express $\dot{y}$ in terms of the parameter matrices:
\begin{align*}
	\dot{y}={}&-\sum_i c_{2,i}y^i\ee_i
	+\sum_{i,i'}d_i^{i'}c_{3,i}x^i\ee_{i'}\\
	&+\sum_{i,i'}\tilde d_i^{i'}c_{4,i}y^i\ee_{i'}
	-\sum_i c_{4,i}y^i\ee_i-\sum_i c_{5,i}y^i\ee_i\\
	=& -C_2 y + D C_3 x + \tilde{D} C_4 y - C_4 y - C_5 y\\
	=& D C_3 x\\
	&\quad-\left(C_2 + (\One - \tilde{D}) C_4 + C_5\right)y\\
\end{align*}
For completeness, the boundedness needed for the finite-time limit can
be checked directly. Displacements cancel when the sessile equations
are summed at a macrosite, since their two identifiers have the same
macrosite. The remaining gains vanish at capacity, while mortality
points inward. The non-negative sessile polytope is therefore
forward-invariant. Summing the propagule equation gives
\begin{align*}
\frac{\dd}{\dd t}\one^\transpose y
&=\one^\transpose DC_3x
-\sum_i\left(c_{2,i}+c_{5,i}
  +c_{4,i}\Bigl(1-\sum_{i'}\tilde d_i^{i'}\Bigr)\right)y^i\\
&\leq\one^\transpose DC_3x
  -(\min_i c_{5,i})\one^\transpose y.
\end{align*}
The forcing is bounded by the sessile capacities, so propagule
concentrations stay bounded and the limiting solution exists globally.
For the stochastic models, only propagule production can increase the
total propagule count, and its rate is bounded by a constant times
$S_0$ on the sessile state space. A Poisson bound on these births
excludes explosion and allows the bounded-set fluctuation argument to
be localized on any fixed time interval. This yields the uniform
convergence in probability stated in the proposition.
\qed

\section{Equilibrium condition}\label{app:equilibrium_condition}

We show that the matrix $K = C_5 (C_2 + (\One-\tilde{D})C_4 + C_5)^{-1} D C_3$ is non-negative. Recall that by assumption all diagonal entries of $C_5$ are positive and the column sums of $\tilde{D}$ are bounded between 0 and 1. Thus $C_2 + (\One-\tilde{D})C_4 + C_5$ is strictly column-diagonally-dominant and hence a non-singular M-matrix (according to \cite[Ch.6,Thm.2.3.M35]{Berman1994}) with non-negative inverse \cite[Ch.6,Thm.2.3.N38]{Berman1994}. As $K$ is obtained by multiplying this inverse by another non-negative matrix, $K$ is also non-negative.

Furthermore, the propagule equation of the macroscopic equilibrium condition
$$\begin{pmatrix}	0 \\ 0 \end{pmatrix} = \begin{pmatrix}
	-\left(C_1 + \diag(P^\transpose \bar{S}^{-1} C_5 y)\right) & \left(P_0\bar{S} - \diag((P_0 E - P)x)\right)\bar{S}^{-1}C_5\\
	DC_3 & - (C_2 + (\One-\tilde{D})C_4 + C_5)
\end{pmatrix}\begin{pmatrix} x \\ y \end{pmatrix}$$

can now be solved for
$$y = (C_2 + (\One-\tilde{D})C_4 + C_5)^{-1} D C_3 x = C_5^{-1} K x$$
and substituting this into the sessile equation yields the non-trivial equilibrium condition
$$\left(- C_1 - \diag(P^\transpose \bar{S}^{-1}Kx) + (P_0 - \bar{S}^{-1}\diag((P_0 E - P)x))K\right) x = 0.$$
This elimination is reversible: every admissible sessile solution
determines exactly one non-negative propagule state, and substitution
back into both rows recovers the full equilibrium equation. The
identity $y=C_5^{-1}Kx$ is an equilibrium relation; it need not hold
along a non-stationary trajectory. The latter approximation requires
the additional time-scale argument in Appendix \ref{app:quasi-static}.
\qed

\section{Quasi-static approximation}\label{app:quasi-static}

We first note that $\tilde{K}$ is non-negative by the same argument as in Appendix \ref{app:equilibrium_condition}.
In order to derive the quasi-static limit, we substitute the rescaled propensity coefficients $\tfrac1\kappa c_{m,i}$, $m=2,3,4$, into the respective parameter matrices yielding the substitutions $C_m \leftarrow \tfrac1\kappa C_m$. The macroscopic dynamics then becomes
\begin{equation}
	\begin{pmatrix}	x_\kappa \\ y_\kappa \end{pmatrix}^{\boldsymbol{\cdot}} 
	= \begin{pmatrix}
		-\left(C_1 + \diag(P^\transpose\bar{S}^{-1} C_5 y_\kappa)\right) & \left(P_0 \bar{S} - \diag((P_0 E - P)x_\kappa)\right)\bar{S}^{-1}C_5\\
		D\tfrac{C_3}{\kappa} & - (\tfrac{C_2}{\kappa} + (\One-\tilde{D})\tfrac{C_4}{\kappa} + C_5)
	\end{pmatrix}\begin{pmatrix} x_\kappa \\ y_\kappa \end{pmatrix}
\end{equation}
This amounts to the coupled ODEs
\begin{align}
	\dot{x}_\kappa =& -\left(C_1 + \diag(P^\transpose\bar{S}^{-1} C_5 y_\kappa)\right) x_\kappa + \left(P_0 - \bar{S}^{-1}\diag((P_0 E - P)x_\kappa)\right)C_5 y_\kappa \label{eq:proof_quasi_static_deformed_sessiles}\\
	\kappa \dot{y}_\kappa =& DC_3 x_\kappa - (C_2 + (\One-\tilde{D})C_4 + \kappa C_5) y_\kappa \label{eq:proof_quasi_static_deformed_propagules}
\end{align}
We transform (\ref{eq:proof_quasi_static_deformed_propagules}) to the fast time scale $\tau = t/\kappa$ to get
\begin{equation*}
	\tfrac{\dd y_\kappa}{\dd \tau} = DC_3 x_\kappa - (C_2 + (\One-\tilde{D})C_4 + \kappa C_5) y_\kappa 
\end{equation*}
and, taking the limit $\kappa\to 0$, we obtain
\begin{equation}\label{eq:proof_quasi_static_deformed_propagules_adjoined}
	\tfrac{\dd y_0}{\dd \tau}(\tau; x) = DC_3 x - (C_2 + (\One-\tilde{D})C_4) y_0(\tau;x)
\end{equation}

Through the same reasoning as in Appendix \ref{app:equilibrium_condition}, we may observe that $(C_2 + (\One-\tilde{D})C_4)$ is a strictly column-diagonally-dominant M-matrix and thus invertible with non-negative inverse. Hence (\ref{eq:proof_quasi_static_deformed_propagules_adjoined}) has a unique non-negative equilibrium
\begin{equation}\label{eq:proof_quasi_static_deformed_propagules_adjoined_so}
	y = (C_2 + (\One-\tilde{D})C_4)^{-1} DC_3 x = C_5^{-1} \tilde{K} x
\end{equation}
whose Jacobian
$$\tfrac{\partial}{\partial y_0} \tfrac{\dd y_0}{\dd \tau} = - (C_2 + (\One-\tilde{D})C_4)$$
is the negative of a non-singular M-matrix and thus all its Eigenvalues have negative real parts \cite[Ch.6,Def.1.2]{Berman1994}. As (\ref{eq:proof_quasi_static_deformed_propagules_adjoined}) is an affine-linear ODE, we can conclude that  (\ref{eq:proof_quasi_static_deformed_propagules_adjoined_so}) is its unique globally exponentially asymptotically stable equilibrium. Furthermore, substituting  (\ref{eq:proof_quasi_static_deformed_propagules_adjoined_so}) into  (\ref{eq:proof_quasi_static_deformed_sessiles}) yields precisely the ODE (\ref{eq:quasi-static}) from the proposition. Its flow preserves $\bar{\mathcal{X}}$ and thus it has a bounded solution $x: [0,T] \to \bar{\mathcal{X}}$, for any $x(0)\in \bar{\mathcal{X}}$ and fixed $T>0$. Furthermore observe that the initial condition requirements
$$x_\kappa(0)\to x(0),
\qquad
\sup_{\kappa}\| y_\kappa(0)\| <\infty$$
imply that the solutions
$(x_\kappa, y_\kappa)$ exist on $[0,T]$ for all sufficiently small $\kappa$ (because $x_\kappa$ remains in the compact invariant sessile polytope $\bar{\mathcal{X}}$, while the stable linear equation for $y_\kappa$, forced by bounded $x_\kappa$, prevents finite-time blow-up).

The attraction rate in the fast equation is uniform in the frozen
sessile state $x$, because its linear part is independent of $x$.
Indeed, subtracting the equilibrium in
(\ref{eq:proof_quasi_static_deformed_propagules_adjoined_so}) gives
\begin{align*}
y_0(\tau;x)-C_5^{-1}\tilde Kx
={}&\exp\!\left[-\tau\bigl(C_2+(\One-\tilde D)C_4\bigr)\right]\\
&\quad\cdot\bigl(y_0(0;x)-C_5^{-1}\tilde Kx\bigr).
\end{align*}
The exponential has induced $1$-norm at most
$\exp(-\tau\min_i c_{2,i})$, by non-negativity and the column-sum
bound on its generator. Thus every bounded initial propagule state
lies in the basin of the same attracting constraint for its frozen
value of $x$. Under these conditions, Tikhonov's theorem \cite{Klonowski1983} yields
$$\lim_{\kappa\to 0} \sup_{0\leq t\leq T}
\| x_\kappa(t)- x(t)\|
= 0$$
As an alternative to using Tikhonov's theorem, convergence can also be quantified directly in the affine-linear fast system at hand, explicitly revealing the initial propagule layer:
Set
\[
e_\kappa(t):=y_\kappa(t)-C_5^{-1}\tilde Kx_\kappa(t).
\]
The total-propagule estimate from Appendix \ref{app:macroscopic_vf},
with the rescaled coefficients, bounds $y_\kappa$ uniformly on
$[0,T]$. Since $x_\kappa$ stays in the sessile polytope, its derivative
is uniformly bounded as well. Differentiating the preceding identity
and using the two evolution equations gives
\[
\kappa\dot e_\kappa
=-\bigl(C_2+(\One-\tilde D)C_4\bigr)e_\kappa
-\kappa C_5y_\kappa
-\kappa C_5^{-1}\tilde K\dot x_\kappa.
\]
Variation of constants and the exponential estimate above imply,
uniformly for $0\leq t\leq T$,
\[
\|e_\kappa(t)\|_1
\leq e^{-(\min_i c_{2,i})t/\kappa}\|e_\kappa(0)\|_1+O(\kappa),
\qquad
\int_0^T\|e_\kappa(t)\|_1\,\dd t=O(\kappa).
\]
Substitution of $y_\kappa=C_5^{-1}\tilde Kx_\kappa+e_\kappa$
into the sessile equation, followed by the Lipschitz estimate on the
compact state region and Gronwall's inequality, now gives
\[
\sup_{0\leq t\leq T}\|x_\kappa(t)-x(t)\|_1
\leq C_T\bigl(\|x_\kappa(0)-x(0)\|_1+\kappa\bigr)
\]
for a constant $C_T$ independent of sufficiently small $\kappa$.
Consequently $y_\kappa\to C_5^{-1}\tilde Kx$ uniformly on every
$[\delta,T]$ with $\delta>0$. The estimate identifies the initial
layer as a transient on the fast time scale $t/\kappa$; reducing an
order-one mismatch to order $\kappa$ may take time of order
$\kappa|\log\kappa|$. Uniform convergence on $[0,T]$ also holds for
the propagules if $e_\kappa(0)\to0$.
\qed

\section{Equilibria of the trait-drift model}\label{app:trait-drift}

\subsection{Spectrum of \texorpdfstring{$C_1^{-1}P_0 K$}{C1 inverse P0 K}}\label{app:trait-drift_spectrum}
We already showed that $K$ is non-negative (see Appendix \ref{app:equilibrium_condition}). The matrices entering $K$ are $\mathfrak N$-block-diagonal under the restrictions in Table \ref{tbl:mechanisms}; hence $K$, and therefore $C_1^{-1}P_0K$, is $\mathfrak N$-block-diagonal. Furthermore, each principal $n$-block of $K$ is irreducible due to Assumption \ref{ass:block_irreducibility}. Indeed, writing
\[
\Lambda:=C_2+C_4+C_5,
\]
we have
\[
C_2+(\One-\tilde D)C_4+C_5
=
\left(\One-\tilde DC_4\Lambda^{-1}\right)\Lambda.
\]
Moreover,
\[
\left\|\tilde DC_4\Lambda^{-1}\right\|_1
\leq
\max_i
\frac{c_{4,i}}{c_{2,i}+c_{4,i}+c_{5,i}}
<1,
\]
because $\tilde D$ is column-substochastic.  Hence
\begin{align*}
	K
	&=
	C_5\Lambda^{-1}
	\left(\One-\tilde DC_4\Lambda^{-1}\right)^{-1}
	DC_3\\
	&=
	C_5\Lambda^{-1}
	\sum_{l=0}^{\infty}
	\left(\tilde DC_4\Lambda^{-1}\right)^l
	DC_3.
\end{align*}
Since $C_3$, $C_5$, and $\Lambda^{-1}$ are positive diagonal
matrices, their presence does not change the resulting zero pattern.
Consequently, $K^{i'}_i>0$ precisely when an $i$-sessile can produce a
propagule through a transition represented by $D$, which then undergoes zero
or more operative migration transitions represented by $\tilde DC_4$
before reaching $i'$.  Thus the directed graph of each principal
species block of $K$ is exactly the graph specified in Assumption
\ref{ass:block_irreducibility}, and those blocks are irreducible in the sense of Appendix \ref{app:non-negative}.

Each summand represents a fixed number $l$ of migration events after
initial dispersal. As all summands are non-negative, a positive path
cannot be canceled by another path. A single entry of $K$ describes
one such reproductive passage; strong connectivity of its graph
allows several passages, and therefore several generations. In the
trait-drift specialization $\tilde D=0$, only the term $l=0$ remains.

There is a harmless degenerate case under the usual convention that a
one-vertex graph is strongly connected by a path of length zero: a
one-coordinate species block may have $D=0$, hence $K=0$. Its
reproduction radius is $r_n=0$, and it cannot occur in any non-zero
equilibrium. All statements about positive reproduction radii below
concern non-zero irreducible blocks; the zero one-coordinate case is
discarded directly whenever it arises.

As $K$ is $\mathfrak{N}$-block-diagonal with non-negative, irreducible principal $\mathfrak{N}$-blocks, the same holds for $C_1^{-1}P_0 K$. Then, according to the Perron--Frobenius theorem (see Lem.\ref{lem:Perron-Frobenius} in Appendix \ref{app:non-negative}), the spectral radius $r_n$ of any such block is its simple dominant Eigenvalue and an associated Eigenvector can be chosen with strictly positive entries and sum of entries equal to $1$. Denoting by $\rho_n$ the standard injection (see Appendix \ref{app:block_notation}) of this Eigenvector into $\mathcal{V}$, the corresponding Eigenspace on the $n$-block is spanned by $\rho_n$. For every non-zero irreducible block, $r_n>0$.

\subsection{Proof of Proposition \ref{prop:trait-drift}}\label{app:trait-drift_proof}
At a fully occupied macrosite, the total sessile derivative is
$-\one^\transpose C_1x<0$. Hence an equilibrium satisfies
$\one^\transpose x<1$, and the division by $1-\one^\transpose x$
in (\ref{eq:trait-drift_EV_equation}) is legitimate.
Of course, the equilibrium condition (\ref{eq:trait-drift_EV_equation}) is always solved by the zero state. Apart from that, suppose $x$ is a non-zero solution in $\bar{\mathcal{X}}$. As the right-hand side of (\ref{eq:trait-drift_EV_equation}) is collinear to $x$, the solution $x$ must lie in an Eigenspace of $C_1^{-1}P_0 K$, namely the one with Eigenvalue $r = (1-\one^\transpose x)^{-1}$ or put differently
\begin{equation}\label{eq:trait-drift_proof_normalization}
	\one^\transpose x = 1 - \tfrac{1}{r}
\end{equation}
The condition $x\in\bar{\mathcal{X}}$ then implies that $x$ is non-negative and thus $r\geq 1$, and as we assumed $x$ to be non-zero, we even get $r>1$, i.e. $r\in\mathcal{R}$. Moreover, the condition that $x$ be non-negative implies via the Perron--Frobenius theorem (Lem.\ref{lem:Perron-Frobenius}) that, for any $n\in\mathfrak{N}$, $[x]_n$ is either zero or a positive multiple of $\rho_n$ for some $n$ with $r_n = r$. Thus $x$ lies in the convex cone spanned by $\set{\rho_n}{r_n=r}$ and, taking into account (\ref{eq:trait-drift_proof_normalization}) as well as the normalization $\one^\transpose \rho_n = 1$, we get $x\in\bar{\mathcal{X}}_{0,r}$, for some $r\in\mathcal{R}$.

What remains to be shown are the assertions about stability of equilibria. For this we will need to investigate the Eigenvalues of the Jacobian $\dd\bar{v}$ i.e. the solutions $s\in\C$, $\sigma\in \C^\mathfrak{I}\oplus\C^\mathfrak{I}$ of 
\begin{equation}\label{eq:trait-drift_proof_jacobianEV}
	\dd \bar{v}(x,y)\sigma = s\sigma
\end{equation}
where $x$ is a solution of the equilibrium condition  (\ref{eq:trait-drift_EV_equation}) and $y = C_5^{-1}Kx$. Writing $\sigma = (\sigma_x,\sigma_y)^\transpose$ we may express the Jacobian in sessile/propagule-block form as
\begin{equation}\label{eq:trait-drift_proof_jacobian_xy}
	\dd\bar{v}(x,y)\sigma = \begin{pmatrix}
		-C_1 & (1 - \one^\transpose x)P_0 C_5\\
		DC_3 & - (C_2 + C_4 + C_5)
	\end{pmatrix}\begin{pmatrix} \sigma_x \\ \sigma_y \end{pmatrix} + \begin{pmatrix}
		- (\one^\transpose \sigma_x) P_0 C_5 y \\ 0
	\end{pmatrix}
\end{equation}
We investigate stability of the zero and non-zero equilibria separately.

\textbf{Zero-equilibrium:} Here, we begin by proving a stability lemma for the general case. To that end, we write the equilibrium condition from Prop.\ref{prop:stationary} as
\begin{equation}\label{eq:trait-drift_proof_equilibrium_x}
	F(x) x = 0 \quad \text{ with }\quad F(x):= -\left(C_1 + \diag(P^\transpose \bar{S}^{-1}Kx)\right) + (P_0 - \bar{S}^{-1}\diag((P_0 E - P)x))K
\end{equation}
We observe that $F(x)$ is $\mathfrak{N}$-block-diagonal and denote by $z_0^n$ the Eigenvalue with largest real part of the principal $n$-block of $F(0)$. As before, it can be shown that any such principal block is irreducible and has non-negative off-diagonal entries. Thus $z_0^n\in\R$ according to Lem.\ref{lem:essentially-non-negative}. 

\begin{lem}\label{lem:trait-drift_proof_stability_of_zero}
	The zero equilibrium in Prop.\ref{prop:stationary} is stable relative
	to $\bar{\mathcal{X}}\times\bar{\mathcal{Y}}$ if and only if
	\[
	\max_{n\in\mathfrak{N}} z_0^n \leq 0.
	\]
	If $\max_n z_0^n<0$, it is globally exponentially stable relative to
	$\bar{\mathcal{X}}\times\bar{\mathcal{Y}}$; if
	$\max_n z_0^n>0$, it is unstable.
\end{lem}

\begin{proof}
	Set
	\[
	J:=\dd\bar{v}(0,0)
	=
	\begin{pmatrix}
		-C_1 & P_0C_5\\
		DC_3 &
		-\bigl(C_2+(\One-\tilde{D})C_4+C_5\bigr)
	\end{pmatrix}.
	\]
	This matrix is $\mathfrak{N}$-block-diagonal and has non-negative
	off-diagonal entries.  Assumption
	\ref{ass:block_irreducibility}, together with the positivity of the
	diagonal entries of $P_0$, $C_3$, and $C_5$, implies that every
	principal $n$-block of $J$ is irreducible.
	To see this in the full sessile/propagule graph, lift each edge of
	the graph in Assumption \ref{ass:block_irreducibility} to initial
	dispersal, subsequent migration, and establishment. Every propagule
	coordinate reaches its matching sessile coordinate by establishment,
	and every coordinate is reached along one of these lifted paths.
	The sole exception is the zero one-coordinate reproduction block
	described above. Its full block is triangular, with strictly negative
	diagonal entries $-c_{1,i}$ and
	$-(c_{2,i}+(1-\tilde d_i^i)c_{4,i}+c_{5,i})$; it is therefore
	Hurwitz and may be treated separately in the following argument.
	
	The matrix
	\[
	C_2+(\One-\tilde{D})C_4+C_5
	\]
	is a non-singular M-matrix.  The Schur complement of its principal
	$n$-block in the corresponding principal $n$-block of $-J$ is
	\begin{align*}
		&\left[
		C_1
		-P_0C_5
		\bigl(C_2+(\One-\tilde{D})C_4+C_5\bigr)^{-1}
		DC_3
		\right]^n_n\\
		&\hspace{4cm}
		=[C_1-P_0K]^n_n
		=-[F(0)]^n_n.
	\end{align*}
	By the inverse-positivity characterization of nonsingular
	\(M\)-matrices
	\cite[Ch.~6, \S2, Thm.~2.3, N38]{Berman1994},
	applied to the block-inverse formula, the principal \(n\)-block
	of \(-J\) is a nonsingular \(M\)-matrix if and only if its Schur
	complement \(-[F(0)]^n_n\) is a nonsingular \(M\)-matrix.
	For clarity, the block-inverse formula proves both implications:
	if the Schur complement has a non-negative inverse, every block
	of the inverse of $-J$ is a sum of products of non-negative
	matrices; conversely, the sessile principal block of that inverse
	is the inverse of the Schur complement.

	To include the zero and positive cases explicitly, take the positive
	right Eigenvector $\varpi_n$ of $[F(0)]^n_n$. Direct multiplication
	gives, on this species block,
	\[
	J\begin{pmatrix}
	\varpi_n\\
	\bigl(C_2+(\One-\tilde D)C_4+C_5\bigr)^{-1}DC_3\varpi_n
	\end{pmatrix}
	=\begin{pmatrix}z_0^n\varpi_n\\0\end{pmatrix}.
	\]
	The vector on the left has positive sessile and propagule entries.
	Pairing this identity with a positive left dominant Eigenvector
	of the irreducible species block of $J$ shows that its dominant
	Eigenvalue has exactly the sign of $z_0^n$. This proves the
	negative, zero, and positive alternatives without extending the
	nonsingular Schur-complement criterion to a singular inverse.
	
	We now use the special sign structure of the nonlinear terms.  From
	(\ref{eq:macroscopic_vectorfield}) we have
	\begin{equation}\label{eq:zero_stability_comparison}
		\bar{v}(x,y)
		=
		J\begin{pmatrix}x\\y\end{pmatrix}
		-
		\begin{pmatrix}
			\diag(P^\transpose\bar{S}^{-1}C_5y)x
			+
			\diag((P_0E-P)x)\bar{S}^{-1}C_5y\\
			0
		\end{pmatrix}.
	\end{equation}
	For $(x,y)^\transpose\in
	\bar{\mathcal{X}}\times\bar{\mathcal{Y}}$, the vector subtracted on
	the right-hand side is non-negative, because $P\geq0$ and
	$P_0E-P\geq0$.  Variation of constants therefore gives
	\begin{align*}
		\begin{pmatrix}x(t)\\y(t)\end{pmatrix}
		={}&
		e^{tJ}\begin{pmatrix}x(0)\\y(0)\end{pmatrix}\\
		&-
		\int_0^t e^{(t-\tau)J}
		\begin{pmatrix}
			\diag(P^\transpose\bar{S}^{-1}C_5y(\tau))x(\tau)
			+
			\diag((P_0E-P)x(\tau))
			\bar{S}^{-1}C_5y(\tau)\\
			0
		\end{pmatrix}
		\dd\tau.
	\end{align*}
	Since $J$ has non-negative off-diagonal entries, $e^{tJ}$ is
	non-negative for every $t\geq0$.  The invariance of
	$\bar{\mathcal{X}}\times\bar{\mathcal{Y}}$ consequently implies the
	componentwise estimate
	\begin{equation}\label{eq:zero_stability_semigroup_bound}
		0
		\leq
		\begin{pmatrix}x(t)\\y(t)\end{pmatrix}
		\leq
		e^{tJ}\begin{pmatrix}x(0)\\y(0)\end{pmatrix}.
	\end{equation}
	
	If $\max_n z_0^n<0$, then $J$ is Hurwitz.  Hence there exist constants
	$M\geq1$ and $\omega>0$ such that
	\[
	\lVert e^{tJ}\rVert\leq Me^{-\omega t},
	\qquad t\geq0.
	\]
	Estimate (\ref{eq:zero_stability_semigroup_bound}) then proves global
	exponential stability of the zero equilibrium relative to
	$\bar{\mathcal{X}}\times\bar{\mathcal{Y}}$.
	
	Suppose next that $\max_n z_0^n=0$.  Every critical principal
	$n$-block of $J$ is irreducible and has dominant Eigenvalue zero.
	By Lem.\ref{lem:essentially-non-negative}, this Eigenvalue is simple,
	while all remaining Eigenvalues of the block have strictly negative
	real parts.  Since $J$ is $\mathfrak{N}$-block-diagonal, zero is
	therefore a semisimple Eigenvalue of $J$, even if several species
	blocks are critical.  All other Eigenvalues of $J$ have negative
	real parts, and consequently
	\[
	\sup_{t\geq0}\lVert e^{tJ}\rVert<\infty.
	\]
	Together with (\ref{eq:zero_stability_semigroup_bound}), this proves
	Lyapunov stability relative to
	$\bar{\mathcal{X}}\times\bar{\mathcal{Y}}$.
	
	Finally, if $\max_n z_0^n>0$, some irreducible principal block of
	$J$ has a positive dominant Eigenvalue and an associated strictly
	positive Eigenvector.  Sufficiently small positive multiples of
	this Eigenvector belong to
	$\bar{\mathcal{X}}\times\bar{\mathcal{Y}}$.  The principle of
	linearized instability therefore shows that the zero equilibrium is
	unstable relative to the physical state space.
	One can verify the relative assertion directly on the invariant
	face of this species. If $\nu>0$ is a left dominant Eigenvector
	of its block of $J$ with Eigenvalue $\lambda>0$, the amplitude
	$\nu\cdot(x,y)^\transpose$ has derivative
	$\lambda\nu\cdot(x,y)^\transpose+O(\|(x,y)\|^2)$.
	For non-negative states, positivity of $\nu$ makes this amplitude
	comparable to the norm. In a sufficiently small neighborhood its
	derivative is therefore at least half its positive linear growth
	rate. Every sufficiently small non-zero state on that face must
	eventually leave the neighborhood.
\end{proof}

We now return to the proof of Prop.\ref{prop:trait-drift}.  For every
$n\in\mathfrak{N}$, the matrix
\[
-[F(0)]^n_n=[C_1-P_0K]^n_n
\]
	is a non-singular or singular M-matrix precisely when
\[
\rho\!\left([C_1^{-1}P_0K]^n_n\right)=r_n
\]
	is smaller than or equal to \(1\), respectively. Consequently,
\[
z_0^n<0,\quad z_0^n=0,\quad\text{or}\quad z_0^n>0
\]
holds precisely when
\[
r_n<1,\quad r_n=1,\quad\text{or}\quad r_n>1,
\]
respectively.  Since
\[
\mathcal{R}
=\set{r_n}{n\in\mathfrak{N},\ r_n>1},
\]
we obtain
\[
\mathcal{R}=\varnothing
\quad\Longleftrightarrow\quad
\max_n r_n\leq1
\quad\Longleftrightarrow\quad
\max_n z_0^n\leq0.
\]
The stability assertions of Prop.\ref{prop:trait-drift} now follow from
Lem.\ref{lem:trait-drift_proof_stability_of_zero}.

\textbf{Non-zero equilibrium:} We now examine the case $\mathcal{R}\neq \varnothing$ and $x\in \bar{\mathcal{X}}_{0,r}$ for some $r \in \mathcal{R}\setminus\{\max(\mathcal{R})\}$. In this case, we may rewrite the Jacobian Eigenvalue equation (\ref{eq:trait-drift_proof_jacobianEV}) as
\begin{align}
(s+C_1 - (1-\one^\transpose x)P_0K_s)\sigma_x &= - (\one^\transpose\sigma_x)P_0K_0 x \label{eq:trait-drift_jacobianEVx}\\
\sigma_y =& C_5^{-1} K_s \sigma_x \label{eq:trait-drift_jacobianEVy}
\end{align}
with $K_s= C_5(s + C_2 + C_4 + C_5)^{-1} DC_3$ being defined if $s + C_2 + C_4 + C_5$ is non-singular, so in particular if $\Re(s)\geq 0$.
We will prove instability of the equilibrium by constructing a solution $s>0$ and $\sigma\neq 0$ to (\ref{eq:trait-drift_jacobianEVx},\ref{eq:trait-drift_jacobianEVy}).
To that end we first prove a useful lemma:
\begin{lem}\label{lem:trait-drift_proof_Ks}
	For real $s>0$, we have
	$$0 \leq K_s \leq K$$
	and all principal $\mathfrak{N}$-blocks of $K_s$ are irreducible.
\end{lem}
\begin{proof}
	Setting $A := C_2 + (\One-\tilde{D})C_4 + C_5$ we may show that
	\begin{equation}\label{eq:M-matrix_monotonicity}
		0 \leq (A+s)^{-1} \leq A^{-1}
	\end{equation}
	where the first inequality follows again from  \cite[Ch.6,Thm.2.3.M35,N38]{Berman1994} and the second can be seen through
	\begin{align*}
		(A+s)^{-1} =& (A+s)^{-1} - A^{-1} + A^{-1}\\
		=& (A+s)^{-1}(A - (A+s)) A^{-1} + A^{-1}\\
		=& \underbrace{-s (A+s)^{-1} A^{-1}}_{\leq 0} + A^{-1}
	\end{align*}
	The lemma then follows because (\ref{eq:M-matrix_monotonicity}) implies
	\begin{align*}
		0 \leq K_s &= C_5(C_2 + s + (\One - \tilde{D})C_4 + C_5)^{-1} DC_3\\
		&\leq C_5(C_2 + (\One - \tilde{D})C_4 + C_5)^{-1} DC_3\\
		&= K
	\end{align*}
	and irreducibility following from Assumption \ref{ass:block_irreducibility}.
\end{proof}

 We now recall that $x\in\bar{\mathcal{X}}_{0,r}$ implies $C_1^{-1}P_0K_0x = rx$. We may then turn (\ref{eq:trait-drift_jacobianEVx}) into
\begin{equation}\label{eq:trait-drift_proof_jacobianEVxa}
	(r\One - (C_1+s)^{-1}P_0 K_s)\sigma_x = - r^2 (\one^\transpose\cdot\sigma_x)(C_1+s)^{-1}C_1 x
\end{equation}
for any $s>0$. The spectral radii of all principal blocks of $(C_1+s)^{-1}P_0K_s$ are continuous functions of $s$ and tend to zero as $s\to\infty$. In particular, take a quasi-species $n_1\in\mathfrak{N}$ with $r_{n_1}>r$, which exists by the current assumption $r<\max\mathcal{R}$. The spectral radius of $[C_1^{-1}P_0K_0]^{n_1}_{n_1}$ is larger than $r$, so continuity gives $s_1>0$ such that the spectral radius of $[(C_1+s_1)^{-1}P_0K_{s_1}]^{n_1}_{n_1}$ equals $r$. We consider the corresponding Perron--Frobenius Eigenvector $\sigma_{x,n_1}\in\mathcal{V}$, which has strictly positive entries in its $n_1$-block and zero entries in all other blocks, thus in particular
\begin{equation}\label{eq:trait-drift_proof_intermed_n1}
	\one^\transpose \cdot \sigma_{x,n_1} > 0
\end{equation}
and satisfies
$$[r\One - (C_1+s_1)^{-1}P_0K_{s_1}]^{n_1}_{n_1}\sigma_{x,n_1} = 0$$
On the other hand, define $\mathfrak{N}_0\subset\mathfrak{N}$ to be the subset of those $n_0\in\mathfrak{N}$ satisfying $r_{n_0} = r$ and $[x]_{n_0}\neq 0$. The subset $\mathfrak{N}_0$ is not empty because $x\in\bar{\mathcal{X}}_{0,r}$.
Furthermore, for any $n_0\in\mathfrak{N}_0$, the spectral radius of $[ C_1^{-1}P_0K_0]^{n_0}_{n_0}$ is equal to $r$, and thus we may deduce from Lem.\ref{lem:trait-drift_proof_Ks} and \cite[Ch.2,Cor.1.5]{Berman1994} that the spectral radius of $[(C_1+s_1)^{-1}P_0K_{s_1}]^{n_0}_{n_0}$ is strictly smaller than $r$, meaning that the principal $n_0$-block of $[r\One - (C_1+s_1)^{-1}P_0K_{s_1}]^{n_0}_{n_0}$ is a non-singular irreducible M-matrix and thus has strictly positive inverse. We define $\left([r\One - (C_1+s_1)^{-1}P_0K_{s_1}]^{\mathfrak{N}_0}_{\mathfrak{N}_0}\right)^+$ by inverting the principal $n_0$-blocks for $n_0\in\mathfrak{N}_0$ and leaving all other entries equal to zero. Its principal $\mathfrak{N}_0$-blocks are strictly positive. In addition, all $\mathfrak{N}_0$-blocks of $x$ are strictly positive, and thus $$\sigma_{x,\mathfrak{N}_0} := - r^2([r\One - (C_1+s_1)^{-1}P_0K_{s_1}]^{\mathfrak{N}_0}_{\mathfrak{N}_0})^+  (C_1+s_1)^{-1}C_1 x$$
satisfies
$\one^\transpose\cdot \sigma_{x,\mathfrak{N}_0} < 0$. Combining this with (\ref{eq:trait-drift_proof_intermed_n1}), we can find $u > 0$ such that
\begin{equation}\label{eq:trait-drift_proof_linear_combination}
	\one^\transpose\cdot \left(\sigma_{x,\mathfrak{N}_0} + u\sigma_{x,n_1}\right) = 1
\end{equation}
We now define $\sigma_x := \sigma_{x,\mathfrak{N}_0} + u\sigma_{x,n_1}$ and may promptly verify that it solves (\ref{eq:trait-drift_proof_jacobianEVxa}) with $s=s_1>0$, proving that any $(x,y)^\transpose\in\bar{\mathcal M}_{0,r}$, $r<\max\mathcal{R}$, is indeed an unstable equilibrium.
The normalization in (\ref{eq:trait-drift_proof_linear_combination})
makes the right-hand side of (\ref{eq:trait-drift_proof_jacobianEVxa})
equal to the forcing used to define $\sigma_{x,\mathfrak N_0}$.
On the invading block the equation is homogeneous and is satisfied by
construction; on every other absent block both sides vanish.
Furthermore, the eigen-direction is physically admissible: its
negative sessile and propagule entries occur only in occupied blocks,
where the equilibrium entries are strictly positive, whereas its new
invader entries are positive. The capacity inequality is strict as
well. Thus sufficiently small positive perturbations in this
direction stay in the state space, and the instability is relative
to that space, not merely to its ambient extension.
\qed
	
\subsection{Proof of Proposition \ref{prop:trait-drift_special}}\label{app:trait-drift_special_proof}

As done previously in Appendix \ref{app:trait-drift_proof}, we consider the Jacobian Eigenvalue equation (\ref{eq:trait-drift_proof_jacobianEV},\ref{eq:trait-drift_proof_jacobian_xy}) and investigate solutions $s,\sigma$ with $\Re(s)\geq 0$. Using $(1 - \one^\transpose x) = \tfrac1r$, $y=C_5^{-1}Kx$, and the equilibrium condition $C_1^{-1}P_0 Kx=rx$, we may restate (\ref{eq:trait-drift_proof_jacobianEV},\ref{eq:trait-drift_proof_jacobian_xy}) as
\begin{equation}\label{eq:trait-drift_proof_jacobian_EV_special}
	0=\left(\dd\bar v(x,y)-s\One\right)\sigma = \underbrace{\begin{pmatrix}
		-(C_1 + s) & \tfrac{1}{r}P_0 C_5\\
		DC_3 & - (C_2 + C_4 + C_5 + s)
	\end{pmatrix}}_{(\ast)}\begin{pmatrix} \sigma_x \\ \sigma_y \end{pmatrix} + \begin{pmatrix}
		- (\one^\transpose \sigma_x) rC_1 x \\ 0
	\end{pmatrix}
\end{equation}

where $r=\max(\mathcal{R})$ due to the assumption
$x\in\bar{\mathcal{X}}_{0,\max(\mathcal{R})}$. For any species $n'$ with
$r_{n'}<r$, the $n'$-block of the trailing vector in
(\ref{eq:trait-drift_proof_jacobian_EV_special}) vanishes. The principal
$n'$-block of $(\ast)$ at $s=0$ is Hurwitz: the zero-equilibrium spectral
argument, with $P_0$ replaced by $P_0/r$, gives a negative dominant Eigenvalue
because $r_{n'}/r<1$. Its shift by $-s\One$ is therefore invertible for every
$\Re(s)\geq0$, and the $n'$-blocks of $\sigma_x$ and $\sigma_y$ vanish. The same
reduction applies to generalized Eigenvectors at zero. Ordering dominant species
before subdominant species makes the full Jacobian block upper triangular, with
a Hurwitz subdominant diagonal block. Every generalized zero Eigenvector
therefore has zero subdominant component. We may henceforth restrict to the
dominant species. At $s=0$, each principal species block of $(\ast)$ has a
simple zero Eigenvalue and all its other Eigenvalues have strictly negative real
parts. Thus zero is semisimple for $(\ast)$, which is invertible whenever
$\Re(s)\geq0$ and $s\neq0$.

We may therefore assume that $r_n=r$ for every
$n\in\mathfrak N$. For $\Re(s)\geq0$, define
$K_s:=C_5(C_2+C_4+C_5+s)^{-1}DC_3$ and rewrite
(\ref{eq:trait-drift_proof_jacobian_EV_special}) as
\begin{align}
	\sigma_y &= C_5^{-1} K_s \sigma_x \nonumber\\
	0 &= \left(C_1 + s - \tfrac{1}{r}P_0 K_s + r C_1 x \otimes \one^\transpose\right)\sigma_x \label{eq:trait-drift_proof_special_intermed0}
\end{align}
Note that these equations have a solution if and only if
\begin{equation}\label{eq:trait-drift_proof_special_intermed1}
	0 = \det\left(C_1 + s - \tfrac{1}{r}P_0 K_s + r C_1 x \otimes \one^\transpose\right)
\end{equation}

We now define the set of complex numbers $\C^+:=\set{s\in\mathbb{C}}{\Re(s)\geq 0 \text{ and } s\neq 0}$ and  distinguish two cases:
\begin{itemize}[leftmargin=*, itemsep=3pt]
	\item Case $s\in\C^+$: 
	Observe that $C_1 + s - \tfrac{1}{r}P_0 K_s$ is non-singular, for otherwise there would be some non-zero $\tilde{\sigma}_x$ in its kernel and by setting $\tilde{\sigma}_y:= C_5^{-1} K_s \tilde{\sigma}_x$ we could construct a non-zero vector $(\tilde{\sigma}_x, \tilde{\sigma}_y)^\transpose$ in the kernel of $(\ast)$. But such a vector can't exist because, as observed previously, $(\ast)$ has dominant Eigenvalue $-s$ with $s\in\C^+$. Thus we may rewrite (\ref{eq:trait-drift_proof_special_intermed1}) using the matrix determinant lemma for rank-1 modifications of a non-singular matrix as
	\begin{equation}\label{eq:trait-drift_proof_special_intermed2}
		0 = \left(1 + r \one^\transpose \cdot \left(C_1 + s - \tfrac{1}{r}P_0 K_s\right)^{-1} \cdot C_1 x\right)\det\left(C_1 + s - \tfrac{1}{r}P_0 K_s\right)
	\end{equation}
	We just observed that the trailing determinant is non-zero and thus (\ref{eq:trait-drift_proof_special_intermed2}) is equivalent to
	\begin{equation}\label{eq:trait-drift_proof_special_intermed3}
		0 = 1 + r \one^\transpose \cdot \left(C_1 + s - \tfrac{1}{r} P_0 K_s\right)^{-1} \cdot C_1 x
	\end{equation}
	Recall now that the Proposition assumes there exist $\tilde{c}_1, \tilde{c}_2 > 0$ such that $C_1 = \tilde{c}_1 \One$ and $\tilde{C}_2 := (C_2 + C_4 + C_5) = \tilde{c}_2 \One$. As a consequence we have 
	$$K_s = \tfrac{1}{\tilde{c}_2 + s} C_5 D C_3 = \tfrac{\tilde{c}_2}{\tilde{c}_2 + s} \tfrac{1}{\tilde{c}_2}C_5 D C_3 = \tfrac{\tilde{c}_2}{\tilde{c}_2 + s} K$$
	and hence (\ref{eq:trait-drift_proof_special_intermed3}) can be rewritten as
	\begin{align}
		0 &= 1 + r \tilde{c}_1\one^\transpose \cdot \left(\tilde{c}_1 + s - \tfrac{\tilde{c}_2}{\tilde{c}_2 + s}\tfrac{1}{r} P_0 K\right)^{-1} \cdot x \nonumber\\
		&= 1 + r \tilde{c}_1\one^\transpose \cdot \left(\tilde{c}_1 + s - \tfrac{\tilde{c}_1\tilde{c}_2}{\tilde{c}_2 + s} + \tfrac{\tilde{c}_1\tilde{c}_2}{\tilde{c}_2 + s} - \tfrac{\tilde{c}_2}{\tilde{c}_2 + s}\tfrac{1}{r} P_0 K\right)^{-1} \cdot x
		\nonumber\\
		&= 1 + r \tilde{c}_1\one^\transpose \cdot \left(\tilde{c}_1 + s - \tfrac{\tilde{c}_1\tilde{c}_2}{\tilde{c}_2 + s} + \tfrac{\tilde{c}_2}{\tilde{c}_2 + s} \left(\tilde{c}_1 - \tfrac{1}{r} P_0 K\right)\right)^{-1} \cdot x
		\nonumber\\
		&= 1 + r \tilde{c}_1\one^\transpose \cdot \left(\tilde{c}_1 + s - \tfrac{\tilde{c}_1\tilde{c}_2}{\tilde{c}_2 + s} + \tfrac{\tilde{c}_2}{\tilde{c}_2 + s} \left(C_1 - \tfrac{1}{r} P_0 K\right)\right)^{-1} \cdot x \label{eq:trait-drift_proof_special_intermed4}
	\end{align}
	We take the subexpression $f(s):= \tilde{c}_1 + s - \tfrac{\tilde{c}_1\tilde{c}_2}{\tilde{c}_2 + s}$ and compute its real part
	\begin{align*}
		\Re(f(s)) &= \Re\left(\frac{(\tilde{c}_1 + s)(\tilde{c}_2 + s) -\tilde{c}_1 \tilde{c}_2}{\tilde{c}_2 + s}\right)\\
		&= \Re\left(\frac{\tilde{c}_1 s + s(\tilde{c}_2 + s)}{\tilde{c}_2 + s}\right)\\
		&=  \tilde{c}_1 \Re\left(\frac{ s}{\tilde{c}_2 + s}\right) + \Re(s)\\
		&= \frac{\tilde{c}_1}{|\tilde{c}_2 + s|^2} \Re(\tilde{c}_2 s + |s|^2) + \Re(s)\\
		&= \frac{\tilde{c}_1}{|\tilde{c}_2 + s|^2} (\tilde{c}_2\Re(s) + |s|^2) + \Re(s)\\
		&\geq \frac{\tilde{c}_1 |s|^2}{|\tilde{c}_2 + s|^2}\\
		&> 0
	\end{align*}
	where we used $s\in\C^+$ in the last two estimates.
	In particular, we have $f(s)\neq 0$.
	The reason the inverse simplifies is that
	$(C_1-r^{-1}P_0K)x=0$: the matrix being inverted in
	(\ref{eq:trait-drift_proof_special_intermed4}) acts on $x$ as
	multiplication by $f(s)$, so its inverse acts on $x$ as
	multiplication by $f(s)^{-1}$. Explicitly,
	\begin{align}
		0 &= 1 + r \tilde{c}_1\one^\transpose \cdot f(s)^{-1}\left(\One + f(s)^{-1} \tfrac{\tilde{c}_2}{\tilde{c}_2 + s} \left(C_1 - \tfrac{1}{r} P_0 K\right)\right)^{-1} \cdot x\nonumber\\
		&= 1 + r \tilde{c}_1 f(s)^{-1} \one^\transpose \cdot \left(\One + f(s)^{-1} \tfrac{\tilde{c}_2}{\tilde{c}_2 + s} \left(C_1 - \tfrac{1}{r} P_0 K\right)\right)^{-1} \cdot x\nonumber\\
		&= 1 + r \tilde{c}_1 f(s)^{-1} \one^\transpose \cdot \left(\One - \left(\One + f(s)^{-1} \tfrac{\tilde{c}_2}{\tilde{c}_2 + s} \left(C_1 - \tfrac{1}{r} P_0 K\right)\right)^{-1} f(s)^{-1} \tfrac{\tilde{c}_2}{\tilde{c}_2 + s} \left(C_1 - \tfrac{1}{r} P_0 K\right)\right) \cdot x\nonumber\\
		&= 1 + r \tilde{c}_1 f(s)^{-1} \one^\transpose \cdot x\nonumber\\
		&= 1 + (r-1) \tilde{c}_1 f(s)^{-1} \label{eq:trait-drift_proof_special_intermed5}
	\end{align}
	where we used $(C_1 - \tfrac{1}{r} P_0 K)x = 0$ and $\one^\transpose \cdot x = 1-\tfrac{1}{r}$. Taking the real part of (\ref{eq:trait-drift_proof_special_intermed5}) we arrive at the contradiction
	\begin{align*}
		0 &= \Re(1 + (r-1) \tilde{c}_1 f(s)^{-1})\\ 
		&= 1 + (r-1) \tilde{c}_1 \Re(f(s)^{-1})\\
		&= 1 + (r-1) \tilde{c}_1 \frac{\Re(f(s))}{|f(s)|^2}\\
		&> 1
	\end{align*}
	proving that a solution $s,\sigma$ of (\ref{eq:trait-drift_proof_jacobian_EV_special}) with $s\in\C^+$ cannot exist.
	\item Case $s=0$: Condition (\ref{eq:trait-drift_proof_special_intermed0}) can now be rewritten as
	\begin{equation}\label{eq:trait-drift_proof_special_intermed6}
		\left(\One - \tfrac{1}{r}C_1^{-1} P_0 K\right)\sigma_x = -(\one^\transpose\sigma_x)r x
	\end{equation}
	Furthermore recall that 
	$
	x\in\ker\left(
	\One-\tfrac1r C_1^{-1}P_0K
	\right)
	$.
	Since the principal $\mathfrak{N}$-blocks of $C_1^{-1}P_0K$ are irreducible and
	non-negative, $x$ lies in the non-negative cone spanned by the
	Perron--Frobenius Eigenvectors $\rho_n$ with $r_n=r$. Choosing such an $n_0$ which additionally satisfies $[x]_{n_0}>0$, and denoting by $\rho^{n_0}$ its (block-dominant) left Eigenvector, we get $\rho^{n_0}\cdot x \neq 0$, and on the other hand also
	\begin{equation*}
		0 = -(\one^\transpose\sigma_x)r \rho^{n_0}\cdot x
	\end{equation*}
	Thus $\sigma_x$ has to satisfy
	\begin{align*}
		\one^\transpose\sigma_x &= 0\\
		\left(1 - \tfrac{1}{r}C_1^{-1} P_0 K\right)\sigma_x &= 0
	\end{align*}
	It remains to show that the Eigenvalue \(0\) is semisimple.  Set
	$$M:=\One-\frac{1}{r}C_1^{-1}P_0K,\qquad J:=\dd\bar v(x,y)$$
	We have just shown that every \(\sigma\in\ker J\) satisfies
	$$\one^\transpose\sigma_x=0,\qquad M\sigma_x=0,\qquad \sigma_y=C_5^{-1}K\sigma_x$$
	Suppose that $\sigma\in\ker(J)\cap\mathrm{im}(J)$, and write
	$J\eta=\sigma$, where
	$\eta=(\eta_x,\eta_y)^\transpose$.
	Using
	$$
	P_0C_5y=P_0Kx=rC_1x,\qquad
	C_1=\tilde c_1\One,\qquad
	C_2+C_4+C_5=\tilde c_2\One,
	$$
	the propagule block of $J\eta=\sigma$ gives
	$$
	\eta_y
	=C_5^{-1}K\eta_x-\tilde c_2^{-1}\sigma_y
	$$
	Substitution into the sessile block, together with
	$P_0 K\sigma_x=rC_1\sigma_x$, yields
	\begin{equation}\label{eq:trait-drift_proof_special_intermed7}
		M\eta_x+r(\one^\transpose\eta_x)x
		=
		-\bigl(\tilde c_1^{-1}+\tilde c_2^{-1}\bigr)\sigma_x
	\end{equation}
	
	For every $n$ with $r_n=r$, let $\rho^n$ be the
	corresponding left Perron--Frobenius Eigenvector, normalized such that
	$$\rho^n\cdot\rho_{n'}=\delta^n_{n'}$$
	Since $\rho^n M=0$, multiplication of (\ref{eq:trait-drift_proof_special_intermed7}) by
	$\rho^n$ gives
	$$
	r(\one^\transpose\eta_x)(\rho^n\cdot x)
	=
	-\bigl(\tilde c_1^{-1}+\tilde c_2^{-1}\bigr)
	(\rho^n\cdot\sigma_x).
	$$
	Both $x$ and $\sigma_x$ belong to the span of the dominant
	Perron--Frobenius Eigenvectors, so that
	$$
	x=\sum_{n:r_n=r}(\rho^n\cdot x)\rho_n,
	\qquad
	\sigma_x
	=\sum_{n:r_n=r}(\rho^n\cdot\sigma_x)\rho_n$$
	The preceding identities therefore show that
	$$
	\sigma_x
	=
	-\frac{r(\one^\transpose\eta_x)}
	{\tilde c_1^{-1}+\tilde c_2^{-1}}\,x.
	$$
	But $\one^\transpose\sigma_x=0$, whereas
	$\one^\transpose x=1-r^{-1}>0$.  Hence $\sigma_x=0$, and therefore
	also $\sigma_y=C_5^{-1}K\sigma_x=0$.  Consequently,
	$$
	\ker J\cap\mathrm{im}(J)=\{0\},
	$$
	which is equivalent to
	$$
	\ker J^2=\ker J.
	$$
	Thus the Eigenvalue $0$ is semisimple, and its generalized
	Eigenspace coincides with the tangent space of the equilibrium
	simplex
	$$\bar{\mathcal{M}}_{0,r} := \set{(x, C_5^{-1}K x)\in \bar{\mathcal{X}}\times\bar{\mathcal{Y}}}{x \in \bar{\mathcal{X}}_{0,r}}$$
	where tangent space at a face means the tangent space of its affine
	hull. To handle the corners explicitly, allow arbitrary real
	coefficients in $x=\sum_{n:r_n=r}a_n\rho_n$, subject only to
	$\sum_n a_n=1-r^{-1}$, and keep $y=C_5^{-1}Kx$. These identities
	still solve the polynomial equilibrium equations. They define a
	smooth affine manifold of equilibria extending the simplex beyond
	its boundary. Its tangent space is exactly the kernel found above. This is the tangent space of the affine hull of the entire equilibrium simplex, including at its faces and vertices.
	If only one species has reproduction radius $r$, the tangent space is zero and the Jacobian is Hurwitz.

	At every point of the simplex, all complementary Eigenvalues have
	strictly negative real part. Their number is constant because zero
	is semisimple and the tangent dimension is fixed. Compactness of
	the simplex and continuity of the spectrum therefore give a uniform
	negative bound on the complementary Eigenvalues. This bound persists
	in a small neighborhood of the simplex within the affine manifold.
	The local stable-foliation theorem for normally attracting manifolds
	\cite{HirschPughShub1977} now gives attraction of nearby trajectories
	to a nearby point of that manifold, with continuous dependence of
	the limiting point on the initial state. Compactness permits a
	common neighborhood and uniform local stability estimates.

	For trajectories starting in the physical state space, the limiting
	point also lies in that space by forward invariance and closedness.
	The intersection of the affine equilibrium manifold with the
	physical state space is precisely $\bar{\mathcal M}_{0,r}$, since
	non-negativity is equivalent to $a_n\geq0$. The attraction and
	stability estimates consequently apply to the simplex relative to
	$\bar{\mathcal X}\times\bar{\mathcal Y}$. This proves that it is a
	neutrally stable attractor, including at its faces and vertices.
\end{itemize}
\qed

\subsection{Example \ref{exmp:trait-drift_counterexample}}\label{app:trait-drift_counterexample}

We take $C_4=0$ in this example, so propagules are removed only by
mortality or establishment attempts.

In order to arrive at a counterexample we begin by setting the mortality and colonization attempt matrices so that the two trait-types are governed by very different time scales
\begin{equation*}
	C_1 := \begin{pmatrix} 0.2 & 0 \\ 0 & 20	\end{pmatrix}\ ,\quad C_2 := \begin{pmatrix} 0.9 & 0 \\ 0 & 27	\end{pmatrix}\ ,\quad C_5 := \begin{pmatrix} 0.1 & 0 \\ 0 & 3	\end{pmatrix}
\end{equation*}
We also choose $P_0 := \One$, meaning that colonization attempts on unoccupied microsites are always successful. Further, recall that an equilibrium is determined by (\ref{eq:trait-drift_EV_equation}), i.e.
\begin{align}
	C_1^{-1} K x &= r x \label{eq:counterexample_proof_0}\\
	\one^\transpose x &= 1 - \tfrac{1}{r}\label{eq:counterexample_proof_1}
\end{align}
Rather than guessing the remaining parameter matrices from scratch, we take a more interpretable approach by prescribing the equilibrium data $r,x$. Thus we choose $r:=4$ and a very uneven Perron-Frobenius Eigenvector
$$x := \begin{pmatrix} 0.7425 \\ 0.0075 \end{pmatrix}$$
scaled according to (\ref{eq:counterexample_proof_1}).
With this choice, the matrix $X:= \diag(x)$ is approximately proportional to the inverse of $C_1$ which implies that the ``forcing vector'' $C_1 x$ in (\ref{eq:trait-drift_proof_special_intermed0}) is balanced, ensuring that both trait-types are affected. What remains to be done is to couple the trait-types through a suitable trait-drift mechanism in $C_3$ and $D$. Note that (\ref{eq:counterexample_proof_0}) can be rewritten as the row-stochasticity condition
\begin{equation}
	Q \one = \one\quad \text{for}\quad Q = \tfrac{1}{r} X^{-1} C_1^{-1} K X
\end{equation}
Heuristically, we choose for $Q$ a very simple structure that still ensures strong one-way coupling between the trait-types, namely
\begin{equation*}
	Q := \begin{pmatrix}
		\alpha & 1-\alpha\\
		\alpha & 1-\alpha\\
	\end{pmatrix} \quad \text{with} \quad \alpha = 0.01
\end{equation*}
As a consequence, $K = r C_1 X Q X^{-1}$ becomes extremely asymmetric --- a promising scenario for instability. In order to find suitable $C_3, D$, we recall
\begin{equation}\label{eq:counterexample_proof_2}
	K=C_5 (C_2 + C_5)^{-1} D C_3
\end{equation}
We choose 
$$D:= K \diag(\one^\transpose K)^{-1} \approx \begin{pmatrix}
	0.497 & 0.497\\
	0.503 & 0.503
\end{pmatrix} $$
ensuring that $D$ is column-stochastic, and
$$C_3 := 10 \cdot \diag(\one^\transpose K)
\approx \begin{pmatrix}
	0.16081 & 0.0\\
	0.0 & 1576.1
\end{pmatrix}$$
which then satisfies (\ref{eq:counterexample_proof_2}).
With all of these parameter choices in place, we can verify whether the equilibrium is unstable by looking for a solution of (\ref{eq:trait-drift_proof_special_intermed3}) with positive real part. Such a solution is found at $s \approx 0.2116 \pm 1.6199i$.
An independent direct check uses the full Jacobian from
(\ref{eq:trait-drift_proof_jacobian_xy}). With the exact defining
parameters, rather than the displayed rounded entries of $D$ and
$C_3$, its characteristic polynomial is
\[
\det(s\One-\dd\bar v(x,y))
=s^4+52.394s^3+115.212s^2+83.88s+360.
\]
The remaining two Eigenvalues are approximately $-50.1261$ and
$-2.6911$. Thus the positive real parts of the complex pair establish
instability without any conclusion about the long-time orbit from
linearization alone.
\qed

\section{Equilibria in the general case}\label{app:general_case}

\subsection{Spectrum of \texorpdfstring{$F(0)$}{F(0)}}\label{app:general_spectrum}
We first note that $F(0)$ is $\mathfrak{N}$-block-diagonal with irreducible principal $\mathfrak{N}$-blocks (due to Assumption \ref{ass:block_irreducibility}) and non-negative off-diagonal entries. Thus Lem.\ref{lem:essentially-non-negative} shows that the dominant Eigenvalue $z_0^n$ of the principal $n$-block of $F(0)$ is real and simple, with left and right Eigenvectors having only positive entries. We denote the images of these Eigenvectors under the respective canonical injections by $\varpi^n\in\mathcal{V}^\ast$ and $\varpi_n\in\mathcal{V}$. We normalize them by $\one^\transpose\varpi_n=1$ and $\varpi^n\cdot\varpi_{n'}=\delta^n_{n'}$.

\subsection{Equilibrium preliminaries}\label{app:general_equilibrium_preliminaries}
The block structure concerns the matrix $F(x)$ at a fixed state.
Its coefficients can depend on all species through competition.
Consequently the Jacobian of the sessile equilibrium equation need
not be block diagonal: affinity of $F$ gives, for a perturbation $h$,
\[
\dd_x\bigl(F(x)x\bigr)h
=F(x)h+(F(h)-F(0))x.
\]
If an entire species block of $x$ is zero, the second term has zero
entries in that block. Thus $[F(x)]^n_n$ does describe the linearized
equilibrium equation for an absent species. For occupied species,
the second term must be retained. The Perron arguments below use
this distinction and do not assert stability of non-zero equilibria.
We let $(x,y)^\transpose$ be a non-trivial equilibrium and first show that it then must satisfy $[x]_n = 0$ for any $n$ with $z_0^n\leq 0$. For this, recall that $K$ has irreducible non-negative principal $\mathfrak{N}$-blocks (due to Assumption \ref{ass:block_irreducibility}). For $\one_b:=\sum_{B(i)=b} \ee_i$, summing the sessile equations at macrosite $b$ gives
$$(\one_b^\transpose x)^{\boldsymbol{\cdot}} = -\one_b^\transpose C_1 x + \left(1-\frac{\one_b^\transpose x}{\bar{S}^b}\right)\one_b^\transpose P_0 C_5 y$$
This expression vanishes at equilibrium. If $\one_b^\transpose x=0$, then $\one_b^\transpose x<\bar S^b$ follows directly from $\bar S^b>0$. If $\one_b^\transpose x>0$, the mortality term is strictly negative, so the colonization term cannot vanish and again $\one_b^\transpose x<\bar S^b$. Thus the diagonal matrix $P_0 - \bar{S}^{-1}\diag((P_0E-P)x)$ occurring in $F(x)$ has diagonal entries
$$p_0^i \left(1-\frac{\one_{B(i)}^\transpose x}{\bar{S}^{B(i)}}\right) + \frac{1}{\bar{S}^{B(i)}} \ee^i P x > 0$$
Hence $F(x)$ has the same off-diagonal directed graph as $K$, meaning that the principal $\mathfrak{N}$-blocks of $F(x)$ are also irreducible. 
Therefore, $F(x)x=0$ implies that either $[x]_n = 0$ or $[x]_n$ is strictly positive in its $n$-block. In the latter case we can estimate
\begin{align}
	0 &= \varpi^n F(x)x \nonumber\\
	&= \varpi^n F(x) [x]_n \nonumber\\
	&= \varpi^n F(0)[x]_n + \varpi^n (F(x) - F(0))[x]_n \nonumber\\
	&= z_0^n \varpi^n [x]_n + \varpi^n (F(x) - F(0))[x]_n \nonumber\\
	&\leq  \varpi^n (F(x) - F(0))[x]_n \nonumber\\
	&= -\varpi^n\cdot \left(\diag(P^\transpose \bar{S}^{-1}Kx) + \bar{S}^{-1}\diag((P_0 E - P)x)K\right) \cdot [x]_n \nonumber\\
	&\leq -\varpi^n\cdot \left(\diag(P^\transpose \bar{S}^{-1}K[x]_n) + \bar{S}^{-1}\diag((P_0 E - P)[x]_n)K\right) \cdot [x]_n \nonumber\\
	&\leq -\varpi^n\cdot \left(\diag(\diag(P) \bar{S}^{-1}K[x]_n) + \diag(\diag(P_0 E - P)[x]_n)\bar{S}^{-1}K\right) \cdot [x]_n \label{eq:proof_general_intermed1}
\end{align}
where $z_0^n\leq0$ is the hypothesis in the first inequality and $\diag(M)$ denotes the diagonal matrix with the same diagonal as $M$. Each later inequality discards non-negative competition terms. The two retained diagonal contributions can now be added exactly: since $p_i^i+(p_0^i-p_i^i)=p_0^i$, the last line of (\ref{eq:proof_general_intermed1}) equals
\[
-\sum_{N(i)=n}\varpi_i^n p_0^i x^i
\frac{(K[x]_n)^i}{\bar S^{B(i)}}<0.
\]
Indeed, in an occupied block every $x^i$ is positive, and the coordinate equilibrium equation with its positive colonization coefficient shows that $(K[x]_n)^i>0$. Alternatively, this follows from irreducibility for every non-zero reproduction block. The zero one-coordinate block discussed in Appendix \ref{app:trait-drift_spectrum} cannot be occupied at all, since its equation contains only sessile losses. We have therefore obtained a contradiction to $0\leq$ the last expression, proving $[x]_n=0$.

Thus no quasi-species $n$ with $z_0^n\leq 0$ participates in a non-negative equilibrium. In particular, the zero state is the only non-negative equilibrium if $\max_n z_0^n\leq0$. The assertions about its stability are precisely those shown previously in Lemma \ref{lem:trait-drift_proof_stability_of_zero}.

The cone definition of $\mathfrak C(q)$ has the same restriction. Every projected coordinate of every ray $\psi_m$ is non-negative, because $F(0)-F(\varpi_m)\geq0$. Moreover, when $z_0^n\geq0$, the Eigenvector identity $P_0K\varpi_n=(C_1+z_0^n)\varpi_n$ shows $K\varpi_n>0$ on its block. The strictly positive self-term just calculated, with $x$ replaced by $\varpi_n$, therefore gives $\varpi^n\psi_n>0$. A positive combination of the projected rays cannot have a negative coordinate, and cannot have zero in the $n$-coordinate when the $n$-ray is present and $z_0^n=0$. Defining $\mathfrak N_+(q):=\{n\in\mathfrak N\mid z_0^n>0\}$, we can then conclude that
\begin{equation}\label{eq:positive_blocks}
	\mathfrak{n}\in \mathfrak{C}(q)\ \Rightarrow\ \mathfrak{n}\subset \mathfrak{N}_+(q)
\end{equation}
i.e. that only the quasi-species from $\mathfrak{N}_+(q)$ can participate in any
non-trivial equilibrium. Conversely, we also have
\begin{equation}
	\mathfrak{N}_+(q)\neq \varnothing\ \Rightarrow\ \mathfrak{C}(q)\neq \varnothing
\end{equation}
because the latter set contains at least the singleton coalitions $\{n\}$, for $n\in \mathfrak{N}_+(q)$.

\subsection{Proof of Theorem \ref{thm:near-threshold}}\label{app:general_local_proof}

We first classify all non-negative equilibria sufficiently close to the threshold along a fixed ``critical-mortality'' deformation. We then show that the classification persists on open parameter neighborhoods and that these form the cover required by Theorem \ref{thm:near-threshold}.

For fixed $q$, the critical-mortality deformation is defined as
\begin{equation}\label{eq:deformed_equilibria}
	F(x;s):=F(x)+(s-1)\sum_{n\in\mathfrak{N}_+(q)}z_0^n[\One]^n_n,
	\qquad F(x;s)x=0,\qquad s\in[0,1].
\end{equation}
It adds the blockwise mortality $(1-s)z_0^n$ to every $n\in\mathfrak N_+(q)$, as defined above in Appendix~\ref{app:general_equilibrium_preliminaries}, while leaving the remaining model parameters unchanged. Hence $s=1$ recovers the target model, whereas at $s=0$ the dominant Eigenvalue of every positive block is zero. Along the deformation that Eigenvalue is $s z_0^n$.

\begin{prop}[Critical mortality coexistence]\label{prop:critical_mortality}
	Let $q\in\mathcal Q_+$ be generic in the sense that Properties \ref{property:W_principal} and \ref{property:positive_bifurcation} in Appendix \ref{app:general_local_proof} hold after the blocks outside $\mathfrak N_+(q)$ have been discarded. Then
	\[
	\varnothing\neq\mathfrak{C}(q)\subset
	\left\{\mathfrak{n}\subset\mathfrak{N}\ \middle|\ z_0^n>0\text{ for every }n\in\mathfrak{n}\right\}.
	\]
	Moreover, there is an $\varepsilon\in(0,1]$ such that, for every $s\in(0,\varepsilon)$, the non-zero non-negative equilibria of the deformed model are in one-to-one correspondence with $\mathfrak C(q)$. More precisely, for every $\mathfrak n\in\mathfrak C(q)$ there is exactly one equilibrium $(x_\mathfrak n(s),y_\mathfrak n(s))^\transpose\in\bar{\mathcal X}\times\bar{\mathcal Y}$ whose positive quasi-species blocks are precisely those in $\mathfrak n$, and it satisfies
	\begin{equation}\label{eq:critical_mortality_expansion}
		x_\mathfrak n(s)=s\sum_{n\in\mathfrak n}
		(\zeta_\mathfrak n)_n\varpi_n+O(s^2),
		\qquad
		[W]^\mathfrak n_\mathfrak n\zeta_\mathfrak n=(z_0^n)_{n\in\mathfrak n}.
	\end{equation}
	The zero equilibrium is unstable for every $s>0$.
\end{prop}
\begin{proof}
After discarding the blocks with $z_0^n\leq0$ as justified above, we may assume without loss of generality that $z_0^n>0$ for all $n\in\mathfrak N$. We deform the matrix-valued function $x\mapsto F(x)$ with a parameter $s\in(-\infty,1]$ as
\begin{equation}\label{eq:proof_general_deformedF}
	F(x;s) := F(x) + (s-1)\sum_{n\in\mathfrak N}z_0^n [\One]^n_n
\end{equation}
and consider the deformed sessiles' stationarity equation
\begin{equation}\label{eq:proof_general_deformed}
	F(x;s)x = 0
\end{equation}
For $s=1$, equation (\ref{eq:proof_general_deformed}) is identical to the original sessile stationarity equation (\ref{eq:sessiles_stationarity}), while for $s\leq0$ its only non-negative solution is zero: all principal $\mathfrak N$-blocks of $F(0;s)$ have dominant Eigenvalues at most zero, and the preceding argument excludes any non-zero non-negative equilibrium. As $s$ crosses zero, the zero equilibrium loses stability in every retained block. We now resolve the resulting simultaneous transcritical bifurcation.

To reveal this bifurcation, we are going to construct a suitable coordinate transform. Start by defining the projectors $\Pi_{0,n} = \varpi_n\otimes\varpi^n$ and $\Pi_{-,n}:=[\One]^n_n - \Pi_{0,n}$ and assemble them into
$$\Pi_0 := \sum_{n\in\mathfrak{N}} \Pi_{0,n}\quad \text{and}\quad \Pi_- := \sum_{n\in\mathfrak{N}} \Pi_{-,n}$$
Observe that $\Pi_0$ and $\Pi_-$ are projectors whose images together span $\mathcal{V}$. Therefore the following is an invertible change of coordinates:
\begin{equation}
	\mathcal{V}\times (0,1] \to \mathcal{V}\times (0,1] ,\ 
	(u,s)\mapsto (x,s) = \left((s\Pi_0 + s^2\Pi_-)u,\ s\right) \equiv \left( s u_0 + s^2 u_-,\ s \right) \label{eq:proof_general_transformation}
\end{equation}
where we used the shorthand $u_0 = \Pi_0 u$ and $u_- = \Pi_- u$.
The inverse is $u=s^{-1}\Pi_0x+s^{-2}\Pi_-x$ for $s>0$. The two powers have a specific purpose: on the Perron subspace the linear growth is of order $s$, so it balances the quadratic competition at $x=O(s)$. On the complementary subspace the linear operator remains invertible at $s=0$, so the quadratic forcing produces a correction of order $s^2$. The exhaustion argument below will verify these orders for every small non-negative equilibrium. Because $F(x)-F(0)$ is exactly linear in $x$, the following expansion has no omitted Taylor remainder:
\begin{align}
	0 &= F(0;s)x + (F(x)-F(0))x \nonumber\\
	&= s F(0;s) u_0 + s^2 F(0;s)u_- + s^2 (F(u_0+s u_-)-F(0))\cdot (u_0+s u_-)\label{eq:proof_general_intermed4}
\end{align}
Defining $\mathcal{Z}:=\mathrm{span}\{\varpi_n\}_{n\in\mathfrak{N}} = \mathrm{image}(\Pi_0)$ and the Hadamard-product $\odot_\varpi$ with respect to the basis $\{\varpi_n\}$ as the bilinear map $$\mathcal{Z}\times\mathcal{Z}\to\mathcal{Z},\quad \varpi_n \odot_\varpi \varpi_{n'} := \delta_{n n'} \varpi_n$$
we may express
$F(0;s)u_0 = s \zeta_0\odot_\varpi u_0$, with $\zeta_0:=\sum_n z_0^n \varpi_n$. We substitute this into (\ref{eq:proof_general_intermed4}) and  divide by $s^2$ to obtain
\begin{equation}\label{eq:proof_general_transformed_equation}
	0 = \zeta_0\odot_\varpi u_0 + F(0;s)u_- + (F(u_0+s u_-)-F(0))\cdot (u_0+s u_-) =: G(u;s)
\end{equation}
The expression $G(u;s)$ is polynomial in $u$ and $s$ and therefore extends smoothly to $s=0$, although the inverse coordinate change does not. This smooth extension is the equation to which the implicit function theorem is applied. We are going to apply the implicit function theorem to construct a curve $s\mapsto u(s)$ solving
\begin{equation}\label{eq:proof_general_IFT_equation}
	G(u(s);s) = 0
\end{equation}
To that end, we first need to find a solution $u(0)$ of 
\begin{align}
	0 &= G(u(0);0) \nonumber\\
	&= \zeta_0\odot_\varpi u(0)_0 + F(0;0) u(0)_- + (F(u(0)_0)-F(0))u(0)_0  \label{eq:proof_general_transformed_equation_at0}
\end{align}
We multiply this by $\Pi_0$ to get an equation for $\zeta:= u(0)_0 \in \mathcal{Z}$, namely
\begin{equation}\label{eq:proof_general_intermed2}
	0 = \zeta_0\odot_\varpi \zeta + \Pi_0 (F(\zeta) - F(0)) \zeta
\end{equation}
The second summand can be rewritten as
\begin{align*}
	\Pi_0 (F(\zeta)-F(0)) \zeta &= \sum_n (\varpi^n \zeta) \Pi_0 (F(\zeta)-F(0)) \varpi_n \\
	&= \zeta \odot_\varpi \Pi_0 (F(\zeta)-F(0)) \sum_{n'} \varpi_{n'} \\
	&= - \zeta \odot_\varpi W \zeta
\end{align*}
with the linear map
$$W:\ \mathcal{Z}\to\mathcal{Z},\quad W\zeta := -\Pi_0 \cdot (F(\zeta)-F(0)) \cdot \sum_n \varpi_n$$
For later use, define the non-negative, matrix-valued linear maps
\begin{equation}\label{eq:competition_operator}
	R_{nm}(u):=[F(0)-F(u)]_n^n
	\qquad(n,m\in\mathfrak N,\ [u]_m=u).
\end{equation}
Thus the scalar entry of $W$ in row $n$ and column $m$ is
\[
[W]^n_m=\varpi^n\cdot R_{nm}(\varpi_m)\varpi_n.
\]
It measures the projected competition exerted by the Perron composition of species $m$ on that of species $n$. Below, $W$ denotes either the operator on $\mathcal Z$ or its matrix in the basis $\{\varpi_n\}$, according to context. Correspondingly, $\zeta_\mathfrak n$ in a principal system is its coefficient vector; $[\zeta_\mathfrak n]_\mathfrak n$ denotes $\sum_{n\in\mathfrak n}(\zeta_\mathfrak n)_n\varpi_n$ when an element of $\mathcal Z$ is required. Positivity of that coefficient vector means positivity on its retained blocks, with zeros elsewhere. Substituting this into (\ref{eq:proof_general_intermed2}), we get
\begin{equation}\label{eq:proof_general_cone_equation}
	0 = \zeta\odot_\varpi \left(\zeta_0 - W \zeta\right)
\end{equation}
Note that (\ref{eq:proof_general_cone_equation}) always has ``monodominance solutions'' with $\zeta$ of the form $\zeta^n\varpi_n$ for a single $n\in\mathfrak{N}$, but we'll shortly see that there may be more. 
At any rate, substituting a given solution $\zeta$ into (\ref{eq:proof_general_transformed_equation_at0}) and applying the projector $\Pi_-$ yields
\begin{equation}\label{eq:proof_general_uminus}
	0 = F(0;0) u(0)_- + \Pi_-\cdot (F(\zeta)-F(0))\cdot \zeta
\end{equation}
Observe that $F(0;0)$ acts as an isomorphism on the image of $\Pi_-$ because, by construction of $\Pi_-$, this image is spanned by generalized Eigenspaces of $F(0;0)$ whose Eigenvalues have strictly negative real parts. Therefore (\ref{eq:proof_general_uminus}) determines a unique vector $u(0)_-$, implying that any solution $u(0)$ of (\ref{eq:proof_general_transformed_equation_at0}) is entirely determined by the respective solution $\zeta= u(0)_0$ of (\ref{eq:proof_general_cone_equation}). We classify solutions according to their \emph{non-vanishing blocks} $\mathfrak{n}\subset\mathfrak{N}$, defined as the smallest subset such that $[\zeta]_\mathfrak{n} = \zeta$.
In order to study solutions of (\ref{eq:proof_general_cone_equation}), we first investigate generic solutions of \begin{equation}\label{eq:proof_general_principal_system}
	[W]_\mathfrak{n}^\mathfrak{n}\zeta_\mathfrak{n}=(z_0^n)_{n\in\mathfrak{n}}
\end{equation}
for $\varnothing\neq\mathfrak n\subseteq\mathfrak N$. By genericity (cf. Appendix \ref{app:genericity}) we mean in particular the following two properties:
\begin{enumerate}[label={Property \arabic*:},ref=\arabic*]
	\item\label{property:W_principal} Every principal submatrix of $W$ (with respect to the basis $\{\varpi_n\}_{n\in\mathfrak{N}}$) is nonsingular.
	\item\label{property:positive_bifurcation} For every non-empty $\mathfrak n\subset\mathfrak N$, the solution $\zeta_\mathfrak n$ of (\ref{eq:proof_general_principal_system}) has no zero component, and for every $n'\notin\mathfrak n$ its transverse coefficient satisfies
	\[
		z_0^{n'}-\varpi^{n'}\cdot W[\zeta_\mathfrak n]_\mathfrak n\neq0.
	\]
\end{enumerate}
We now demonstrate that these are actually generic on the parameter space of Table \ref{tbl:mechanisms}, with the allowed establishment probabilities free to vary. Fixing dispersal and migration graphs does not impose additional equalities on those probabilities. If an extra model restriction is imposed, such as $P=0$ together with a single macrosite, genericity must be checked again within that restricted family; it does not follow from genericity in $\mathcal Q$. Fix a parameter subset on which the directed graphs in Assumption \ref{ass:block_irreducibility} are fixed and strongly connected (cf. Appendix \ref{app:genericity}), henceforth called a parameter stratum. Now use relative coordinates in the affine hull of this stratum. The simple dominant Eigenvalues and normalized Perron Eigenvectors depend real-analytically on the parameters and are algebraic over their rational function field. Positivity of the left and right Eigenvectors selects the Perron branch. A nonzero algebraic expression on this branch has a minimal polynomial with nonzero constant coefficient. After clearing denominators, the resulting polynomial identity shows that its zero set is contained in the zero set of a nonzero polynomial in the parameters. The identity extends to admissible boundary points by continuity. Apply this argument to principal determinants, Cramer numerators, and transverse residuals multiplied by the relevant principal determinant. These expressions are analytic even where that determinant vanishes. The following parameter choices show that said expressions are not identically zero. To show that none of the principal minors of $W$ vanishes identically on a fixed dispersal and migration stratum, we construct an admissible parameter choice for which all these minors are non-zero. Recall that the displacement probabilities can be chosen independently of the fixed dispersal and migration graphs. Order the species, take $p_0^i=1$, and choose, within each macrosite, $p^i_{i'}=1$ when $N(i)>N(i')$ and $p^i_{i'}=0$ when $N(i)<N(i')$, retaining zero entries between different macrosites. Then $R_{nm}=0$ whenever $n>m$, whereas
\[
	[W]^n_n=\varpi^n\cdot R_{nn}(\varpi_n)\varpi_n>0
\]
because $p_0^i>0$, $K\varpi_n>0$, and both Perron vectors are positive. Hence $W$ and all its principal submatrices are triangular with non-zero diagonal. A sufficiently small perturbation moves the displacement probabilities into the relative interior of their probability constraints without changing the fixed dispersal/migration stratum or destroying the principal minors. Thus the failure of Property 1 is contained in a finite union of proper algebraic subsets on each stratum.

Suppose now that Property 1 holds. For a fixed non-empty $\mathfrak{n}\subset\mathfrak{N}$, the solution of (\ref{eq:proof_general_principal_system}) has a zero component precisely when the corresponding Cramer numerator vanishes. None of these numerators is a structural identity: adding an independently chosen scalar multiple of the identity to each species block of $C_1$ shifts the respective $z_0^n$ independently while leaving the Perron Eigenvectors, $F(x)-F(0)$, and hence $W$ unchanged. The perturbation may be chosen away from the resulting finite collection of hyperplanes. Applying the same observation to $\mathfrak{n}\cup\{n'\}$ shows generically that
\begin{equation}\label{eq:proof_general_invasion_nondegeneracy}
	z_0^{n'}-\varpi^{n'}\cdot W[\zeta_\mathfrak n]_\mathfrak n\neq0
	\qquad(n'\notin\mathfrak{n}).
\end{equation}
Hence Properties 1 and 2 hold outside a finite union of proper algebraic subsets, i.e. they are generic.

We now perform the coexistence chamber construction on the generic supercritical locus. For this we remove the internal threshold hypersurfaces on which a block-dominant Eigenvalue vanishes while another one is positive, i.e.
\begin{equation}\label{eq:proof_general_generic_supercritical}
	\mathcal G_+:=\left\{q\in\mathcal Q_+\ \middle|\
	\begin{array}{l}
		z_0^n\neq0\text{ for every }n\in\mathfrak N,\text{ and}\\
		\text{Properties \ref{property:W_principal} and
		\ref{property:positive_bifurcation} hold after all blocks}\\
		\text{outside }\mathfrak N_+(q)\text{ have been discarded}
	\end{array}\right\}.
\end{equation}
The additional conditions $z_0^n\neq0$ are generic, because adding a scalar multiple of the identity to the principal $n$-block of $C_1$ shifts $z_0^n$ by the opposite scalar and leaves the other block-dominant Eigenvalues unchanged. Together with the preceding genericity argument, this shows that $\mathcal G_+$ is a dense relatively open subset of $\mathcal Q_+$, in the stratum-wise sense of Appendix \ref{app:genericity}. The polynomial Eigenvector equations, positivity, and normalizations specify a map from parameters to the selected Perron data, such that the graph of this map is a semi-algebraic set. Adding the finitely many conditions defining $\mathcal G_+$ and projecting to the model parameters therefore proves semi-algebraicity by the Tarski--Seidenberg theorem. There are only finitely many zero patterns. Consequently $\mathcal G_+$ has finitely many connected components, which are relatively open because $\mathcal G_+$ is relatively open in the parameter domain under consideration. We call these connected components $\mathcal H$ of $\mathcal G_+$ the \emph{generic coexistence chambers}. Note that $\mathfrak C(q)$ is constant on every $\mathcal H$ due to the following argument: On $\mathcal G_+$, the set $\mathfrak N_+(q)$ is locally constant, and all components of every principal-system solution depend continuously on $q$ due to Property \ref{property:W_principal} and Cramer's rule. Furthermore, Property \ref{property:positive_bifurcation} keeps those components and all transverse coefficients away from zero locally.

We next show that $u(0)$ can be deformed into a curve of solutions. It is enough to prove that $\dd G(u(0);0)$ is invertible. For a solution of (\ref{eq:proof_general_cone_equation}) with non-vanishing blocks $\mathfrak{n}$, let $h\in\ker\dd G(u(0);0)$ and write $\eta:=h_0$. Applying $\Pi_0$ to $\dd G(u(0);0)h=0$ gives
\begin{align}
	0
	&=[\eta]_{\mathfrak{N}\setminus\mathfrak{n}}
	  \odot_\varpi[\zeta_0-W\zeta]_{\mathfrak{N}\setminus\mathfrak{n}}
	  -[\zeta]_\mathfrak{n}\odot_\varpi[W\eta]_\mathfrak{n}.
	\label{eq:proof_general_intermed6}
\end{align}
The two summands have disjoint block supports. Equation (\ref{eq:proof_general_invasion_nondegeneracy}) first yields $[\eta]_{\mathfrak{N}\setminus\mathfrak{n}}=0$. Since every component of $[\zeta]_\mathfrak{n}$ is non-zero, (\ref{eq:proof_general_intermed6}) then reduces to
\[
	[W]_\mathfrak{n}^\mathfrak{n}[\eta]_\mathfrak{n}=0.
\]
Property 1 gives $\eta=0$. The $\Pi_-$-projection of $\dd G(u(0);0)h=0$ then gives $h_-=0$, because $F(0;0)$ is invertible on $\operatorname{image}(\Pi_-)$. Therefore $h=0$ and $\dd G(u(0);0)$ is invertible. The implicit function theorem supplies $\varepsilon>0$ and a unique smooth solution $u:[0,\varepsilon)\to\mathcal{V}$ of (\ref{eq:proof_general_IFT_equation}). The block support is preserved exactly, rather than merely to leading order. On the $\mathfrak n$-block subspace, the equations in all excluded blocks vanish identically, since $F(x;s)$ is block diagonal. The same restricted implicit-function argument gives a solution in that subspace, which must coincide with the full solution by uniqueness. On every retained block a positive leading coefficient times $\varpi_n>0$ dominates the order-$s^2$ correction. A negative coefficient similarly forces negative coordinates for small $s$. Transformation (\ref{eq:proof_general_transformation}) thus produces a locally unique branch
\begin{equation}\label{eq:proof_general_local_branch}
	x(s)=s\zeta+O(s^2),\qquad s\downarrow0.
\end{equation}
It lies in the relative interior of the positive-$\mathfrak{n}$-block cone exactly when $\zeta>0$, which is equivalent to $\mathfrak{n}\in\mathfrak{C}(q)$.

We now show that the non-negative equilibrium branches constructed by the implicit function theorem exhaust all non-negative equilibria when approaching the bifurcation along the fixed deformation ray. Suppose, to the contrary, that $s_k\downarrow0$ and that $x_k\neq0$ are equilibria not belonging to the branches just constructed. Since only finitely many block supports are possible, we may pass to a subsequence on which the set $\mathfrak m$ of positive blocks of $x_k$ is constant. Compactness of $\bar{\mathcal X}$ gives a convergent subsequence. Its limit solves the equation at $s=0$, whose only non-negative solution is zero by Appendix \ref{app:general_equilibrium_preliminaries}; hence $x_k\to0$.
For $m\in\mathfrak m$, set
\[
	b_{k,m}:=\varpi^m\cdot[x_k]_m,
	\qquad
	u_{k,m}:=\frac{[x_k]_m}{b_{k,m}}.
\]
Then $\varpi^m\cdot u_{k,m}=1$. Since
$[F(x_k;s_k)]^m_m u_{k,m}=0$ and
$[F(x_k;s_k)]^m_m\to[F(0;0)]^m_m$, every accumulation point of $u_{k,m}$ is the normalized non-negative nullvector of the irreducible matrix $[F(0;0)]^m_m$. Simplicity of its dominant Eigenvalue therefore gives
\begin{equation}\label{eq:proof_general_normalized_block_limit}
	u_{k,m}\longrightarrow\varpi_m
	\qquad(m\in\mathfrak m).
\end{equation}
Left multiplication of the $m$-block equilibrium equation by $\varpi^m$, followed by division by $b_{k,m}$, gives the exact balance
\begin{equation}\label{eq:proof_general_small_balance}
	s_kz_0^m
	=\sum_{\ell\in\mathfrak m}b_{k,\ell}\,
	\varpi^m\cdot R_{m\ell}(u_{k,\ell})u_{k,m}
	\qquad(m\in\mathfrak m).
\end{equation}
All summands on the right are non-negative. By
(\ref{eq:proof_general_normalized_block_limit}), its diagonal coefficient converges to
\[
	\varpi^m\cdot R_{mm}(\varpi_m)\varpi_m=[W]^m_m>0.
\]
Equation (\ref{eq:proof_general_small_balance}) consequently gives
$b_{k,m}=O(s_k)$ for every $m\in\mathfrak m$. Conversely, the coefficients in that equation remain bounded, whereas $z_0^m>0$; hence at least one of the quotients $b_{k,m}/s_k$ is bounded away from zero.
After another subsequence, write
\[
	\frac{b_{k,m}}{s_k}\longrightarrow a_m\geq0
	\qquad(m\in\mathfrak m)
\]
and let $\mathfrak r:=\{m\in\mathfrak m\mid a_m>0\}$, which is non-empty by the preceding lower bound. Passing to the limit in (\ref{eq:proof_general_small_balance}) yields
\begin{equation}\label{eq:proof_general_small_limit_system}
	z_0^m=\sum_{\ell\in\mathfrak r}[W]^m_\ell a_\ell
	\qquad(m\in\mathfrak m).
\end{equation}
For $m\in\mathfrak r$, this is the principal system
$[W]^\mathfrak r_\mathfrak r(a_m)_{m\in\mathfrak r}
=(z_0^m)_{m\in\mathfrak r}$, so Property \ref{property:W_principal} identifies $(a_m)_{m\in\mathfrak r}$ with $\zeta_\mathfrak r$. If $\mathfrak r\subsetneq\mathfrak m$, equation (\ref{eq:proof_general_small_limit_system}) for any $m\in\mathfrak m\setminus\mathfrak r$ would say that the corresponding transverse coefficient
\[
	z_0^m-\varpi^m\cdot W[\zeta_\mathfrak r]_\mathfrak r
\]
vanishes, contrary to Property \ref{property:positive_bifurcation}. Thus $\mathfrak r=\mathfrak m$, the solution $\zeta_\mathfrak m$ is strictly positive, and $\mathfrak m\in\mathfrak C(q)$.
We have also shown that $x_k=O(s_k)$. Projecting its equilibrium equation onto $\operatorname{image}(\Pi_-)$ gives
\[
	F(0;0)\Pi_-x_k
	=-\Pi_-\left(
	s_k\sum_{n\in\mathfrak N}z_0^n[\One]^n_nx_k
	+(F(x_k)-F(0))x_k\right)
	=O(s_k^2).
\]
The restriction of $F(0;0)$ to $\operatorname{image}(\Pi_-)$ is invertible, so $\Pi_-x_k=O(s_k^2)$. Hence the transformed points
\[
	s_k^{-1}\Pi_0x_k+s_k^{-2}\Pi_-x_k
\]
are bounded. Every accumulation point solves $G(u;0)=0$, and the preceding argument fixes its center component as the unique positive solution $\zeta_\mathfrak m$; equation (\ref{eq:proof_general_uminus}) then fixes its remaining component. The local uniqueness supplied by the implicit function theorem therefore forces $x_k$ onto the already constructed $\mathfrak m$-branch for all sufficiently large $k$, a contradiction. Since there are only finitely many block supports, one common $\varepsilon>0$ excludes all additional non-negative equilibria. For completeness, this also proves the two assertions about $\mathfrak C(q)$ in the proposition. In the basis $\{\varpi_n\}$, the rays $\pi_\mathfrak n(\psi_m)$ are precisely the columns of $[W]^\mathfrak n_\mathfrak n$. Property \ref{property:W_principal} therefore makes their cone full-dimensional in $[\mathcal Z]_\mathfrak n$, and membership of $[\zeta_0]_\mathfrak n$ in its relative interior is equivalent to strict positivity of the unique solution of (\ref{eq:proof_general_principal_system}). Moreover, all entries of $W$ are non-negative and
\[
	[W]^n_n=\varpi^n\cdot R_{nn}(\varpi_n)\varpi_n>0.
\]
Consequently, every coordinate of $[\zeta_0]_\mathfrak n$ is positive for $\mathfrak n\in\mathfrak C(q)$, which proves
\[
	\mathfrak C(q)\subset
	\set{\mathfrak n\subset\mathfrak N}{z_0^n>0\text{ for all }n\in\mathfrak n}.
\]
Conversely, if $z_0^n>0$, the singleton principal system has the positive solution $z_0^n/[W]^n_n$. Thus $\{n\}\in\mathfrak C(q)$, and $\mathfrak C(q)$ is non-empty whenever $\mathfrak N_+(q)$ is non-empty. Because only finitely many coalitions are involved, the individual implicit-function neighborhoods have a common $\varepsilon\in(0,1]$. For every $s\in(0,\varepsilon)$ the dominant Eigenvalue of $[F(0;s)]_n^n$ is $s z_0^n>0$, so the zero equilibrium is unstable. 

Finally, the stationary propagule equation assigns to every local sessile branch the unique vector
\[
	y_\mathfrak n(s)=(C_2+(\One-\tilde D)C_4+C_5)^{-1}DC_3x_\mathfrak n(s).
\]
It is non-negative and satisfies $C_5y_\mathfrak n(s)=Kx_\mathfrak n(s)$. For each retained species, $z_0^n>0$ and $P_0K\varpi_n=(C_1+z_0^n\One)\varpi_n>0$ on its block. Hence that block of $K$ has no zero row. It maps every strictly positive block vector to a strictly positive vector. Thus $y_\mathfrak n(s)$ has exactly the same positive species blocks as $x_\mathfrak n(s)$. For sufficiently small $s$, the latter also belongs to $\bar{\mathcal X}$ by (\ref{eq:proof_general_local_branch}). This proves Proposition \ref{prop:critical_mortality}.
\end{proof}

\emph{Proof of the Theorem:} For $q\in\mathcal Q_-\cup\mathcal Q_0$, Appendix
\ref{app:general_equilibrium_preliminaries} gives both the uniqueness of the
zero equilibrium and its stability; equation (\ref{eq:positive_blocks}) also
gives $\mathfrak C(q)=\varnothing$.  Conversely, suppose that $z_0^n>0$ for some
$n$ and restrict the model to this species. Since $[W]^n_n>0$, Properties
\ref{property:W_principal} and \ref{property:positive_bifurcation} hold for this
one-species restriction. Proposition \ref{prop:critical_mortality} therefore
gives a positive equilibrium branch near $s=0$. This branch, and indeed any
positive equilibria in the same species face, cannot approach zero at a
deformation parameter $s>0$: their irreducible matrices $[F(x;s)]^n_n$ have
dominant Eigenvalue zero, whereas $[F(0;s)]^n_n$ has dominant Eigenvalue $s
z_0^n>0$. Continuity rules out such a limit. Capacity saturation and loss of
individual coordinates within an occupied species block are excluded by the
preceding arguments. Proposition \ref{prop:general_continuation_no_collision},
applied to the one-species restriction, then gives a positive equilibrium at
$s=1$. Extending it by zero on the other species blocks proves the asserted
existence for every $q\in\mathcal Q_+$.

It remains to construct the generic near-threshold cover from the ray-wise
statement of Proposition \ref{prop:critical_mortality}. We first show that the
critical-mortality rays remain in one chamber: For $q\in\mathcal G_+$, let
$q(s)$ denote the parameter point obtained by replacing $C_1$ with
\begin{equation}\label{eq:proof_general_mortality_ray}
	C_1(s):=C_1+(1-s)
	\sum_{n\in\mathfrak N_+(q)}z_0^n[\One]^n_n
\end{equation}
and leaving every other parameter unchanged.  Then the sessile equilibrium
matrix belonging to $q(s)$ is exactly $F(x;s)$ in
(\ref{eq:deformed_equilibria}).  For $s>0$ its positive block-dominant
Eigenvalues are $s z_0^n$; the negative ones, the Perron Eigenvectors, and
$F(x;s)-F(0;s)=F(x)-F(0)$ are unchanged.  Hence $W$ is unchanged, every
principal-system solution is multiplied by $s$, and every transverse coefficient
in (\ref{eq:proof_general_invasion_nondegeneracy}) is multiplied by $s$. It
follows that
\begin{equation}\label{eq:proof_general_ray_in_chamber}
	q(s)\in\mathcal G_+,
	\qquad
	\mathfrak C(q(s))=\mathfrak C(q)
	\qquad(0<s\leq1).
\end{equation}
The path $s\mapsto q(s)$ therefore stays in the generic coexistence chamber containing $q$.  Its endpoint $q(0)$ belongs to $\mathcal Q_0$.

Proposition \ref{prop:critical_mortality} says that every ray
(\ref{eq:proof_general_mortality_ray}) consists, for all sufficiently small $s>0$, of points at which the required equilibrium classification is exact. We next verify that each such point has an open parameter neighborhood on which the classification remains exact.
Fix $q\in\mathcal G_+$ and choose $s>0$ sufficiently small for Proposition \ref{prop:critical_mortality} to apply.  There are then only finitely many non-zero non-negative sessile equilibria, exactly one for each $\mathfrak n\in\mathfrak C(q)$.  The derivative of the transformed equation $G(u;s)=0$ at each corresponding solution is invertible for small $s$, by the invertibility at $s=0$ proved in (\ref{eq:proof_general_intermed6}) and continuity. Since the coordinate change (\ref{eq:proof_general_transformation}) is invertible for $s>0$, the original sessile equilibrium is non-degenerate on the retained variables as well. The discarded species have strictly negative $z_0^n$ because $q\in\mathcal G_+$. Their diagonal Jacobian blocks are $[F(x;s)]^n_n$, and $[F(x;s)]^n_n\varpi_n\leq z_0^n\varpi_n<0$ on the block. They are therefore Hurwitz. Ordering retained species first makes the full Jacobian block upper triangular, since the discarded species are absent. Thus non-degeneracy holds on the entire sessile space.
The zero equilibrium is also non-degenerate: at $s=0$ every non-dominant Eigenvalue of a retained block has strictly negative real part, while every discarded block is Hurwitz; for $s>0$ sufficiently small the formerly critical Eigenvalues are positive and no other Eigenvalue can reach zero. The parameter-dependent implicit function theorem therefore continues the zero equilibrium and every listed non-zero equilibrium uniquely under small changes of the model parameters.  The block-coordinate faces are invariant for the equilibrium equation, strict positivity of the retained blocks and the capacity inequalities persist, and the signs defining $\mathfrak C(q)$ are constant after the neighborhood has been restricted to the same generic coexistence chamber.  Finally, no unlisted non-negative equilibrium can appear under an arbitrarily small parameter change.  Indeed, the sessile state spaces for parameters in a sufficiently small neighborhood are contained in one fixed compact box.  A hypothetical sequence of new equilibria would consequently have a convergent subsequence, and its limit would be either the zero equilibrium or one of the finitely many listed equilibria, contradicting the corresponding local uniqueness. The propagule component is uniquely determined by the sessile component through (\ref{eq:stationary_y}), so it introduces no additional stationary branches.

Let $\mathcal U\subset\mathcal G_+$ be the union of all parameter
neighborhoods obtained in this way, over all $q\in\mathcal G_+$ and all $s>0$ sufficiently small along their critical-mortality rays.  By construction $\mathcal U$ is open in $\mathcal Q_+$, and for every $q\in\mathcal U$ the zero equilibrium is unstable and the non-zero non-negative equilibria are in one-to-one correspondence with $\mathfrak C(q)$, with precisely the asserted block supports.
We now show that $\mathcal{U}$ is a generic near-threshold cover of the extinction threshold: Let $\mathcal H$ be a generic coexistence chamber and choose any $q\in\mathcal H$. Equations (\ref{eq:proof_general_ray_in_chamber}) and Proposition \ref{prop:critical_mortality} show that $q(s)\in\mathcal U\cap \mathcal H$ for every sufficiently small $s>0$, whereas $q(s)\to q(0)\in\mathcal Q_0$.  Consequently,
\begin{equation}\label{eq:proof_general_each_chamber_reaches_threshold}
	\operatorname{cl}_{\mathcal Q}(\mathcal U\cap\mathcal H)
	\cap\mathcal Q_0\neq\varnothing
\end{equation}
for every generic coexistence chamber. It remains to show that these chamber traces cover all of $\mathcal Q_0$.
First observe that every $q_0\in\mathcal Q_0$ lies in the closure of $\mathcal G_+$.  Choose a block with $z_0^n(q_0)=0$ and decrease all mortality entries in that block by the same sufficiently small amount.  Its dominant Eigenvalue becomes positive, so this produces supercritical parameter points arbitrarily close to $q_0$.  The model parameters can then be perturbed by an arbitrarily smaller amount so as to avoid the finitely many algebraic exceptional sets in Properties \ref{property:W_principal} and \ref{property:positive_bifurcation}, as well as all hypersurfaces $z_0^m=0$.  Hence there is a sequence $q_k\in\mathcal G_+$ with $q_k\to q_0$.
For each $q_k$, take a sufficiently small point $q_k(s_k)$ on its critical-mortality ray such that $q_k(s_k)\in\mathcal U$ and
\begin{equation}\label{eq:proof_general_covering_sequence}
	\lVert q_k(s_k)-q_k(0)\rVert<\frac1k.
\end{equation}
The modification from $q_k$ to $q_k(0)$ is proportional to the positive parts of the block-dominant Eigenvalues.  Since all block-dominant Eigenvalues at $q_0$ are non-positive, continuity gives $q_k(0)\to q_0$, and therefore $q_k(s_k)\to q_0$.  There are only finitely many generic coexistence chambers, so a subsequence of $q_k(s_k)$ lies in one fixed chamber $\mathcal H$. Thus
\[
	q_0\in\operatorname{cl}_{\mathcal Q}
	(\mathcal U\cap\mathcal H)\cap\mathcal Q_0.
\]
As $q_0\in\mathcal Q_0$ was arbitrary, we have proved
\begin{equation}\label{eq:proof_general_threshold_cover}
	\mathcal Q_0=
	\bigcup_{\mathcal H}
	\left(
	\operatorname{cl}_{\mathcal Q}(\mathcal U\cap\mathcal H)
	\cap\mathcal Q_0
	\right),
\end{equation}
where the union ranges over all generic coexistence chambers.  Together with (\ref{eq:proof_general_each_chamber_reaches_threshold}), this says exactly that $\mathcal U$ is a generic near-threshold cover and completes the proof of Theorem \ref{thm:near-threshold}.

\begin{rem}\label{rem:sectorial_near_threshold}
Note that the use of finiteness is essential to the stated covering identity: the approximating sequence has a subsequence in one chamber, so its limit lies in the closure of that chamber's intersection with $\mathcal U$. The threshold point itself doesn't need to be generic nor is $\mathcal U$ required to contain an ordinary neighborhood of it, and different chambers can have overlapping traces on the threshold. Accordingly, the cover just constructed need not have a uniform radial thickness near $\mathcal Q_0$; its thickness may shrink when a generic direction approaches one of the removed algebraic degeneracy loci. A stronger local statement is available on one-sided sectors. It can be devised along the following lines:  Fix a critical parameter point, normalize the outward vector formed by the block-dominant Eigenvalues that become positive, and restrict this normalized vector to a compact set of directions. If every relevant principal determinant of $W$, every component of the corresponding normalized principal-system solutions, and every transverse coefficient is bounded away from zero throughout that compact set, the implicit-function neighborhoods in the blow-up construction can be chosen uniformly. The exhaustion argument is uniform for the same reason: the positive diagonal coefficients of $W$, the stable spectral gaps, and all non-degeneracy bounds have positive minima on the compact direction set.  Consequently, a single radial $\varepsilon>0$ gives the exact coalition classification throughout the resulting sector. An entire ordinary neighborhood of the sector's vertex cannot in general be used, because degenerate directions may approach the same critical parameter point.
\end{rem}

\subsection{Proof of Theorem \ref{thm:general_case}}\label{app:general_proof}
The central insight here is that the Theorem's conclusion remains valid on any face, which contains a near-threshold branch of equilibria, if this branch is prevented from hitting a proper subface. The proof therefore proceeds in three steps. First, we construct a parameter subset $\mathcal{Q}_1$ which prevents such face collisions. Second, we show in Proposition \ref{prop:general_continuation_no_collision} below that no face collisions generically imply the conclusion of the theorem. Notably, Proposition \ref{prop:general_continuation_no_collision} doesn't rely on $q\in\mathcal{Q}_1$ but could be used under alternative parameter constraints as long as they prevent face collisions. Third, we will lift the genericity constraint.

\emph{Step 1, construction of $\mathcal{Q}_1$:} Fix $q\in\mathcal Q_+$, discard the blocks outside $\mathfrak N_+(q)$ as justified in Appendix \ref{app:general_equilibrium_preliminaries}, and consider the critical-mortality deformation (\ref{eq:deformed_equilibria}). The dominant Eigenvalues along this deformation are $s z_0^n$, while the Perron Eigenvectors and the difference $F(x;s)-F(0;s)=F(x)-F(0)$ are independent of $s$. Consequently, $W$ is independent of $s$, and the principal system (\ref{eq:proof_general_principal_system}) at deformation parameter $s>0$ is obtained by multiplying its right-hand side by $s$. Multiplication by $s$ gives a bijection between its solution sets at $s=1$ and at $s>0$, even when a principal matrix is singular. Thus its strictly positive solutions, and hence $\mathfrak C(q)$, remain unchanged up to scaling for $s\in(0,1]$.

Under the generic hypotheses of Proposition \ref{prop:critical_mortality}, the local positive-$\mathfrak n$-block equilibrium is unique and has non-zero degree. Proposition \ref{prop:general_continuation_no_collision} below shows that this degree persists until $s=1$ if the positive equilibrium set cannot approach a proper face. The following parameter conditions provide one way of excluding that escape:
Writing $F_n(x;s):=[F(x;s)]_n^n$, affinity of $F$ gives
\begin{equation}\label{eq:proof_general_competition_decomposition}
	F_n(x;s)=F_n(0;s)-
	\sum_{m\in\mathfrak N}R_{nm}([x]_m)
	\qquad(n\in\mathfrak N).
\end{equation}
For $m\in\mathfrak N$ put
\begin{equation}\label{eq:proof_general_block_simplex}
	\Delta_m:=\left\{u\geq0\ \middle|\ [u]_m=u,
	\ \varpi^m\cdot u=1\right\},
\end{equation}
whose vertices are $\ee_j/\varpi^m_j$, $j\in N^{-1}(m)$.
Positivity of every entry of $\varpi^m$ makes this a compact simplex.
We define the least self-competition over it by
\begin{equation}\label{eq:self_competition}
	\underline c_m:=\min_{u\in\Delta_m}
	\varpi^m\cdot R_{mm}(u)u.
\end{equation}
The minimum is attained. Put $\mathfrak I_m:=N^{-1}(m)$. Then
\begin{equation}\label{eq:proof_general_finite_self_test}
	\underline c_m>0
	\quad\Longleftrightarrow\quad
	\varpi^m\cdot R_{mm}(\ee_j)\ee_j>0
	\quad\text{for every }j\in\mathfrak I_m.
\end{equation}
Indeed, expanding $\varpi^m\cdot R_{mm}(u)u$ gives a quadratic polynomial with non-negative coefficients. The implication from left to right follows by taking $u=\ee_j/(\varpi^m\cdot\ee_j)$. Conversely, its diagonal terms give the explicit lower bound
\begin{equation}\label{eq:proof_general_c_lower_bound}
	\underline c_m\geq\underline c_m^{\mathrm{diag}}
	:=\left(
	\sum_{j\in\mathfrak I_m}
	\frac{(\varpi^m_j)^2}
	{\varpi^m\cdot R_{mm}(\ee_j)\ee_j}
	\right)^{-1}>0.
\end{equation}
The bound is the exact minimum of the diagonal part of the quadratic polynomial, obtained by Cauchy--Schwarz under the constraint $\sum_j\varpi^m_j u^j=1$. Thus positivity of $\underline c_m$ is decided by finitely many model-parameter expressions.

This condition is stronger than irreducibility. The latter controls positive compositions, whereas the minimum also tests compositions concentrated at a single coordinate. A simple sufficient parameter condition is $K^j_j>0$ for every $j\in\mathfrak I_m$, since the corresponding diagonal coefficient of the quadratic polynomial contains the positive term $\varpi^m_j p_0^j K^j_j/\bar S^{B(j)}$. The finite test (\ref{eq:proof_general_finite_self_test}) can also hold through other within-macrosite contributions when this sufficient diagonal condition does not.

We next retain the separate cross-competition bounds as an inexpensive first level of the face-wise certificate. For $n\neq m$ with $z_0^n,z_0^m>0$ and
$\underline c_m>0$, let $\gamma_{nm}$ be the infimum
of all $\gamma\geq0$ such that
\begin{equation}\label{eq:cross_competition}
	R_{nm}(u)\varpi_n
	\leq\gamma\bigl(\varpi^m\cdot R_{mm}(u)u\bigr)\varpi_n
	\qquad(u\in\Delta_m).
\end{equation}
Here and below inequalities between vectors are understood componentwise. Define
\begin{equation}\label{eq:proof_general_cross_upper_bound}
	\bar c_{nm}:=\max_{\substack{i\in\mathfrak I_n\\j\in\mathfrak I_m}}
	\frac{[R_{nm}(\ee_j)\varpi_n]^i}
	{\varpi_n^i\varpi^m_j},
	\qquad
	\widehat\gamma_{nm}:=\frac{\bar c_{nm}}
	{\underline c_m^{\mathrm{diag}}}.
\end{equation}
Linearity and non-negativity give $R_{nm}(u)\varpi_n\leq\bar c_{nm}\varpi_n$ for every $u\in\Delta_m$, while (\ref{eq:proof_general_c_lower_bound}) bounds the denominator in (\ref{eq:cross_competition}). Consequently,
\begin{equation}\label{eq:proof_general_gamma_hierarchy}
	\gamma_{nm}\leq\widehat\gamma_{nm}.
\end{equation}
More explicitly, every $u\in\Delta_m$ is the convex combination
\[
u=\sum_{j\in\mathfrak I_m}(\varpi^m_j u^j)
\frac{\ee_j}{\varpi^m_j},
\qquad \sum_j\varpi^m_j u^j=1.
\]
Each numerator ratio is linear in this combination, so its maximum is attained at a vertex. This explains the finite maximum defining $\bar c_{nm}$. The ratio defining $\gamma_{nm}$ also has a quadratic, composition-dependent denominator, so that vertex reduction does not in general compute its exact value. Certification nevertheless needs only a rigorous upper bound: a proposed $\gamma$ can be checked by the polynomial inequalities (\ref{eq:cross_competition}) throughout the compact simplex, without locating the exact maximizer.

The sharper test uses the resident equations jointly. Let $x$ be a non-negative solution of (\ref{eq:proof_general_deformed}) at some $s>0$, and let $\mathfrak m$ be its set of non-vanishing blocks. For $m\in\mathfrak m$ put
\begin{equation}\label{eq:proof_general_scaled_residents}
	a_m:=\frac{\varpi^m\cdot[x]_m}{s},
	\qquad
	u_m:=\frac{[x]_m}{s a_m},
\end{equation}
so that $u_m\in\Delta_m$ and $[x]_m=s a_m u_m$. Multiplying the $m$-block equilibrium equation by $\varpi^m$, using (\ref{eq:proof_general_competition_decomposition}), and dividing by $s^2a_m$ gives the exact projected balance
\begin{equation}\label{eq:proof_general_resident_balance}
	z_0^m=
	\sum_{k\in\mathfrak m}
	a_k\,\varpi^m\cdot R_{mk}(u_k)u_m
	\qquad(m\in\mathfrak m).
\end{equation}
In particular,
\begin{equation}\label{eq:proof_general_amplitude_bound2}
	a_m\,\varpi^m\cdot R_{mm}(u_m)u_m\leq z_0^m,
	\qquad
	a_m\leq\frac{z_0^m}{\underline c_m}.
\end{equation}

For every non-empty $\mathfrak m\subset\mathfrak N_+(q)$, let $\mathcal A_\mathfrak m$ be the set of all tuples $(a_m,u_m)_{m\in\mathfrak m}$ satisfying
\begin{equation}\label{eq:proof_general_resident_set}
	a_m\geq0,\qquad u_m\in\Delta_m,
	\qquad
	z_0^m=\sum_{k\in\mathfrak m}
	a_k\,\varpi^m\cdot R_{mk}(u_k)u_m
	\quad(m\in\mathfrak m).
\end{equation}
If $\underline c_m>0$ for $m\in\mathfrak m$, (\ref{eq:proof_general_amplitude_bound2}) shows that $\mathcal A_\mathfrak m$ is compact. It is an outer set: every equilibrium supported on $\mathfrak m$ determines an element of $\mathcal A_\mathfrak m$ through (\ref{eq:proof_general_scaled_residents}), although not every element need satisfy the unprojected equilibrium equations.

Allowing $a_m=0$ makes this outer set closed. At an actual equilibrium with support exactly $\mathfrak m$ all these amplitudes are positive; the additional zero-amplitude tuples merely make the certificate more conservative. Crucially, neither $s$ nor an unknown equilibrium appears in the definition of $\mathcal A_\mathfrak m$: division by $s^2a_m$ in the projected balances has removed the deformation parameter. One set of bounds therefore controls the whole interval $0<s\leq1$.

For $n\notin\mathfrak m$ and $\mathcal A_\mathfrak m\neq\varnothing$, define the upper and lower transverse competition bounds
\begin{align}
	\bar\rho_{n\mathfrak m}
	&:=\max_{\substack{(a,u)\in\mathcal A_\mathfrak m\\i\in\mathfrak I_n}}
	\frac{\left[\sum_{m\in\mathfrak m}
	a_mR_{nm}(u_m)\varpi_n\right]^i}{\varpi_n^i},
	\label{eq:proof_general_rho_upper}\\
	\underline\rho_{n\mathfrak m}
	&:=\min_{\substack{(a,u)\in\mathcal A_\mathfrak m\\i\in\mathfrak I_n}}
	\frac{\left[\sum_{m\in\mathfrak m}
	a_mR_{nm}(u_m)\varpi_n\right]^i}{\varpi_n^i}.
	\label{eq:proof_general_rho_lower}
\end{align}
For a locally admissible coalition $\mathfrak n\in\mathfrak C(q)$, the face-wise no-boundary condition is
\begin{equation}\label{eq:proof_general_FNB}
\begin{split}
	(\mathrm{FNB}_{\mathfrak n})\qquad
	&\underline c_m>0\quad\text{for every }m\in\mathfrak n,\quad\text{and}\\
	&\text{for every }\varnothing\neq\mathfrak m\subsetneq\mathfrak n,
	\text{ either }\mathcal A_\mathfrak m=\varnothing\text{ or there is an }
	n\in\mathfrak n\setminus\mathfrak m\text{ such that}\\[-2pt]
	&\hspace{38mm}
	\bar\rho_{n\mathfrak m}<z_0^n
	\quad\text{or}\quad
	\underline\rho_{n\mathfrak m}>z_0^n.
\end{split}
\end{equation}
For each proper resident face, one missing species and one strict sign must work for every tuple in the resident outer set. The missing species may be different on different faces. Both signs are useful: a collision with the face would force zero growth in every vanishing species, so a uniformly positive or a uniformly negative growth rate in just one of them rules out the collision. This is why the condition does not require every missing species to be able to invade, and why it is a continuation condition rather than a stability criterion for the resident equilibrium.

The empty face requires no additional condition because $z_0^n>0$ for every $n\in\mathfrak n$. We retain both the continuation subset and its generic part, used in the first stage of the proof:
\begin{align}
	\mathcal Q_1
	&:=\left\{q\in\mathcal Q_+\ \middle|\
	(\mathrm{FNB}_{\mathfrak n})\text{ holds for every }
	\mathfrak n\in\mathfrak C(q)\right\},
	\label{eq:parameter_subset_Q1}\\
	\widetilde{\mathcal Q}_1
	&:=\mathcal Q_1\cap\mathcal G_+,
	\label{eq:parameter_subset_Q1_generic}
\end{align}
where $\mathcal G_+$ is the generic supercritical set in
(\ref{eq:proof_general_generic_supercritical}). Thus $q\in\mathcal Q_1$ means that every coalition selected by the bifurcation calculation passes its face-wise no-boundary test, while membership in $\widetilde{\mathcal Q}_1$ additionally excludes internal threshold points and the local degeneracies listed in Properties \ref{property:W_principal}--\ref{property:positive_bifurcation}. Both stages of the proof apply coalition by coalition: one coalition satisfying $(\mathrm{FNB}_{\mathfrak n})$ can be certified even if another coalition fails its test and hence $q\notin\mathcal Q_1$. 

We next prove that $\mathcal Q_1$ is semi-algebraic and record some generic approximation facts for later use. Work first on a fixed parameter stratum of Appendix \ref{app:genericity}. The entries of $F(0)$ and of the operators $R_{nm}$ are rational functions of the model parameters wherever the inverse entering $K$ exists. The simple Perron Eigenvalues and the normalized left and right Perron Eigenvectors are algebraic over this rational function field: they are characterized by their left- and right-Eigenvector equations, positivity, and the normalizations in Appendix \ref{app:general_spectrum}. Hence they may be included as auxiliary variables in a finite system of polynomial equalities and inequalities after denominators have been cleared.
All remaining conditions are first-order conditions over the real numbers. In particular, $\Delta_m$ and $\mathcal A_\mathfrak m$ are defined by the polynomial conditions (\ref{eq:proof_general_block_simplex}) and (\ref{eq:proof_general_resident_set}); emptiness of $\mathcal A_\mathfrak m$ is the negation of an existential condition. By (\ref{eq:proof_general_finite_self_test}), $\underline c_m>0$ is equivalent to finitely many strict polynomial inequalities. When $\mathcal A_\mathfrak m\neq\varnothing$, the alternatives in (\ref{eq:proof_general_FNB}) can be written without optimization notation as
\begin{align*}
	&\left[\sum_{m\in\mathfrak m}a_mR_{nm}(u_m)\varpi_n\right]^i
	<z_0^n\varpi_n^i
	&&\text{for every }(a,u)\in\mathcal A_\mathfrak m
	\text{ and }i\in\mathfrak I_n,\\
	\text{or}\qquad
	&\left[\sum_{m\in\mathfrak m}a_mR_{nm}(u_m)\varpi_n\right]^i
	>z_0^n\varpi_n^i
	&&\text{for every }(a,u)\in\mathcal A_\mathfrak m
	\text{ and }i\in\mathfrak I_n.
\end{align*}
These are precisely $\bar\rho_{n\mathfrak m}<z_0^n$ and $\underline\rho_{n\mathfrak m}>z_0^n$, respectively, because the extrema in (\ref{eq:proof_general_rho_upper})--(\ref{eq:proof_general_rho_lower}) are attained. The Tarski--Seidenberg theorem now shows that each $(\mathrm{FNB}_{\mathfrak n})$, and therefore the finite intersection and union defining $\mathcal Q_1$, is semi-algebraic on the stratum. Since there are only finitely many parameter strata and coalitions, $\mathcal Q_1\subset\mathcal Q_+$ is semi-algebraic.
The set $\mathcal G_+$ is relatively open and dense in $\mathcal Q_+$ in the parameter-space sense established in Appendix
\ref{app:general_local_proof}. Lemma \ref{lem:general_FNB_persistence} below will show that, for a fixed coalition satisfying $(\mathrm{FNB}_{\mathfrak n})$, both coalition membership and its certificate persist in a relative parameter neighborhood. It follows in particular that the certified part of any generic coexistence chamber is relatively open. For the parameter-limit proof we will use density in $\mathcal Q_+$ and this coalition-specific neighborhood; no density assertion about $\mathcal Q_1\cap\mathcal G_+$ in all of $\mathcal Q_1$ is needed.

All quantities in (\ref{eq:proof_general_FNB}) are determined by the model parameters and by the Perron data already calculated at the bifurcation. Indeed, when $\mathfrak m\in\mathfrak C(q)$, substituting $u_m=\varpi_m$ into (\ref{eq:proof_general_resident_balance}) reduces it to the principal system $[W]^\mathfrak m_\mathfrak m\zeta_\mathfrak m=(z_0^m)_{m\in\mathfrak m}$. Thus $W$ supplies both the candidate resident amplitudes at $s=0$ and the local transverse signs, while $\mathcal A_\mathfrak m$ accounts for all normalized resident compositions allowed by the projected equations away from the bifurcation.

The separate bounds above give a cheap sufficient test before optimizing over $\mathcal A_\mathfrak m$. Equations (\ref{eq:cross_competition}) and (\ref{eq:proof_general_amplitude_bound2}) imply
\begin{equation}\label{eq:proof_general_rho_gamma_bound}
	\bar\rho_{n\mathfrak m}\leq
	\sum_{m\in\mathfrak m}\gamma_{nm}z_0^m
	\leq\sum_{m\in\mathfrak m}\widehat\gamma_{nm}z_0^m.
\end{equation}
Consequently, a face $\mathfrak m\subsetneq\mathfrak n$ is certified whenever, for at least one $n\in\mathfrak n\setminus\mathfrak m$,
\begin{equation}\label{eq:proof_general_coarse_FNB_certificate}
	z_0^n>\sum_{m\in\mathfrak m}\widehat\gamma_{nm}z_0^m.
\end{equation}
Failure of this finite test is inconclusive. One may next use the exact $\gamma_{nm}$ and finally solve the constrained optimization problems (\ref{eq:proof_general_rho_upper})--(\ref{eq:proof_general_rho_lower}), which can also certify a uniformly negative transverse Eigenvalue through $\underline\rho_{n\mathfrak m}>z_0^n$. For algebraic input parameters, emptiness of $\mathcal A_\mathfrak m$ and the strict inequalities in (\ref{eq:proof_general_FNB}) are finitely decidable by quantifier elimination. In numerical work they can be treated by certified polynomial or rational optimization.

\emph{Step 2, generic continuation without face collisions:}
The following proposition records the continuation argument in a form which does not require membership in $\mathcal Q_1$.

\begin{prop}\label{prop:general_continuation_no_collision}
Let $q\in\mathcal Q_+$ satisfy Properties \ref{property:W_principal} and \ref{property:positive_bifurcation} after the blocks outside $\mathfrak N_+(q)$ have been discarded, and fix $\mathfrak n\in\mathfrak C(q)$. Assume the following no-face-collision hypothesis: whenever $s_k\in(0,1]$, $s_k\to s\in(0,1]$, and $x_k\in[\bar{\mathcal X}]_\mathfrak n$ are equilibria of the deformed model with every $\mathfrak n$-block positive and $x_k\to x$, every $\mathfrak n$-block of $x$ is again positive. Then, for every sufficiently small $s_0>0$, a compact connected set of equilibria with precisely these positive blocks contains the local equilibrium at $s=s_0$ and meets every parameter value in $[s_0,1]$. In particular, the coalition is realized by a non-negative equilibrium of the target model at $s=1$.
\end{prop}
\begin{proof}
Proposition \ref{prop:critical_mortality} supplies a unique positive $\mathfrak n$-block equilibrium for every sufficiently small $s>0$. Its face-restricted Jacobian is nonsingular by the implicit-function argument in that proposition. The exhaustion part of the proof is relevant here: it excludes other positive-$\mathfrak n$-block equilibria at the same small $s$, so that the degree of all such equilibria initially equals the non-zero local degree of this branch. This conclusion uses the generic local hypotheses, without imposing the global lower bounds $\underline c_m>0$.

At any limiting physical equilibrium, the macrosite-sum calculation in Appendix \ref{app:general_equilibrium_preliminaries} gives $\one_b^\transpose x<\bar S^b$. It applies also to the deformation, whose additional mortality is non-negative. At that limit every non-zero species block is strictly positive by irreducibility of $F_n(x;s)$. Thus the no-face-collision hypothesis excludes loss of whole blocks, while the model itself excludes capacity saturation and partial coordinate loss within surviving blocks. On each compact interval $[s_0,1]$, compactness of the physical state space then separates all positive-$\mathfrak n$-block equilibria from the relative boundary of $[\bar{\mathcal X}]_\mathfrak n$. Note that by continuation we simply mean a connected set of equilibria meeting every parameter value, not necessarily the graph of a differentiable function of $s$. Fix such an $s_0>0$ and restrict
\[
	h_\mathfrak n(x;s):=[F(x;s)x]_\mathfrak n
\]
to $[\mathcal V]_\mathfrak n$. By compactness of $[\bar{\mathcal X}]_\mathfrak n$ and the no-face-collision hypothesis, the union, over $s\in[s_0,1]$, of the zeros of $h_\mathfrak n(\mathord{\cdot};s)$ lying in the relative interior of $[\bar{\mathcal X}]_\mathfrak n$ is a compact subset of that relative interior. Hence there is a bounded open set $\mathcal U_\mathfrak n$ whose closure lies in the relative interior, which contains all of these full-support zeros and has no zero on its boundary for any $s\in[s_0,1]$. Boundary zeros belonging to proper subcoalitions lie outside $\overline{\mathcal U_\mathfrak n}$ and play no role in the degree calculation. At $s=s_0$ there is exactly one zero in $\mathcal U_\mathfrak n$, and it is non-degenerate. Therefore
\[
	\deg\bigl(h_\mathfrak n(\mathord{\cdot};s_0),\mathcal U_\mathfrak n,0\bigr)=\pm1.
\]
The local degree at a non-degenerate zero is the sign of the determinant of its face-restricted Jacobian, with the same orientation chosen in domain and codomain. Uniqueness at $s_0$ makes this local degree the degree of the whole isolating set. As $s$ varies, zeros may become singular or additional zeros may appear, but their total degree cannot change while the boundary stays free of zeros. Thus the argument does not need a fixed sign for every individual equilibrium Jacobian along the continuation. The solution property and homotopy invariance of Brouwer degree \cite[Sec.~II.1]{Kielhofer2012} give the same non-zero degree at $s=1$, so $h_\mathfrak n(x;1)=0$ has at least one solution in $\mathcal U_\mathfrak n$. More strongly, the finite-dimensional continuation principle \cite[Lem.~4.1, p.~546]{Mawhin2008} applies to $\mathcal U_\mathfrak n\times[s_0,1]$, regarded as an open bounded subset of $[\mathcal V]_\mathfrak n\times[s_0,1]$ with the relative product topology. Its relative boundary is $\partial\mathcal U_\mathfrak n\times[s_0,1]$, on which there are no zeros. The continuation principle therefore yields a compact connected component of the zero set along which $s$ takes every value in $[s_0,1]$. Since the zero at $s=s_0$ is unique, this component contains the local positive branch and reaches $s=1$. Singular Jacobians, folds, branch points, and additional components may occur in its interior and do not affect the argument.

For the resulting sessile equilibrium $x_\mathfrak n$, the propagule equilibrium $y_\mathfrak{n}$ is non-negative due to $C_5y_\mathfrak n=Kx_\mathfrak n$ and has precisely the same positive species blocks as $x_\mathfrak n$ because, for each retained species $n$, the principal $n$-block of $K$ has no zero row (as $z_0^n>0$ and
$P_0K\varpi_n=(C_1+z_0^n\One)\varpi_n>0$
on its $n$-block).
\end{proof}

The hypothesis concerns all positive-$\mathfrak n$-block equilibria at deformation parameters bounded away from zero. This is the uniform form of boundary exclusion needed by the common isolating set in the proof. It can be verified directly, or by a different sufficient condition: neither membership in $\mathcal Q_1$ nor the specific estimates defining $(\mathrm{FNB}_{\mathfrak n})$ enter Proposition \ref{prop:general_continuation_no_collision}.

For present purposes we however record
\begin{lem}\label{lem:general_FNB_no_collision}
For $q\in\mathcal Q_+$ and $\mathfrak n\in\mathfrak C(q)$,
condition $(\mathrm{FNB}_{\mathfrak n})$ implies the no-face-collision hypothesis of Proposition \ref{prop:general_continuation_no_collision}. This implication requires no genericity assumption.
\end{lem}
\begin{proof}
For an equilibrium $x$ with support $\mathfrak m$ and a missing species $n\in\mathfrak n\setminus\mathfrak m$, equations (\ref{eq:proof_general_competition_decomposition}) and (\ref{eq:proof_general_scaled_residents}) give
\begin{equation}\label{eq:proof_general_transverse_sign}
 F_n(x;s)\varpi_n
 =s\left(z_0^n\varpi_n-
 \sum_{m\in\mathfrak m}a_mR_{nm}(u_m)\varpi_n\right).
\end{equation}
If $\bar\rho_{n\mathfrak m}<z_0^n$, the right-hand side is strictly positive; if $\underline\rho_{n\mathfrak m}>z_0^n$, it is strictly negative. Since $F_n(x;s)$ is irreducible with non-negative off-diagonal entries, multiplication by a positive left dominant Eigenvector shows that its dominant Eigenvalue has the same strict sign.

Now suppose that positive-$\mathfrak n$-block equilibria $x_k$ at $s_k\to s>0$ converge to an equilibrium with non-vanishing blocks $\mathfrak m\subsetneq\mathfrak n$. A surviving block is strictly positive, since its matrix $F_m(x;s)$ has a non-zero non-negative nullvector and is irreducible; see Lemma \ref{lem:essentially-non-negative}. On every member of the sequence, each principal $n$-block of $F(x_k;s_k)$, $n\in\mathfrak n$, has dominant Eigenvalue zero and a strictly positive Eigenvector. Continuity of the dominant Eigenvalue therefore gives dominant Eigenvalue zero in every missing block at the limit.

If $\mathfrak m\neq\varnothing$, its scaled amplitudes and normalized compositions belong to $\mathcal A_\mathfrak m$ by (\ref{eq:proof_general_resident_balance}). Condition $(\mathrm{FNB}_{\mathfrak n})$ and (\ref{eq:proof_general_transverse_sign}) give at least one missing block with a strictly positive or strictly negative dominant Eigenvalue, a contradiction. If $\mathfrak m=\varnothing$, every missing block has dominant Eigenvalue $s z_0^n>0$, also a contradiction. Capacity saturation cannot occur, by the macrosite-sum calculation used in the preceding proposition.
\end{proof}

For generic parameters, these two results recover the existing continuation assertion coalition by coalition. In particular they apply to every $q\in\widetilde{\mathcal Q}_1$. They also apply to a single coalition satisfying $(\mathrm{FNB}_{\mathfrak n})$ when $q\notin\mathcal Q_1$ because another coalition fails its test.

\emph{Step 3, removal of the genericity premise:}
The limit argument in Proposition \ref{prop:general_nongeneric_existence} below needs generic approximations which retain the particular coalition being certified. The next lemma supplies this fact without assuming nonsingularity of unrelated principal matrices.

\begin{lem}\label{lem:general_FNB_persistence}
Let $q\in\mathcal Q_+$ and $\mathfrak n\in\mathfrak C(q)$ satisfy $(\mathrm{FNB}_{\mathfrak n})$. Then $[W]^\mathfrak n_\mathfrak n$ is nonsingular and its principal system has a unique non-negative solution, which is strictly positive. Moreover, $\mathfrak n\in\mathfrak C(q')$ and $(\mathrm{FNB}_{\mathfrak n})$ hold for every $q'$ in some relative neighborhood of $q$ in $\mathcal Q$.
\end{lem}
\begin{proof}
First consider any non-negative solution
\begin{equation}\label{eq:proof_general_FNB_nonnegative_system}
 [W]^\mathfrak n_\mathfrak n a=(z_0^n)_{n\in\mathfrak n},
 \qquad a\geq0.
\end{equation}
We show that its positive support cannot be a proper subset $\mathfrak m\subsetneq\mathfrak n$. That support is non-empty, because every $z_0^n$ in the displayed system is positive. The resident rows of the system show that
$(a_m,\varpi_m)_{m\in\mathfrak m}\in\mathcal A_\mathfrak m$.
For every missing species $n\in\mathfrak n\setminus\mathfrak m$,
the remaining rows give
\[
 z_0^n=\sum_{m\in\mathfrak m}a_m[W]^n_m
 =\varpi^n\cdot\left(
 \sum_{m\in\mathfrak m}a_mR_{nm}(\varpi_m)\varpi_n\right).
\]
The last expression is a weighted average of the coordinate ratios in (\ref{eq:proof_general_rho_upper})--(\ref{eq:proof_general_rho_lower}), with positive weights $\varpi_i^n\varpi_n^i$ whose sum is $\varpi^n\cdot\varpi_n=1$. Consequently,
\begin{equation}\label{eq:proof_general_FNB_boundary_contradiction}
 \underline\rho_{n\mathfrak m}\leq z_0^n
 \leq\bar\rho_{n\mathfrak m}
 \qquad(n\in\mathfrak n\setminus\mathfrak m).
\end{equation}
Neither strict alternative in $(\mathrm{FNB}_{\mathfrak n})$
can hold for any missing species. This contradiction proves that every non-negative solution of (\ref{eq:proof_general_FNB_nonnegative_system}) is strictly positive.

The definition of $\mathfrak C(q)$ supplies at least one such positive solution. If $[W]^\mathfrak n_\mathfrak n$ were singular, move this solution along a non-zero real kernel vector, choosing its sign so that at least one coordinate decreases, until the first coordinate reaches zero. The resulting vector is still a non-negative solution, and it is non-zero because the right-hand side of the principal system is non-zero. This contradicts the preceding paragraph. The matrix is therefore nonsingular, and its solution is the unique positive vector $\zeta_\mathfrak n$.

The Perron data and $W$ depend continuously on the parameters under the standing irreducibility assumptions. Nonsingularity and strict positivity therefore persist near $q$, so $\mathfrak n\in\mathfrak C(q')$ persists as well. Every species in this fixed coalition still has $z_0^n(q')>0$; species outside it need not keep the same invasion sign.

It remains to check persistence of the certificate. The finite strict test (\ref{eq:proof_general_finite_self_test}) and the lower bound (\ref{eq:proof_general_c_lower_bound}) keep the self-competition constants for the coalition bounded away from zero near $q$. Positivity and continuity of the left Perron coordinates make the varying normalized simplexes $\Delta_m(q')$ uniformly bounded. Equation (\ref{eq:proof_general_amplitude_bound2}) then makes the resident sets $\mathcal A_\mathfrak m(q')$ uniformly bounded for each proper subcoalition. Their defining equations have a closed graph: if $q_k\to q$ and tuples from $\mathcal A_\mathfrak m(q_k)$ converge, their limit belongs to $\mathcal A_\mathfrak m(q)$.

If $\mathcal A_\mathfrak m(q)$ is empty, non-empty resident sets arbitrarily close to $q$ would supply, by these uniform bounds, a limiting tuple in the empty set. Thus emptiness persists. If the set is non-empty, choose the missing species and the strict sign which certify its face at $q$. Failure of that same strict alternative at arbitrarily close parameters would give a sequence of violating tuples and coordinate indices. A convergent subsequence of tuples and a constant subsequence of indices would violate the strict inequality at $q$. Hence the chosen alternative also persists. There are only finitely many proper faces, so one relative neighborhood preserves all of their certificates simultaneously.
\end{proof}

\begin{prop}\label{prop:general_nongeneric_existence}
Let $q\in\mathcal Q_+$ and $\mathfrak n\in\mathfrak C(q)$ satisfy $(\mathrm{FNB}_{\mathfrak n})$. Then the target model has a non-negative equilibrium whose positive quasi-species blocks are precisely those in $\mathfrak n$, without any genericity assumption on $q$. In particular, for every $q\in\mathcal Q_1$ the zero equilibrium is unstable, $\mathfrak C(q)\neq\varnothing$, and there are at least $|\mathfrak C(q)|$ non-zero non-negative equilibria.
\end{prop}
\begin{proof}
Fix $q$ and $\mathfrak n$ as stated. By Lemma \ref{lem:general_FNB_persistence}, a relative neighborhood of $q$ preserves this coalition and its face-wise certificate. Density of $\mathcal G_+$ in $\mathcal Q_+$, established in Appendix \ref{app:general_local_proof}, gives a sequence $q_k\to q$ in that neighborhood with $q_k\in\mathcal G_+$. At every $q_k$, Lemma \ref{lem:general_FNB_no_collision} and Proposition \ref{prop:general_continuation_no_collision} supply a target equilibrium $x_k$ whose positive blocks are precisely $\mathfrak n$. These approximating parameters need not belong to $\mathcal Q_1$: only the certificate for the fixed coalition is used, and other coalitions may appear or fail their tests nearby.

The sessile capacity bounds are uniform along $q_k\to q$, even if the capacities themselves are varied. After taking a subsequence, $x_k\to x$, with $x$ in the physical state space at $q$. To display parameter dependence in this argument, write $F(x;s;q')$ for the existing deformed matrix evaluated at $q'$. Continuity gives
\[
 F(x;1;q)x=\lim_{k\to\infty}F(x_k;1;q_k)x_k=0.
\]
Every block outside $\mathfrak n$ remains zero. The macrosite-sum argument in Appendix \ref{app:general_equilibrium_preliminaries} gives strict capacity inequalities at the limit. Each surviving species block is therefore strictly positive, by irreducibility of its matrix $F_n(x;1;q)$.

We must exclude loss of whole blocks. For every $n\in\mathfrak n$, the matrix $F_n(x_k;1;q_k)$ has dominant Eigenvalue zero, since it annihilates the strictly positive $n$-block of $x_k$. Continuity of the spectrum gives dominant Eigenvalue zero for every $F_n(x;1;q)$ as well. If $x=0$, this contradicts $z_0^n(q)>0$. Otherwise let $\varnothing\neq\mathfrak m\subset\mathfrak n$ be its positive support. If $\mathfrak m$ were proper, the amplitudes and normalized compositions of $x$ at $s=1$ would belong to $\mathcal A_\mathfrak m(q)$. Condition $(\mathrm{FNB}_{\mathfrak n})$ at the limiting parameter $q$ and (\ref{eq:proof_general_transverse_sign}) would force at least one missing block to have strictly positive or strictly negative dominant Eigenvalue. This contradicts the zero Eigenvalue just obtained. Thus $x$ has precisely the positive $\mathfrak n$-blocks.

The propagule equation then determines $y=C_5^{-1}Kx\geq0$. Each occupied species block of $K$ has no zero row, so $y$ has exactly the same positive blocks as $x$.

Finally, $q\in\mathcal Q_+$ means that at least one $z_0^n>0$. The preliminary spectral argument gives instability of the zero equilibrium, and the positive singleton solution $z_0^n/[W]^n_n$ shows that $\{n\}\in\mathfrak C(q)$. For $q\in\mathcal Q_1$, apply the preceding argument to every coalition in $\mathfrak C(q)$. Different positive block supports give distinct equilibria, proving the stated lower bound.
\end{proof}
\qed

\subsection{Algorithms}\label{app:algorithms}

The numerical approach separates three tasks: computing the locally admissible coalitions underlying Theorem \ref{thm:near-threshold}, checking the sufficient face conditions underlying Theorem \ref{thm:general_case}, and locating representative target equilibria. Coalition membership, near-threshold non-invadability, and the outcome of the face checks are recorded independently. Below, these tasks are subdivided into two algorithms which are implemented in \url{https://gitlab.com/aguirrle/coexistenceexamples.git}, see that repository for more details.

\begin{alg}[Coalition classification and numerical face checks]\label{alg:general}
	Given a propagule disperser model with parameters $q$, proceed as follows:
	\begin{enumerate}[itemsep=3pt]
		\item Form $F(x)$ as in (\ref{eq:sessiles_stationarity}). Compute the block-dominant Eigenvalues $z_0^n$ and positive left and right Perron Eigenvectors $\varpi^n,\varpi_n$ of $F(0)$, normalized as in Appendix \ref{app:general_spectrum}, and form
		\[
		[W]_{n'}^n=\varpi^n\cdot(F(0) - F(\varpi_{n'}))\varpi_n.
		\]
		Only species with $z_0^n>0$ can participate in a locally admissible coalition; values too close to zero are treated as undecided.
		\item For every non-empty candidate coalition $\mathfrak n$, inspect its principal matrix and solve
		\begin{equation}\label{eq:algorithm_cone_system}
			[W]_\mathfrak n^\mathfrak n\zeta_\mathfrak n=(z_0^n)_{n\in\mathfrak n}.
		\end{equation}
		For a numerically nonsingular system on positive-growth species, strict positivity of the solution establishes local admissibility, whereas a negative component excludes the coalition. Singular or poorly conditioned systems and coefficients too close to zero remain undecided.\label{alg:general_last}
		\item For each admissible coalition, evaluate the absent-species transverse coefficients in (\ref{eq:proof_general_invasion_nondegeneracy}). Record near-threshold non-invadability when all are negative, invasion when at least one is positive, and uncertainty near a zero coefficient.
		\item Independently evaluate the face conditions of Appendix \ref{app:general_proof} for each candidate coalition. Compute the finite self-competition tests and bounds (\ref{eq:proof_general_finite_self_test})--(\ref{eq:proof_general_c_lower_bound}), together with the cross bounds (\ref{eq:proof_general_cross_upper_bound}). For each non-empty proper resident face $\mathfrak m\subsetneq\mathfrak n$, use these bounds and the projected resident balances to enclose $\mathcal A_\mathfrak m$ and its transverse competition ratios. Tighten the resident amplitude bounds using the joint balances. A face passes if its outer set is excluded or if, for some missing species, the resulting bounds establish either strict alternative in (\ref{eq:proof_general_FNB}). On singleton resident faces, feasible simplex-vertex tuples also provide witnesses that can exclude both strict alternatives for a missing species.
		\item Report a numerical pass when all self-competition and proper-face checks pass. Report a numerical failure when a self-competition coefficient vanishes or when a proper face has failure witnesses for every missing species; all remaining cases are undecided. Keep these outcomes and their margins or witnesses separate from the coalition classification.
	\end{enumerate}
\end{alg}

Under the generic hypotheses, the principal-system calculation gives the exact local classification in Proposition \ref{prop:critical_mortality}, and hence the classification used in Theorem \ref{thm:near-threshold}. The implementation leaves unresolved degeneracies explicitly undecided. For a locally admissible coalition, the exact face condition implies target existence by Proposition \ref{prop:general_nongeneric_existence}, even without genericity; imposing it on every coalition gives the hypothesis of Theorem \ref{thm:general_case}. A numerical pass supplies evidence for this condition, rather than a proof with controlled roundoff. Face failures neither erase computed coalitions nor prevent attempts to locate their target equilibria.

\begin{alg}[Arclength integration heuristic]\label{alg:generalB}
	Given a numerically locally admissible coalition $\mathfrak n$, seek a corresponding target equilibrium as follows:
	\begin{enumerate}[itemsep=3pt]
		\item Restrict the deformed equilibrium equation to $[\mathcal V]_\mathfrak n$, keeping excluded blocks exactly zero, and write its residual as $H(x,s)=F(x;s)x$. Choose a small $s_0>0$ and initialize with
		\[
			x_0=s_0\sum_{n\in\mathfrak n}(\zeta_\mathfrak n)_n\varpi_n.
		\]
		Refine this seed by damped Newton iteration at fixed $s_0$, reducing $s_0$ if necessary. Require a small residual, positive retained coordinates, and strict capacity inequalities before continuation.
		\item Parametrize the branch by arclength $\ell$ in the augmented state $(x,s)$ and numerically integrate its oriented unit tangent:
		\[
		H_x\frac{\dd x}{\dd\ell}+H_s\frac{\dd s}{\dd\ell}=0,
		\qquad
		\left\|\left(\frac{\dd x}{\dd\ell},\frac{\dd s}{\dd\ell}\right)\right\|_2=1.
		\]
		Choose the initial orientation towards increasing $s$ and preserve its continuity along the branch. The parameter $s$ may reverse at regular folds. Monitor the residual, positivity, capacity inequalities, and rank of the augmented derivative $[H_x\ H_s]$. Stop at the first detected crossing of $s=1$ or when a numerical diagnostic or integration limit prevents further progress.
		\item At a target crossing, apply a final damped Newton refinement at fixed $s=1$ if needed, then reconstruct
		\[
			y_\mathfrak n=C_5^{-1}Kx_\mathfrak n.
		\]
		Accept a target equilibrium only after checking the full stationary residual, the capacity inequalities, and that both components have precisely the positive quasi-species blocks in $\mathfrak n$.
		\item Save the sampled states, arclengths, deformation parameters, residuals, and termination outcome. If the target was not reached and verified, report an incomplete continuation together with its cause.
	\end{enumerate}
\end{alg}

Algorithm \ref{alg:generalB} uses numerical integration, with Newton refinement at the seed and, when needed, at the target. It can traverse regular folds but does not explore outgoing branches at branch points. The connected solution set supplied by the existence proof need not be one regular curve, so success of this heuristic is not guaranteed even when the face condition holds. Conversely, the heuristic may find a target equilibrium when the face test fails or is undecided.

For parameter scans, the principal-system calculation requires blockwise Perron computations and enumeration of candidate coalitions, so it is most convenient for modest numbers of quasi-species. Self-competition bounds and resident-face data can be shared among coalitions at a fixed parameter point and are recomputed when the parameters change. Coalition diagrams and numerical face-check diagrams are reported separately, while arclength integration provides representative equilibria and their diagnostics.

\section{Model parameters for Example \ref{exmp:general}}\label{app:general_exmp}

For full reproducibility, we give an account of all model parameters used in Example \ref{exmp:general}. Recall that there are three superparameters with values in the unit interval, namely \emph{environmental heterogeneity} $s_1$, \emph{trait heterogeneity} $s_2$ and \emph{displacement strength} $s_3$. As described in the main text, these superparameters enter into
\begin{itemize}[topsep=4pt, itemsep=4pt]
	\item Resource availability $f^+_b$ at macrosite $b$ as 
	$$f^+_b(s_1) = 1 + s_1 \tfrac{1}{4}(b_1+b_2)$$
	\item Resource requirement $f^-_{n,a}$ of individuals of quasi-species $n$ with trait-type $a$ as
	$$f^-_{n,a}(s_2) = \tfrac{1}{2} + \tfrac14 s_2(n-1) + \tfrac{1}{16}a$$
\end{itemize}

The model's kinetic parameters are then set up as follows.
Throughout the superparameter cube, resource availability satisfies
$1/2\leq f_b^+\leq3/2$ and resource requirement satisfies
$3/16\leq f_{n,a}^-\leq13/16$. Thus the specified sessile mortality and
propagule production rates are strictly positive even though
$q_{3,0}=0$. The sigmoid probabilities lie between zero and
$s_3\leq1=p_0^i$, so they satisfy the establishment constraints.
\begin{itemize}[leftmargin=*, topsep=4pt, itemsep=4pt]
	\item \emph{Sessile mortality}: We introduce a single additional parameter (all rates are in units of 1 per unit time):
	\begin{itemize}[topsep=4pt, itemsep=3pt]
		\item basal mortality rate $q_{1,0} := 5.0$
	\end{itemize}
	The sessile mortality coefficients are then $c_{1,i} := q_{1,0}$.
	
	\item \emph{Propagule mortality}: We introduce just a single additional parameter
	\begin{itemize}[topsep=4pt, itemsep=3pt]
		\item basal mortality rate $q_{2,0} := 10.0$
	\end{itemize}
	The propagule mortality coefficients are $c_{2,i} := q_{2,0}$.
	
	\item \emph{Drift/Dispersal}: We introduce four additional parameters
		\begin{itemize}[topsep=4pt, itemsep=3pt]
		\item basal propagule production rate $q_{3,0} := 0.0$
		\item resource utilization rate $q_{3,1} := 100.0$
		\item trait-change probability in each available direction $q_{3,2} := 0.1$
		\item dispersal probability $q_{3,3} := 0.1$
		\end{itemize}
	The propagule dispersal coefficients and drift/dispersal success probabilities are
	\begin{align*}
		c_{3,i} &:= q_{3,0} + q_{3,1} f^-_{N(i),A(i)}(s_2) f^+_{B(i)}(s_1)\\
		d_i^{i'}
		&:=
		\delta_{N(i)}^{N(i')} \cdot
		\left\{\begin{aligned}
			1-2q_{3,2}
			& \text{ if } A(i')=A(i)=0,\\
			1-q_{3,2}
			& \text{ if } A(i')=A(i)\neq0,\\
			q_{3,2}
			& \text{ if } |A(i')-A(i)|=1,\\
			0
			& \text{ otherwise}
		\end{aligned}\right\}
		\cdot
		\left\{\begin{aligned}
			1-q_{3,3}
			& \text{ if } B(i')=B(i),\\
			\displaystyle
			\frac{q_{3,3}}{4}\exp(-1)
			& \text{ if } \|B(i')-B(i)\|_1=1,\\
			0
			& \text{ otherwise}.
		\end{aligned}\right\}
	\end{align*}
	The first factor describes one-step trait drift: from the central
	trait-type, drift occurs in either direction with probability
	$q_{3,2}$, whereas each boundary trait-type has only one possible
	one-step drift direction.  The second factor describes dispersal:
	a propagule remains at its original macrosite with probability
	$1-q_{3,3}$, while it attempts to move in each of the four cardinal
	directions with probability $q_{3,3}/4$ and survives the movement with
	probability $\exp(-1)$.  Multiplication of these two factors expresses
	the assumption that trait drift and spatial dispersal act
	independently.  It permits drift without dispersal, dispersal without
	drift, both events, or neither event.
	For every $i$, the trait-transition factor sums to one over all
	trait-types.  Therefore the corresponding column sum of $(d_i^{i'})$ is
	\[
	1-q_{3,3}
	+
	\frac{q_{3,3}}{4}\exp(-1)
	\left|
	\set{b'\in\mathfrak{B}}
	{\|b'-B(i)\|_1=1}
	\right|
	\leq
	1-q_{3,3}+q_{3,3}\exp(-1)
	<1.
	\]
	Thus $(d_i^{i'})$ is column-substochastic.  The missing probability represents
	propagules that fail to survive dispersal or leave the modeled
	metahabitat.
	
	\item \emph{Migration}:  We introduce two additional parameters
	\begin{itemize}[topsep=4pt, itemsep=3pt]
		\item migration rate $q_{4,0} := 0.1$
		\item emigration probability $q_{4,1} := 0.1$
	\end{itemize}
	The propagule migration coefficients and migration success probabilities are
	\begin{align*}
		c_{4,i} &:= q_{4,0}\\
		\tilde{d}_i^{i'} &:= \begin{cases}
			1-q_{4,1} & \text{if } i'=i\\
			\delta_{N(i)}^{N(i')}\delta_{A(i)}^{A(i')}\delta^1_{\|B(i')-B(i)\|_1}\frac{q_{4,1}}{4}\exp(-1) & \text{if } i'\neq i
		\end{cases}
	\end{align*}
	
	\item \emph{Establishment}: We introduce three additional parameters
	\begin{itemize}[topsep=4pt, itemsep=3pt]
		\item establishment rate $q_{5,0} := 10.0$
		\item empty-microsite establishment probabilities $q_{5,1} := 1.0$
		\item displacement sigmoid steepness $q_{5,2} := 10.0$
	\end{itemize}
	The establishment coefficients, empty-microsite establishment probabilities and displacement probabilities are
	\begin{align*}
		c_{5,i} &:= q_{5,0}\\
		p_0^i &:= q_{5,1}\\
		p^i_{i'} &:= \delta_{B(i')}^{B(i)}s_3
		\left(1 + \exp\left(q_{5,2}(f^-_{N(i),A(i)}(s_2)
		- f^-_{N(i'),A(i')}(s_2))\right)\right)^{-1}
	\end{align*}
\end{itemize}

In total we thus parametrized the model through 11 kinetic event-specific parameters plus 3 superparameters affecting several kinetic events. These choices also verify Assumption \ref{ass:block_irreducibility}
throughout the cube.

\end{document}